\documentclass[11pt]{article}

\usepackage{dsfont}

\def\colorful{1}

\usepackage{amsthm,amsfonts,amsmath,amssymb,epsfig,color,float,graphicx,verbatim}
\usepackage{multirow}

\usepackage{algorithm}
\usepackage[noend]{algpseudocode}

\newif\ifhyper\IfFileExists{hyperref.sty}{\hypertrue}{\hyperfalse}
\hypertrue
\ifhyper\usepackage{hyperref}\fi

\usepackage{enumerate}
\usepackage[shortlabels]{enumitem}

\newcommand{\inote}[1]{\footnote{{\bf [[Ilias: {#1}\bf ]] }}}

\usepackage{framed}
\usepackage{nicefrac}

\def\nnewcolor{0}
\ifnum\nnewcolor=1

\fi
\ifnum\nnewcolor=0

\fi

\ifnum\colorful=1

\else

\fi

\newcommand{\tr}{\mathrm{tr}}
\newcommand{\spn}{\mathrm{span}}
\newcommand{\0}{\mathbf{0}}

\newtheorem{theorem}{Theorem}[section]

\newtheorem{lemma}[theorem]{Lemma}
\newtheorem{informal theorem}[theorem]{Theorem (informal statement)}

\newtheorem{proposition}[theorem]{Proposition}

\newtheorem{claim}[theorem]{Claim}
\newtheorem{fact}[theorem]{Fact}

\newtheorem{remark}[theorem]{Remark}

\theoremstyle{definition}
\newtheorem{definition}[theorem]{Definition}
\newcommand{\eqdef}{\stackrel{{\textrm {\footnotesize def}}}{=}}

\newcommand{\R}{\mathbb{R}}
\newcommand{\s}{\mathbb{S}}
\newcommand{\Z}{\mathbb{Z}}

\newcommand{\eps}{\epsilon}

\newcommand{\pr}{\mathbf{Pr}}

\newcommand{\poly}{\mathrm{poly}}
\newcommand{\polylog}{\mathrm{polylog}}

\newcommand{\E}{\mathbf{E}}

\newcommand{\proj}{\mathrm{proj}}
\newcommand{\rank}{\mathrm{rank}}

\newcommand{\Bg}{\mathrm{B}}
\newcommand{\bg}{\mathrm{b}}
\newcommand{\Sm}{\mathrm{S}}
\newcommand{\sm}{\mathrm{s}}

\newcommand{\sgn}{\mathrm{sign}}

\newcommand{\opt}{\mathrm{OPT}}
\newcommand{\D}{\mathcal{D}}

\newcommand{\littlesum}{\mathop{\textstyle \sum}}

\newcommand{\wh}{\widehat}

\newcommand{\nearest}[1]{\left\lceil#1\right\rfloor}

\usepackage{lipsum}

\newcommand\blfootnotea[1]{%
  \begingroup
  \renewcommand\thefootnote{}\footnote{#1}%
  \endgroup
}

\title{A Strongly Polynomial Algorithm for Approximate Forster Transforms
and its Application to Halfspace Learning\blfootnotea{Author names are in randomized order.}}

\author{
Ilias Diakonikolas\thanks{Supported by NSF Medium Award CCF-2107079, NSF Award CCF-1652862 (CAREER), a Sloan Research Fellowship, and a DARPA Learning with Less Labels (LwLL) grant.}\\
University of Wisconsin-Madison\\
{\tt ilias@cs.wisc.edu}\\
\and
Christos Tzamos\thanks{Supported by NSF Award CCF-2008006 and NSF Award CCF-2144298 (CAREER).}\\
University of Wisconsin-Madison\\
{\tt tzamos@cs.wisc.edu}
\and
Daniel M. Kane\thanks{Supported by NSF Medium Award CCF-2107547, NSF Award CCF-1553288 (CAREER), and a grant from CasperLabs.}\\
University of California, San Diego\\
{\tt dakane@cs.ucsd.edu}
}

\begin{document}

\maketitle

\begin{abstract}
The Forster transform is a method of regularizing a dataset 
by placing it in {\em radial isotropic position}
while maintaining some of its essential properties.
Forster transforms have played a key role in a diverse range of settings
spanning computer science and functional analysis. Prior work had given
{\em weakly} polynomial time algorithms for computing Forster transforms, when they exist.
Our main result is the first {\em strongly polynomial time} algorithm to
compute an approximate Forster transform of a given dataset 
or certify that no such transformation exists. By leveraging our strongly polynomial Forster algorithm, 
we obtain the first strongly polynomial time algorithm for {\em distribution-free} PAC learning of halfspaces.
This learning result is surprising because {\em proper} PAC learning of halfspaces 
is {\em equivalent} to linear programming.
Our learning approach extends to give a strongly polynomial halfspace learner
in the presence of random classification noise and, more generally, Massart noise. 
\end{abstract}

\setcounter{page}{0}

\thispagestyle{empty}

\newpage

\section{Introduction} \label{sec:intro}

\subsection{Forster Transforms and Their Applications} \label{ssec:forster-intro}

The Forster transform is a method of regularizing a dataset $X$ 
(in particular, by placing it in {\em radial} isotropic position) 
while maintaining some of its essential properties.
Forster transforms have been an essential tool in a diverse range of settings, 
including functional analysis~\cite{Barthe98, GargGOW17},
communication complexity~\cite{Forster02}, coding theory~\cite{DvirSW17}, 
mixed determinant/volume approximation~\cite{GurSam02},
learning theory~\cite{HardtM13, HopkinsKLM20, DKT21, DPT21}
and the Paulsen problem in frame theory~\cite{KwokLLR18, HamiltonM19}.
The reader is referred to~\cite{AKS20} for a more detailed discussion.

Known algorithms for computing (approximate) Forster transforms~\cite{HardtM13, AKS20, DKT21} 
rely on black-box convex optimization (e.g., the ellipsoid algorithm)
and consequently have {\em weakly} polynomial runtimes. 
Here we study the question of whether Forster transforms 
can be computed in {\em strongly} polynomial time.
We then leverage Forster transforms for the problem of PAC learning halfspaces 
(both in the realizable setting and in the presence of semi-random label noise).


Intuitively speaking, a Forster transform is a mapping that turns a dataset
into one with good {\em anti-concentration} properties.
Specifically, given a dataset $X \subset \R^d_{\ast}$\footnote{We use 
$\R_{\ast}$ to denote the set $\R \setminus \{ \mathbf{0}\}$.},
a Forster transform of $X$ is an invertible linear transformation $A \in \R^{d \times d}$
such that the set of points $Y = \{Ax/\|Ax\|_2, x \in X \}$
is in isotropic position (i.e., has identity second moment matrix). 
Formally, we have the following more general definition 
allowing for {\em approximate} isotropic position.

\begin{definition}[Approximate Forster Transform] \label{def:forster}
Let $X$ be a set of $n$ nonzero points in $\R^d$ 
and $0 \leq \eps  \leq 1$ be an error parameter.
An {\em $\eps$-approximate Forster transform} of $X$ is an invertible linear transformation
$A \in \R^{d \times d}$ such that, considering the mapping $f_A: \R^d_{\ast}  \mapsto \s^d$
defined by $f_A(x) \eqdef Ax/\|Ax\|_2$, the matrix
$M_A(X) \eqdef (1/n) \sum_{x \in X} f_A(x) f_A(x)^\top$ satisfies
$\frac {1-\eps}{d} \, I \preceq M_A(X)  \preceq \frac{1+\eps}{d} I$.
\end{definition}

\noindent An {\em exact} Forster transform (corresponding to $\eps=0$ in Definition~\ref{def:forster})
aims to linearly transform a given dataset so that
the normalizations of these points are in isotropic position. This notion is
known as ``Forster's isotropic position'' or ``radial isotropic position'' and
can be viewed as an outlier-robust analogue of isotropic position. 
As already mentioned, radial isotropy has been extensively studied 
in functional analysis and computer science.

\begin{remark}
{\em At a high-level, a Forster transform aims to transform a given dataset so that it becomes
``well-conditioned'' in a well-defined technical sense. We note that several other such transformations
have been studied in the literature, including the ``outlier-removal technique'' 
of Dunagan and Vempala~\cite{DV:04} (improving on~\cite{BlumFKV96}) 
and the rescaling method of Dunagan and Vempala~\cite{DunaganV04} for linear programming. 
We provide a summary of these techniques and 
a comparison to radial isotropy in Section~\ref{sec:related}}.
\end{remark}

\vspace{-0.2cm}

\paragraph{Existence}
Forster \cite{Forster02} showed that if the set of points $X$
is in general position, 
then a Forster transform exists. Interestingly, generalizations of Forster's theorem
appear implicitly in~\cite{Barthe98} and explicitly in~\cite{GurSam02}.
We note that there are datasets for which a Forster transform does not exist.
For example, if there is a $d/3$-dimensional subspace
that contains half of the points in $X$,
then after applying {\em any} such transformation to our dataset,
this will still be the case; thus, there will be a $d/3$-dimensional subspace
over which the trace of the second moment matrix is at least $1/2$.
In a recent refinement of the aforementioned works, 
\cite{HopkinsKLM20} showed
that this is the only thing that can go wrong. That is, a Forster transform 
of a given dataset $X$ exists unless there is a $k$-dimensional subspace, 
for some $0<k<d$, containing at least 
a $k/d$-fraction of the points in $X$.

\vspace{-0.2cm}

\paragraph{Efficient Computability}
Forster's existence proof proceeds via a non-constructive iterative argument.
By analyzing a convex program proposed by~ Barthe~\cite{Barthe98},
Hardt and Moitra~\cite{HardtM13} (see also~\cite{AKS20})
showed that the ellipsoid method yields a {\em weakly polynomial} time algorithm
to compute an approximate Forster transform (when it exists).
(More recently,~\cite{DKT21} pointed out that a simple 
explicit SDP can be used to obtain a similar guarantee.)
We remind the reader that the term {\em weakly polynomial time} algorithm 
refers to the fact that the {\em number of arithmetic operations} 
performed by the algorithm scales 
polynomially with {\em the bit complexity of the numbers in the input}.
Specifically, in our Forster setting, the number of arithmetic operations 
required by the ellipsoid method is $\poly(n, d, b, \log(1/\eps))$,
where $\eps$ is the accuracy parameter of Definition~\ref{def:forster}, 
$n$ is the size of the dataset $X$, and $b$ is the bit complexity of $X$.

Starting from the convex programming formulation in~\cite{Barthe98},
Artstein-Avidan, Kaplan, and Sharir~\cite{AKS20} gave an SVD-based gradient-descent
method for computing approximate Forster transforms. 
This method incurs a $\poly(1/\eps)$ runtime dependence and is still weakly polynomial, 
i.e., the number of arithmetic operations scales polynomially in the bit complexity $b$.
Finally, it is interesting to remark that Forster's rescaling
is a special case of operator scaling and tensor scaling (see~\cite{GargO18} for a survey).
Efficient algorithms have been developed
for these more general tasks, see, e.g.,~\cite{AGLOW18, BFGOWW18},
albeit with weakly polynomial guarantees.

\vspace{-0.2cm}

\paragraph{Weakly versus Strongly Polynomial Time}
As is standard for computational purposes, we assume that every integer or rational number
appearing in the input is encoded
using its binary representation. 
Let $N \in \Z_+$ denote the number of integer numbers
given as input and $b \in \Z_+$ denote the bit complexity 
of the largest integer appearing in the input description.
An algorithm for the underlying computational problem is called {\em weakly polynomial}, 
if its worst-case running time is bounded by a fixed-degree polynomial in the Turing
machine model of computation.

The concept of {\em strongly polynomial} time 
was introduced by Megiddo~\cite{Meg83}, under the name ``genuinely polynomial''.
A strongly polynomial time algorithm satisfies the following properties
(see, e.g., Section~1.3 of~\cite{GLS:88}):
(i) it uses only elementary arithmetic operations
(specifically, integer addition, subtraction, multiplication, and division), 
(ii) the number of arithmetic operations is bounded above by a polynomial in $N$,
and (iii) the algorithm is a polynomial space algorithm: that is, all numbers appearing
in all intermediate computations are rational numbers with bit complexity bounded
above by a polynomial in the input size (i.e., $\poly(N, b)$).

The key difference between strongly and weakly polynomial time lies in property (ii) above.
In a weakly polynomial algorithm, the {\em number of arithmetic operations} is allowed to 
scale with the bit complexity of the numbers in the input. {\em In sharp contrast, 
in a strongly polynomial time algorithm no bit complexity dependence is allowed.}

\vspace{-0.2cm}

\paragraph{Forster Transforms in {\em Strongly} Polynomial Time?}

Motivated by the fundamental nature and the
varied applications of Forster transforms, here we ask
the following question:
\begin{center}
{\em Is there a {\em strongly} polynomial time algorithm to compute \\ 
an approximate Forster transform of a given dataset (assuming one exists)?}
\end{center}
Our main algorithmic result (Theorem~\ref{thm:forster-intro}) 
answers this question in the affirmative by giving
the first randomized strongly polynomial-time algorithm
for computing approximate Forster transforms ---
corresponding to $\eps = \Omega(\poly(1/(n, d)))$ in Definition~\ref{def:forster}.
Importantly, a constant value of $\eps$ suffices for our learning theory application 
to learning halfspaces. Obtaining a strongly polynomial time algorithm 
for inverse exponential values of $\eps$ is left 
as an interesting open problem (see Section~\ref{sec:conc} for a discussion).

\subsection{Halfspaces and Efficient PAC Learnability} \label{ssec:ltfs-intro}

One of the main motivations behind this work was leveraging Forster transforms
as a tool for the algorithmic problem of distribution-free PAC learning of halfspaces. 
We review the relevant background in the subsequent discussion.

\paragraph{Halfspaces} 
We are concerned with the efficient learnability
of halfspaces in Valiant's distribution-free
PAC model~\cite{val84}. A {\em halfspace} or Linear Threshold Function (LTF)
is any Boolean-valued function $f: \R^d \mapsto \{\pm 1\}$ of the form
$f(x) = \sgn(w \cdot x - t)$, for some $w \in\R^d$ (known as the weight vector)
and $t \in \R$ (known as the  threshold).
(The function $\sgn: \R \mapsto \{ \pm 1\}$ is defined as $\sgn(u) = 1$ if $u \geq 0$,
and $\sgn(u) = -1$ otherwise.)
Halfspaces are one of the most extensively studied classes of Boolean functions
due to their central role in several areas, including
complexity theory, learning theory, and optimization~\cite{Rosenblatt:58, Novikoff:62,
MinskyPapert:68, Yao:90, GHR:92, FreundSchapire:97, Vapnik:98, CristianiniShaweTaylor:00,
ODonnellbook}. 

\vspace{-0.2cm}

\paragraph{Background on PAC Learning}

The major goal of computational learning theory 
is to develop learning algorithms for expressive
concept classes that are both statistically and computationally efficient.
To facilitate the subsequent discussion, we formally define Valiant's PAC model.

\begin{definition}[PAC Learning] \label{def:PAC-learning}
Let $\mathcal{C}$ be a class of Boolean-valued functions over $X= \R^d$ and
$\D_{X}$ be a fixed but unknown distribution over $X$.
Let $f$ be an unknown target function in $\mathcal{C}$.
A {\em PAC example oracle}, $\mathrm{EX}(f, \D_{X})$,
works as follows: Each time $\mathrm{EX}(f, \D_{X})$ is invoked,
it returns a labeled example $(x, y)$, where $x \sim \D_{X}$ and
$y = f(x)$. Let $\D$ denote the joint distribution on $(x, y)$ generated by the above oracle.
Given an accuracy parameter $\gamma>0$ and access to
i.i.d.\ samples from $\D$, the learner wants to output a hypothesis
$h: \R^d \mapsto \{\pm 1\}$ such that with high probability
the misclassification error of $h$ is at most $\gamma$, i.e.,
we have that $\pr_{(x, y) \sim \D} [h(x) \neq y] \leq \gamma$.
\end{definition}

The hypothesis $h$ in Definition~\ref{def:PAC-learning} 
does not necessarily belong to the class $\mathcal{C}$. 
Namely, we focus on the standard notion of {\em improper} learning, where 
the learner can output any efficiently computable hypothesis.
The special case where $h$ is required to lie in $\mathcal{C}$ is known
as {\em proper} learning. While proper learning might be desirable for some applications  
(e.g., due to its interpretability), there exist natural concept classes 
for which proper learning is computationally hard and improper learning
is easy (see, e.g.,~\cite{KearnsVazirani:94}).
An improper hypothesis 
is as useful as a proper one for the purpose of predicting new function values.

\begin{remark}
{\em The PAC model of Definition~\ref{def:PAC-learning} is known as {\em realizable} 
because of the assumption that the labels are consistent with the target concept. While our main 
learning application is on the realizable learning of halfspaces in strongly polynomial time 
(Theorem~\ref{thm:lft-real-intro}), our positive result extends for learning halfspaces 
in the presence of random or semi-random label noise (Theorem~\ref{thm:lft-Massart-intro}).}
\end{remark}

\vspace{-0.2cm}

\paragraph{PAC Learning Halfspaces and Linear Programming}
With this terminology, we return to our discussion on halfspaces.
Suppose we are given a multiset of $n$ labeled examples, $(x^{(i)}, y^{(i)})$, with
$x^{(i)} \sim \D_{X}$  
and $y^{(i)} = f^{\ast} (x^{(i)})$,
where $f^{\ast}(x) = \sgn(w^{\ast} \cdot x - t^{\ast})$
is the target halfspace. Then we can find a consistent halfspace hypothesis
$h(x) = \sgn(\wh{w} \cdot x - \wh{t} \, )$ (i.e., a halfspace that agrees with the training set)
via a reduction to Linear Programming (LP); see, e.g.,~\cite{MT:94}.
Indeed, each example  $(x^{(i)}, y^{(i)})$ gives rise
to the linear inequality $(w \cdot x^{(i)}  - t) y^{(i)} \geq 0$ over variables $(w, t) \in \R^{d+1}$.
This gives us an LP with $d+1$ variables and $n$ constraints, which is feasible
(as $(w^{\ast}, t^{\ast})$ is a feasible solution by assumption).
We can thus use any polynomial-time LP algorithm to compute a feasible solution
$(\wh{w}, \wh{t} \,)$.
By standard VC-dimension generalization results (see, e.g.,~\cite{KearnsVazirani:94}),
if the sample size $n$ is sufficiently large, namely for some $n  = \tilde{O}(d/\gamma)$,
the halfspace hypothesis $h(x) = \sgn (\wh{w} \cdot x - \wh{t}\, )$ with high probability satisfies
$\pr_{(x, y) \sim \D} [h(x) \neq y] \leq \gamma$. This straightforward reduction gives
a PAC learning algorithm for halfspaces on $\R^d$
with sample complexity $\tilde{O}(d/\gamma)$ and running time polynomial in the input size.
Formally speaking, the running time of such an algorithm is {\em weakly polynomial}, i.e.,
its worst-case number of arithmetic operations scales with the bit complexity of the input examples.

Interestingly, the aforementioned reduction can be reversed.
That is, one can use any PAC learner that outputs a halfspace hypothesis as a black-box
to solve the linear feasibility problem $A w \geq 0$, $w \neq 0$, where $A \in \R^{n \times d}$
and $w \in \R^d$, by considering each linear constraint as an example.
Intuitively, the vector $w$ can be viewed as the weight vector defining the target halfspace.

\vspace{-0.2cm}

\paragraph{Learning Halfspaces in {\em Strongly} Polynomial Time?}
All known polynomial time algorithms for LP,
including the ellipsoid algorithm and interior-point methods, are weakly polynomial.
The existence of a strongly polynomial LP algorithm
is a major open question in computer science, famously highlighted
by Smale~\cite{Smale98}.
The straightforward reduction of PAC learning halfspaces
to LP leads to a {\em weakly} polynomial learner.
Interestingly, the reduction in the opposite direction
has lead various authors (see \cite{Cohen:97} and recently~\cite{DGT19, ChenKMY20})
to suggest that learning halfspaces in strongly polynomial time
is {\em equivalent} to strongly polynomial LP. 
The catch, of course, is that this equivalence only holds if we restrict ourselves to {\em proper} learners.

Several weakly polynomial time algorithms for 
PAC learning halfspaces have been developed over the past thirty years,
starting with the pioneering works~\cite{BlumFKV96, Cohen:97, DunaganV04}
and recently in~\cite{DGT19, ChenKMY20, DKT21}. 
(These works do not proceed 
by a black-box reduction to solving LPs.)
These learners succeed 
not only in the realizable setting, 
but also in the presence of (semi)-random label noise.
Importantly, all prior learners are weakly polynomial --- 
even restricted to the realizable setting.
This discussion serves as a motivation for the following question:
\begin{center}
{\em Is there a {\em strongly} polynomial time algorithm for PAC learning halfspaces?}
\end{center}
The main learning-theoretic result of this paper (Theorem~\ref{thm:lft-real-intro})
answers the above question in the affirmative.
This algorithmic result generalizes to yield strongly polynomial time algorithms for
learning halfspaces in ``benign'' noise models,
including Random Classification Noise (RCN)~\cite{AL88}
and, more generally, Massart noise~\cite{Massart2006} (Theorem~\ref{thm:lft-Massart-intro}).

\subsection{Our Results} \label{sec:results}

The main algorithmic result of this work is
the first randomized strongly polynomial time algorithm for computing
an approximate Forster transform of a given dataset, assuming that one exists.

\begin{theorem}[Approximate Forster Transforms in Strongly Polynomial Time] \label{thm:forster-intro}
There exists a randomized algorithm that given a set $X \subset \R_{\ast}^d$
of size $n$ and a parameter $\epsilon \in (0, 1)$,
runs in time strongly polynomial in $n d/\eps$, and
has the following high probability guarantee: 
either the algorithm computes an $\eps$-approximate Forster transform of $X$,
or it correctly detects that no Forster transform of $X$ exists by finding a proper subspace
$W \subset \R^d$ such that $|X \cap W| > (n/d) \, \dim(W)$.
\end{theorem}

In more detail, the algorithm of Theorem~\ref{thm:forster-intro} performs $\poly(n,d,1/\epsilon)$
arithmetic operations on $\poly(n,d, 1/\epsilon,b)$-bit numbers,
where $b$ is the bit complexity of the points in $X$. As discussed in the introduction,
previous algorithms for this problem rely on the ellipsoid method and therefore are
weakly polynomial even for constant values of $\eps$.
The running time of our algorithm has a polynomial dependence in $1/\eps$;
hence, our algorithm does not run in polynomial time 
when $\eps$ is inverse super-polynomially small in $n, d$. 
Importantly, for our application in halfspace learning
(and several other applications of Forster transforms) 
constant values of the parameter $\eps$ suffice.


By using the algorithm of Theorem~\ref{thm:forster-intro} as a black-box (for $\eps = 1/2$),
we establish our main learning result
(see Theorem~\ref{thm:lft-real-full} for a more detailed statement).

\begin{theorem}[PAC Learning Halfspaces in Strongly Polynomial Time]\label{thm:lft-real-intro}
Let $\D$ be a distribution over labeled examples $(x, y) \in \R^d \times \{\pm 1\}$
such that the distribution over examples is arbitrary
and the label $y$ of example $x$ satisfies $y = f(x)$,
for an unknown halfspace $f:\R^d \mapsto \{\pm 1\}$.
There is an algorithm that, given $\gamma>0$, 
draws $n = \poly(d/\gamma)$ i.i.d.\ samples from $\D$,
runs in strongly polynomial time,
and returns a strongly polynomial time computable hypothesis
$h:\R^d \mapsto \{\pm 1\}$ such that with high probability we have that
$\pr_{(x,y)\sim \D}[h(x) \neq y] \leq \gamma$.
\end{theorem}

Given the {\em equivalence} of {\em proper} halfspace learning 
and LP, we view this algorithmic result as fairly surprising.
Theorem~\ref{thm:lft-real-intro} gives the first strongly polynomial time PAC learning algorithm
for halfspaces. In more detail, if $b$ is the bit complexity of the examples
(i.e., the maximum number of bits required to represent each coordinate of each example vector),
our algorithm uses $\poly(n)$ arithmetic operations on $\poly(n ,b)$-bit numbers.
Finally, we note that the hypothesis $h$ computed by our algorithm is a decision-list
of $\poly(d/\gamma)$ many halfspaces. Importantly, for each point $x$, the value $h(x)$
is computable in strongly polynomial time (in $n$).

\begin{remark}
{\em The list of concept classes for which efficient learners 
have been developed in Valiant's distribution-free PAC model is fairly short. 
The class of halfspaces is of central importance in this list. Specifically, 
a strongly polynomial algorithm for PAC learning halfspaces immediately implies (via the kernel trick) 
strongly polynomial learners for broader concept classes, including
degree-$k$ polynomial threshold functions for any $k = O(1)$
(see, e.g.,~\cite{bluehrhauwar89})}.
\end{remark}

It is worth pointing out that the idea of using Forster transforms 
for halfspace learning was recently used in~\cite{DKT21} 
for the problem of PAC learning with Massart noise.
In the Massart model~\cite{Massart2006}, an adversary
independently flips the label of each point $x$
with unknown probability $\eta(x) \leq \eta<1/2$.
The learner of~\cite{DKT21} used a weakly polynomial Forster transform routine.
By instead using our algorithm of Theorem~\ref{thm:forster-intro},
we obtain the following generalization of Theorem~\ref{thm:lft-real-intro}.

\begin{theorem}[PAC Learning Massart Halfspaces in Strongly Polynomial Time] \label{thm:lft-Massart-intro}
Let $\D$ be a distribution over labeled examples $(x, y) \in \R^d \times \{\pm 1\}$
such that the distribution over examples is arbitrary
and the label $y$ of example $x$ satisfies (i) $y = f(x)$ with probability $1-\eta(x)$,
and (ii) $y = -f(x)$ with probability $\eta(x)$,
for an unknown halfspace $f:\R^d \mapsto \{\pm 1\}$.
Here $\eta(x)$ is an unknown function that satisfies $\eta(x) \leq \eta<1/2$ for all $x$.
There is an algorithm that, given $\gamma>0$, draws 
$n = \poly(d/\gamma)$ i.i.d.\ samples from $\D$,
runs in strongly polynomial time, and returns a strongly polynomial time computable hypothesis
$h:\R^d \mapsto \{\pm 1\}$ such that with high probability we have that
$\pr_{(x,y)\sim \D}[h(x) \neq y] \leq \eta+\gamma$.
\end{theorem}


Theorem~\ref{thm:lft-Massart-intro} generalizes Theorem~\ref{thm:lft-real-intro}
(which corresponds to the case of $\eta=0$).
For the special case of uniform noise (i.e., when $\eta(x) = \eta <1/2$ for all $x$) 
--- this is known as Random Classification Noise~\cite{AL88} --- Theorem~\ref{thm:lft-Massart-intro}
achieves the information-theoretically optimal error and runs in {\em strongly} polynomial time.
It thus qualitatively improves on the classical work of~\cite{BlumFKV96} who gave a weakly
polynomial time algorithm with the same error guarantee.

Theorem~\ref{thm:lft-Massart-intro} similarly improves prior work on learning halfspaces
with Massart noise.
Prior algorithms for learning Massart halfspaces have weakly polynomial runtimes 
and achieve the same error as Theorem~\ref{thm:lft-Massart-intro},
which is believed to be the computational limit for the problem.
In more detail, 
the first (weakly) polynomial learner for Massart  halfspaces 
was given in~\cite{DGT19} and
achieves error $\eta+\gamma$, as our 
Theorem~\ref{thm:lft-Massart-intro}. 
While this error guarantee is not information-theoretically 
optimal in the Massart model (the optimal error is $\opt = \E_x [\eta(x)]$),
there exists strong evidence~\cite{DK20-hardness, Nasser22opt, DKMR22} that 
the bound of $\eta$ cannot be improved by any polynomial time algorithm.
Finally, we note that subsequent work to~\cite{DGT19}
gave a proper learner for Massart halfspaces~\cite{ChenKMY20}, 
which is inherently weakly polynomial.

\subsection{Our Techniques} \label{ssec:techniques}

\subsubsection{Strongly Polynomial Approximate Forster Transform}

\vspace{-0.2cm}

\paragraph{Overview of Algorithmic Approach}
Letting $f_A(x) \eqdef Ax/\|Ax\|_2$,
given a dataset $X$ of $n$ points in $\R^d_{\ast}$,
our goal is to efficiently compute an invertible
linear transformation $A \in \R^{d \times d}$
such that the matrix $M_A(X) \eqdef (1/n) \sum_{x \in X}f_A(x)f_A(x)^{\top}$
is approximately equal to $(1/d) \, I$; in particular,
we would like it to have eigenvalues
in $[\frac{1-\eps}{d},\frac{1+\eps}{d}]$.
Since the trace of $M_A(X)$, $\tr(M_A(X))$, is always equal to $1$,
this goal is equivalent to finding a matrix $A$ such that
the squared Frobenius norm of $M_A(X)$,
$\|M_A(X) \|_F^2$, is close to $1/d$ (Lemma~\ref{lem:phi-prop}).
This observation gives rise to the natural idea of using
an iterative algorithm to compute such an $A$.
In particular, given a linear transformation $A$
such that $\|M_A(X)\|^2_F$ is somewhat small,
our goal is then to find another
linear transformation $C \in \R^{d \times d}$
such that the corresponding second moment matrix
$M_{C A}(X) = (1/n) \sum_{x \in X} f_{C A}(x)f_{C A}(x)^{\top}$
has squared Frobenius norm, $\| M_{C A}(X) \|_F^2$,
somewhat smaller than $\|M_A(X)\|_F^2$.
Equivalently, since for any point $x \in \R^d_{\ast}$
it holds that $f_{C A}(x) = f_{C}(f_A(x))$, we
consider the set of transformed points
$X_A = f_A(X) \eqdef \{ f_A(x): x \in X \}$
and aim to make the second moment matrix of $f_{C}(X_A)$
smaller than the second moment matrix of $X_A$.
If for any invertible $A$ we can find such a $C$,
then by iteratively replacing $A$ by $C A$
we can achieve smaller and smaller values of $\|M_A(X) \|_F^2$,
until in the limit it approaches $1/d$.

Since $\tr(M_A(X)) = 1$, if $\|M_A(X) \|_F^2$ is bounded away from $1/d$,
some of the eigenvalues of $M_A(X)$
(which average to $1/d$) must differ substantially from $1/d$.
This in turn implies that $M_A(X)$ must have a reasonably-sized \emph{eigenvalue gap}.
In particular, this means that there 
exist subspaces $V$ and $V^{\perp}$,
that are each spanned by eigenvectors of $M_A(X)$,
such that the eigenvalues of $V^{\perp}$
exceed the eigenvalues on $V$ by at least some reasonably large $\delta >0$.
Roughly speaking, if we can find a matrix $C$ that decreases
the squared Frobenius norm of $M_A(X)$ on $V^{\perp} \times V^{\perp}$
and increases the squared Frobenius norm on $V \times V$,
this will improve the desired squared Frobenius norm.

A natural approach to achieve this goal
is to let $C$ be equal to $I_{V^\perp} + (1+\alpha) I_{V}$, the identity on $V^{\perp}$,
and $(1+\alpha)$ times the identity on $V$, for some suitable $\alpha>0$.
It is not hard to see that this choice of $C$ strictly decreases the
second moment matrix on $V^\perp$,
and strictly increases it on $V$.
Unfortunately, it might also create cross-terms
that will increase the Frobenius norm.
To understand the effect of the cross-terms,
it is important to consider how close vectors in $X_A$ are
to being in $V$ or in $V^\perp$.
In particular, let $\beta$ be the maximum distance
that any vector in $X_A$ is from being in either $V$ or $V^\perp$.
If $\alpha = O(1)$,
this moves approximately $\alpha\beta^2$ of the trace of $M_A(X)$
from $V^\perp$ to $V$, which improves (i.e., decreases)
the squared Frobenius norm by roughly
$\alpha \beta^2$
(times some inverse $\poly(dn/\eps)$ factors).
On the other hand, this also creates cross-terms in the order of $\alpha \beta$,
which increases the squared Frobenius norm by a quantity
on the order of $\alpha^2 \beta^2$.
Thus, as long as $\alpha$ is less than $\beta$
times a sufficiently small polynomial in $dn/\eps$,
we obtain an improvement in the squared Frobenius norm
on the order of $\alpha \beta^2/\poly(dn/\eps)$.

This improvement suffices for our purposes,
unless $\beta$ happens to be very small.
The latter occurs if all of the points in $X_A$ are
either very close to $V$ or very close to $V^\perp$.
In such a case, the simple choice of matrix $C$
described in the previous paragraph
may not be sufficient, as it will produce too many cross-terms.
In order to make progress here,
we require a different approach, which we describe next.
To describe our approach for this case, we introduce additional terminology.
We let $X^{\Bg}_A$
be the set of points in $X_A$ that are close to $V^\perp$.
Moreover, let $U$ be the span of the $|V|$ smallest eigenvectors of
the matrix $\sum_{x \in X^{\Bg}_A} x x^{\top}$,
and let $U^\perp$ be the orthogonal subspace.
We now define the new matrix $C$ to be $I_{U^\perp} + (1+\alpha) I_U$,
the identity on $U^\perp$ and some very large multiple $(1+\alpha)$
of the identity on $U$. We claim that this choice
actually does not create much in the way of cross-terms.
In particular, the matrix
$\sum_{x \in X^{\Bg}_A}(C x)(C x)^{\top}
= C^{\top} \, \sum_{x \in X^{\Bg}_A} xx^{\top} \, C$
will have no $U \times U^\perp$ term,
since $\sum_{x \in  X^{\Bg}_A} xx^{\top} $ does not
--- as $U$ is an eigenspace of $\sum_{x \in  X^{\Bg}_A}xx^{\top}$.
The second moment matrix $\sum_{x \in X^{\Bg}_A} f_{C}(x)f_{C}(x)^{\top}$
will have some contribution to $U \times U^\perp$ cross-terms 
coming from the renormalization; but these will only be on the order of $(\alpha \, \beta)^4$.
On the other hand, the matrix
$\sum_{x \in X_A \setminus X^{\Bg}_A} f_{C}(x) f_{C}(x)^{\top}$
will have small $U \times U^\perp$ terms, because each $f_{C}(x)$ will nearly lie in $U^\perp$.
If $\beta$ is sufficiently small, this leads to roughly $(\alpha \, \beta)^2$ mass being moved
from $U^\perp \times U^\perp$ to $U \times U$, while only creating off-diagonal terms
on the order of $(\alpha \, \beta)^4$. Thus, this alternate choice of $C$
can be used to decrease
the squared Frobenius norm by $\poly(\alpha/(dn))$.

The preceding outline provides a procedure that produces a sequence of matrices
$A_1, A_2, \ldots$ such that if $e_i = \| M_{A_i}(X) \|_F^2 -1/d$,
then $e_{i+1} < e_i - \poly(e_i/(dn))$.
Therefore, after polynomially many iterations,
we have that $\|M_{A_m}(X)\|_F^2 < 1/d+(\eps/d)^2$,
which implies we have obtained an $\eps$-approximate Forster transform.
This gives us an efficient algorithm for computing an approximate Forster transform
in the real RAM model, assuming the availability of an algorithm for
exact eigendecomposition computation.

\vspace{-0.2cm}

\paragraph{Additional Technical Obstacles}
The above iterative procedure forms the basis of our final strongly polynomial
time algorithm. Unfortunately, as is, this procedure does not directly
imply a strongly polynomial time algorithm for two reasons:
First, we need to control the bit complexities of the matrices $A_i$ (which
might become exponentially large). Second, we need to show that our algorithm
works with approximate eigendecompositions
(which can further be implemented in strongly polynomial time).
We elaborate on these issues in the following discussion.

\vspace{-0.2cm}

\paragraph{Controlling the Bit Complexity via Rounding}
Recall that, in a strongly polynomial time algorithm, all intermediate numbers
computed throughout the algorithm must fit in polynomial space.
To handle the bit complexity in our setting,
we establish the following statement.
If the points in the initial dataset $X \subset \R_{\ast}^d$ of size $n$
have bit complexity at most $b$, then the following holds: 
given a matrix $A \in \R^{d \times d}$ and any $\delta>0$,
we can approximate $A$ by another matrix $A'$
of bit complexity $\poly(b,d,n,\log(1/\delta))$ such that
$\|M_{A'}(X) \|_F^2 < \|M_A(X)\|_F^2 + \delta$ (see Theorem~\ref{thm:rounding}). 
This structural result suffices for our purposes for the following reason:
Replacing each intermediate matrix $A_i$
(in our iterative procedure)
by the corresponding $A'_i$ obtained by rounding
(for an appropriately small $\delta$) at each step of our algorithm
suffices to keep the bit-complexity under control.

To prove the desired structural result, we proceed as follows:
First, if $A$ has
condition number at most $\exp(\poly(n,b,d))$, it suffices to merely approximate
each entry of $A$ to some $\poly(bdn/\log(1/\delta))$ bits of precision.
The difficulty arises
if the condition number of $A$ is quite large --- in fact,
exponentially large in our other parameters.
If the condition number of $A$ is large,
it is because there are large multiplicative gaps in the singular values of $A$.
In such a case, there will be subspaces $V$ and $V^\perp$ such
that the $V^\perp$-component of any vector is multiplied by a huge
amount relative to the $V$-component. In particular, any vector that was not
exponentially close to $V$ to begin with,
after multiplying by $A$ ends up essentially in $V^\perp$.
Our basic strategy here is to decrease the size of this singular value gap of $A$
to be at most (merely) exponential,
without much affecting any of the normalized transformed vectors.
Our goal is to scale down the subspace $V^\perp$
to decrease the multiplicative eigenvalue gap.
However, we must ensure that the vectors of $X$
that are sufficiently close to $V^\perp$ after applying $A$
do not end up being essentially in $V$.
To achieve this, we
consider a subspace $W$ spanned by such problematic vectors and build
an improved matrix $A T$ such that $T$
does not affect vectors in $W$, but rescales significantly
vectors lying in a subspace $R$ that is very close to $V^\perp$.
Via this step, we can reduce the condition number
of $A$ to be appropriately bounded without affecting
the mapping $f_A$ significantly;
after that, we can make do with a suitably precise
rounding to obtain the output matrix $A'$.

\vspace{-0.2cm}

\paragraph{Approximate Eigendecomposition in Strongly Polynomial Time}
So far, we have assumed the availability of a routine for exact eigendecomposition.
In fact, there are several places in the above
intuitive overview of our algorithmic approach
where we need to compute an eigenvalue decomposition of a matrix.
This is required first when we need to find the initial eigenvalue gap 
in $ M_A(X) \propto \sum_{x \in X_A} xx^\top$,
and again later when we need to find the span of the large eigenvalues of $\sum_{x \in X_A^\Bg} xx^\top$.
Unfortunately, computing exact eigenvalues is impossible in our
model of computation (as doing so might require finding roots of high-degree polynomials).
Fortunately, it is sufficient for us to find merely an \emph{approximate} eigenvalue decomposition
of these matrices. A subtle and important point is that our required notion
of approximation is {\em significantly stronger than the typical
guarantees explicitly available in the literature}.
Interestingly, we show that the desired strongly polynomial guarantees
can be achieved in our model using some variation of the power iteration method.
This requires a novel proof of correctness, that we provide here.

We are now ready to describe our strongly polynomial 
approximate eigendecomposition routine 
in tandem with a sketch of its analysis (see Proposition~\ref{prop:svd}).
The standard power iteration method says that
in order to approximate the principal eigenvector of a symmetric,
PSD matrix $M$, it suffices to multiply a random vector $v$
by a large power $t$ of $M$. If we express $v$
as a linear combination of eigenvectors of $M$,
then multiplying by a large power of $M$
scales each of these components by an amount depending on the eigenvalue.
It is not hard to see that if there is a reasonable gap between
the largest and second largest eigenvalues, then the vector $M^t v$
will likely end up close to a multiple of the largest eigenvector.
Once an approximate principal eigenvector is computed,
one can attempt to repeat the same procedure, i.e.,
projecting onto the orthogonal subspace
to find the second largest eigenvalue; and so on.
This iterative procedure is known to succeed in
finding approximations to the eigenvectors and eigenvalues in question,
so long as the eigenvalues are not too close to each other.
On the other hand, if $M$ has (nearly) degenerate eigenspaces,
then this method may fail to separate eigenvectors with very similar eigenvalues.
However, in this (near-)degenerate case, such an approximation is usually not needed,
as the eigenvalues are close to begin with.
One can hope that the matrix $\hat M$
corresponding to the computed eigendecomposition
is close to $M$ in an appropriate sense.
In particular, standard results (see, e.g.,~\cite{Parlett98})
show how to compute such an $\hat M$ satisfying
$\|M-\hat M\|_2 \leq \eps \|M\|_2$.

{\em Unfortunately, this notion of approximation is not sufficient for our purposes.}
For example, in the case where the parameter $\beta$
is small in our Forster algorithm,
it is important for us to compute the spaces $V'$ and $W'$  
to {\em very good accuracy}.
This is because the linear transformation
that we apply will multiply elements of $V'$
by a large factor of roughly $1/\beta$.
This means that we need to compute $V'$ to error
on the order of $\beta$ in order to ensure
the accuracy of our result.
More generally, we will need a qualitatively stronger guarantee
for our approximate eigenvalue decomposition.
In particular, we need that for some small $\eps>0$,
for any vector $v$, it holds that
$|v^{\top}(M-\hat M)v| \leq \eps (v^{\top} M v)$.
This means that if $v$ lies in a space
spanned by eigenvectors of $M$ with
very small eigenvalues (as $V'$ is above),
then we need that $\hat M v$ to be correspondingly small.
Fortunately, we can obtain this much stronger ``multiplicative''
guarantee via power iteration.
The intuitive reason this works
is essentially because if we have a space $V'$
spanned by eigenvectors of $M$ with eigenvalues at most $\beta$,
then multiplying a random vector $v$ by powers of $M$
reduces the size of the projection of $v$ onto $V'$
by a power of $\beta$. This means that power iteration
produces vectors that are very nearly orthogonal to $V'$
with the error in this approximation scaling with $\beta$.

\subsubsection{Learning Halfspaces in Strongly Polynomial Time}

As already mentioned in the introduction,
we leverage our algorithm for approximate Forster transforms
to obtain the first strongly polynomial algorithm for PAC learning halfspaces.
It turns out that this approach goes through
both in the realizable case (Definition~\ref{def:PAC-learning})
and in the presence of (semi-random) Massart noise on the labels.
In fact, it is not difficult to verify that by plugging in
our new Forster algorithm into the learning algorithm of~\cite{DKT21},
one directly obtains a strongly polynomial halfspace learner
in the presence of Massart noise. For the sake of the completeness,
here we focus on the realizable case and
provide a simpler, self-contained algorithm and proof.

Note that it is without loss of generality
to assume that the threshold of the target
halfspace is zero (one can reduce the general case to the homogeneous case).
The main challenge in PAC learning halfspaces is that the target halfspace may have
very bad anti-concentration (aka ``margin'').
If the margin is not too small (i.e., at least inverse polynomial),
simple iterative algorithms  (e.g., perceptron) efficiently learn halfspaces
(in strongly polynomial time). A natural idea
is then to reduce the general case to the large margin case
by appropriately transforming the data. A number of such
reductions have been developed in the literature~\cite{BlumFKV96, DV:04, DunaganV04, DKT21}.
The methods developed in~\cite{BlumFKV96, DV:04, DunaganV04} are inherently not strongly polynomial.
Recently,~\cite{DKT21} pointed out that one can use Forster transforms
for this purpose. 

For our purposes, we require a stronger guarantee than
what is provided by the vanilla perceptron algorithm.
Specifically, we want a learning algorithm for halfspaces that correctly classifies
at least some reasonable fraction of points, if the points are guaranteed to be
well-conditioned (for example, in the sense of being unit vectors with $\E[xx^{\top}] \approx I$).
By using an approximate Forster algorithm, we can transform the input points
in order to make them well-conditioned,
while preserving the notion of halfspaces. 
We can then apply our learner to this set in order to learn a classifier
that works on some reasonable fraction of the points.
Repeating this procedure iteratively on the unclassified points
eventually gives a halfspace learning algorithm.

More precisely, the modified perceptron algorithm of \cite{DunaganV04}
is a strongly polynomial time algorithm with the following performance guarantee:
given labeled examples consistent with an unknown linear classifier,
the algorithm learns a classifier that correctly labels
all points whose margin is not too small.
It is not hard to see that, for points in approximate radial isotropic position,
at least a $1/d$-fraction of points have not-too-small margin.
Therefore, if we have a set of points in approximate radial isotropic position,
the modified perceptron algorithm finds (in strongly polynomial time)
an explicit halfspace that separates out a roughly $1/d$-fraction
of the points all of the same sign.
By standard generalization bounds, this gives us an algorithm that
in strongly polynomial time learns a {\em partial classifier},
i.e., outputs a partial function that correctly
classifies an $\Omega(1/d)$-fraction of the points
while misclassifying an $O(\gamma/d)$-fraction.
In other words, this procedure produces a partial classifier
that labels at least a $1/d$-fraction of points
and misclassifies at most a $\gamma$-fraction of these points.

To learn an arbitrary halfspace,
we use our approximate Forster transform to put the
points in approximate radial isotropic position
without changing the notion of a halfspace on them. 
We then apply the above partial learner to these new points
in order to obtain a non-trivial partial classifier that makes mistakes
on only a $\gamma$-fraction of its classified set.
We repeat this process on the unclassified points,
using a new approximate Forster transform,
to learn a non-trivial fraction of the unclassified points.
Repeating this procedure iteratively as necessary,
we eventually obtain a partial classifier that produces
an answer on essentially all points of the domain
and only makes mistakes on a $\gamma$-fraction of them.

\subsection{Related Work} \label{sec:related}

In this section, we summarize additional prior work
that was not covered in the introduction.

\vspace{-0.2cm}

\paragraph{Comparison to Strongly Polynomial Algorithm for Matrix Completion}
It is worthwhile to compare our techniques for the Forster transform to~\cite{LinialSW00},
who developed the first strongly polynomial time algorithm for the matrix scaling problem.
To put this problem in terms more analogous to ours,
one is given a set of $d$ vectors $x_1, x_2,\ldots,x_d$ in $\R^d$.
The goal is to find a \emph{diagonal} matrix $A$ such that
if $y_i := Ax_i / \|Ax_i\|_1$ is the $\ell_1$ normalization of $Ax_i$,
then the absolute deviation of the $j^{th}$ coordinates of the $y$'s around
$0$ are (approximately) the same for all $j$. In particular,
it should hold that $\sum_{i=1}^d | (y_i)_j | \approx 1$ for all $1\leq j \leq d$.
Note that for our problem, we have $n$ (possibly greater than $d$) vectors,
$A$ can be \emph{any} matrix, we take $y_i$ to be the $\ell_2$ normalization
and we want the mean square deviation of the $y$'s in {\em any} direction
(not just along coordinate axes) to be approximately the same.

The algorithm in~\cite{LinialSW00} works roughly as follows.
We construct $A$ through an iterative sequence of improvements.
Given a specific $A$, we compute the appropriate values of $y_i$
and then compute the absolute deviations of each coordinate.
If these are all close to each other, we are done.
Otherwise, by sorting the deviations and finding the largest gap,
we can split our coordinates into two sets, $B$ and $S$,
so that the deviation of any coordinate in $B$ is substantially larger
than the deviation of any coordinate in $S$. One then defines
the diagonal matrix $C$ to be $(1+\delta)$ on the coordinates in $S$
and $1$ on the coordinates in $B$, and replaces $A$ by $A':= CA$.
It is not hard to see that by doing this,
one increases the deviations along all coordinates in $S$
while decreasing it along all coordinates in $B$
(and keeping the total sum of deviations the same).
By picking $\delta$ carefully, \cite{LinialSW00} show that
the variance of these coordinate-wise deviations
can be decreased by some polynomial amount in each step.
Thus, by iterating this method a polynomial number of times,
one obtains a scaling where the coordinate-wise deviations are sufficiently close.

The starting point for our algorithm is somewhat similar.
Given a matrix $A$, we try to find a matrix $C$ such that
the matrix $A' := CA$ is closer to satisfying our condition
(in the sense that $\|M_{A'}(X)\|_F$ should be smaller than $\|M_A(X)\|_F$ by an additive inverse
polynomial term). To do this, we compute subspaces $V_S$ and $V_B$
(by finding an eigenvalue gap in $M_A(X)$) such that
the variance of the $y_i := Ax_i/\|Ax_i\|_2$ in any direction along $V_B$
is substantially larger than along any direction in $V_S$.
Ideally, we would like to take $C=I+\alpha I_{V_S}$ for some carefully selected $\alpha$.
While this does only increase the variance in directions along $V_S$
and decrease it along $V_B$, in our setting this \emph{also} creates off-diagonal terms
that increase our potential. While it is always possible to ensure that this error does not
overwhelm the progress we make by taking $\alpha$ small enough,
in some cases (particularly where all of the $y$'s are either very close to lying in $V_B$
or very close to lying in $V_S$), this is not compatible with making polynomial progress in each step.
In this other case, we need to use a subtly different method for finding $C$
in order to minimize the contribution of these off-diagonal terms.
Furthermore, unlike in~\cite{LinialSW00}, the matrices $C$ used might have
large numerical complexity (perhaps on the order of the complexity of $A$).
If we naively apply the iterative algorithm as is, it might lead to computations
involving matrices with exponentially large bit complexity. In order to fix this,
we also need to add a rounding step, whereby in each stage we reduce
the numerical complexity of $A$ down to some manageable level
but without substantially affecting our potential.

\vspace{-0.2cm}

\paragraph{Comparison to Other Data Transformations}
The Forster transform is one of
several data transformations that have been studied in the literature to
make a dataset ``well-conditioned''. Here we explain two similar in spirit
such transformations, namely the ``outlier removal'' technique~\cite{BlumFKV96, DV:04}
and the rescaling method of~\cite{DunaganV04}. Both of these techniques
have been used to obtain weakly polynomial learners for halfspaces with random noise.

The ``outlier-removal'' technique was introduced
in~\cite{BlumFKV96} and was significantly refined by Dunagan and Vempala~\cite{DV:04}.
Given a dataset $X$ and a parameter $\beta>0$,
a point in $X$ is called a $\beta$-outlier if there exists a direction
$v$ such that the squared length of $x$ along $v$ is more than $\beta$ times the average
squared length of $X$ along $v$. The goal of the method is to efficiently
find a large subset of $X' \subseteq X$ such that $X'$ has no $\beta$-outliers, for as small
$\beta$ as possible. This would give a reasonable sized sub-distribution
on which the desired anti-concentration holds. As shown in~\cite{DV:04}, the parameter $\beta$
(which affects the quality of the resulting anti-concentration) needs to scale
polynomially with the bit complexity $b$ of the dataset $X$. Consequently,
the resulting runtimes in applications of this method will be
{\em inherently weakly polynomial}. Interestingly, this is the reason that the (random noise tolerant)
halfspace learner of~\cite{BlumFKV96} is only weakly polynomial.

A different algorithm for learning halfspaces with random classification noise is implicit
in the {\em rescaled perceptron} algorithm of
Dunagan and Vempala~\cite{DunaganV04} for efficiently solving linear programs
(see also~\cite{Betke04}).
The key ingredient of their approach is a rescaling step that linearly transforms the data
so that, roughly speaking, the margin increases in each iteration 
by a factor of $1+1/d$. 
Since the initial margin scales with the bit complexity, so does the total number of iterations.
(Since this leads to a proper learning algorithm, a dependence on the bit complexity is expected; 
otherwise, one would obtain a strongly polynomial algorithm for LP!)

\vspace{-0.2cm}

\paragraph{Strongly Polynomial Special Cases of LP}
A line of work, starting in the 80s, has developed strongly polynomial time algorithms
for interesting special cases of LP, including
minimum cost circulations~\cite{Tardos85, GoldbergT89, Orlin93},
min cost flow and multi-commodity flow problems~\cite{Tardos86, VavasisY96},
and generalized flow maximization~\cite{Vegh14, OlverV17, OlverV20}
(see also~\cite{DadushHNV20, DadushNV20}).
Strongly polynomial time algorithms have also been developed
for certain structured convex programs, see, e.g.~\cite{Vegh16, GargV19}
in the context of equilibrium
computation, and~\cite{LinialSW00} for matrix scaling.


\subsection{Organization} 
The structure of this paper is as follows:
In Section~\ref{sec:prelims}, we record basic notation and facts that will be used throughout this paper.
Section~\ref{sec:forster-modular} presents our Forster decomposition algorithm, 
assuming exact eigendecomposition and ignoring bit complexity issues. 
Section~\ref{sec:eigen} establishes our strongly polynomial guarantees for approximate
eigendecomposition. Section~\ref{sec:rounding} shows that we can efficiently 
round the entries of the underlying matrix without losing much in the desired guarantees. 
Finally, Section~\ref{sec:full-forster} puts all the pieces together to obtain 
our strongly polynomial Forster algorithm.
Section~\ref{sec:ltfs} presents our strongly polynomial halfspace learning algorithm.
Finally, in Section~\ref{sec:conc} we summarize our results and 
provide directions for future work.

\subsection{Acknowledgements}
We would like to thank Ravi Kannan, Santosh Vempala, and Mihalis Yannakakis for encouragement 
and insightful conversations about this work. We are grateful to Daniel Dadush for sharing his expertise on 
optimization, and for detailed feedback that improved the presentation of this paper. 
We are indebted to Nikhil Srivastava for answering our questions about the complexity of
eigenvalue decomposition, and for technical correspondence regarding 
our strongly polynomial eigendecomposition routine.

\section{Preliminaries} \label{sec:prelims}

Here we introduce some terminology and establish
a basic technical fact (Fact~\ref{fact:transform-properties}) 
that will be used throughout this paper.

\paragraph{Basic Notation}
We use $\Z_+$ to denote the non-negative integers, 
$\R^d$ for the $d$-dimensional real coordinate space, 
$\R_{\ast}^d$ for $\R^d \setminus \{ \0\}$, and $\s_d$ 
for the unit $\ell_2$-sphere.
For a set $S \subset \R$, we will denote $\max (S) \eqdef \max_{x \in S} x $ and 
$\min (S) \eqdef \min_{x \in S} x$. 

For $x \in \R^d$, we use $\|x\|_2$ to denotes the $\ell_2$-norm of $x$.
We use $\tr(\cdot)$, $\| \cdot \|_F$, and $\| \cdot \|_2$ for the trace, Frobenius norm,
and spectral norm of a square matrix. For matrices $A, B \in \R^{d \times d}$, 
we write $A \succeq B$ (or $B \preceq A$) to denote that $A-B $ is positive semidefinite (PSD).
We use $I$ for the $d\times d$ identity matrix, where the dimension will be clear from the context.
If $M \in \R^{d \times d}$ is a PSD matrix, we denote by $\lambda_i(M)$ and $q_i(M)$
the $i$-th largest eigenvalue and corresponding eigenvector of $M$.
That is, $\lambda_1(M) \geq \lambda_2(M) \geq \ldots \geq \lambda_d(M) \geq 0$
and $M q_i(M) = \lambda_i(M) q_i(M)$ for all $i \in [d]$. 
We denote by $\Lambda(M)$ the set of eigenvalues of $M$
and $Q(M)$ the set of eigenvectors. That is, $\Lambda(M) = \{\lambda_i(M), i \in [d] \}$
and  $Q(M) = \{q_i(M), i \in [d] \}$. 
For $S \subseteq [d]$, we denote $\Lambda_S(M) = \{\lambda_i(M), i \in S \}$
and $Q_S(M) = \{q_i(M), i \in S \}$.

For a finite set of vectors $S \subset \R^d$, we use $\spn(S)$ for their span.
For a subspace $V \subset \R^d$, we use $\dim(V)$ for its dimension 
and $V^{\perp}$ for its orthogonal complement. 
For $x \in \R^d$ and a subspace $V$,
we will denote by $\proj_V \, x$  the projection of $x$ onto $V$.
If $V = \spn(S)$, we will sometimes use $\proj_S \, x$ to denote $\proj_V \, x$.
For conciseness, we sometimes use $x^{(V)}$ for $\proj_V \, x$.
We denote by $I_V$ the $d \times d$ matrix 
with eigenvalues $1$ in $V$ and $0$ in $V^{\perp}$ (the projection of $I$ onto $V$).

\vspace{-0.2cm}

\paragraph{Additional Notation and Basic Fact}
For a dataset $X \subset \R_{\ast}^d$ of size $|X| = n$ and 
a linear transformation $A \in \R^{d \times d}$,
let $f_A: \R_{\ast}^d \to \s_d$ be defined by
$f_A(x) \eqdef \frac {A x} {\|A x\|_2}$.
We aim to find an invertible $A \in \R^{d \times d}$ such that 
$f_A$ brings a given dataset $X$ in (approximate) radial isotropic position.
We denote  $f_{A}(X) \eqdef \left\{ \frac{Ax}{\|Ax\|_2} \mid x \in X \right\}$
We will use various ``covariance-like'' matrices
for the initial dataset $X$ and its subsets.
For $X' \subseteq X$, we denote $M_A(X') \eqdef (1/n) \sum_{x \in X'} f_A(x) f_A(x)^{\top}$.
For subspaces $V_1, V_2 \subset \R^d$ and $X' \subseteq X$, we denote by 
$M^{V_1,V_2}_A(X') \eqdef (1/n) \sum_{x \in X'} f^{(V_1)}_A(x) f^{(V_2)}_A(x)^{\top}$, 
where we used the shorthand notation $y^{(V)} = \proj_V \, y$.
Note that the normalization factor is fixed in both cases. 
We start by recording some useful properties of the transformation $f_A$.

\begin{fact}\label{fact:transform-properties}
Let $A, B \in \R^{d\times d}$ be full-rank matrices. 
For any $x \in \R_{\ast}^d$ and $a \in \R_{\ast}$, the following hold:
\begin{enumerate}[(a)]
        \item\label{fact:scale-inv} $f_{a A} (x) = f_{A} (a x) = f_{A} (x)$.
        \item \label{fact:comp} $f_{B A} (x) = f_{B} (f_{A} (x))$.
        \item\label{fact:notfar} For $B \succeq I$, we have that $\| f_{B A} (x) - f_{A}(x ) \|_2 \le \|B - I\|_2$.
        \item\label{fact:dom} Let $V \subseteq \R^d$ be a subspace and let $B = I + a I_V$ for some 
        $a > 0$. Then, $f_{BA}^{(V)}(x)  = \lambda(x) \, f_A^{(V)}(x)$, 
        where $1 \leq \lambda(x) \leq 1+a$, and 
        $f_{BA}^{(V^\perp)}(x)  = \mu(x) \, f_A^{(V^\perp)}(x)$, 
        where $\frac{1}{1+a} \leq \mu(x) \leq 1$. 
    \end{enumerate}
\end{fact}

\noindent See Appendix~\ref{app:transform-properties} for the simple proof.

\section{Approximate Forster Transform in Strongly Polynomial Time} \label{sec:forster-modular}

In this section, we describe and analyze our algorithm that either computes an approximate Forster transform
of a given dataset or certifies that no Forster transform exists. There are two technical caveats in the algorithm 
presented in this section: First, we assume the existence of exact routines for matrix eigendecomposition. 
Second, we do not bound the bit complexity of the associated numbers. Both of these technical issues 
are handled in subsequent sections.

\subsection{Algorithm Pseudocode} \label{ssec:pc}

The algorithm aims to find a matrix $A \in \R^{d \times d}$
such that the transformation $f_A: \R^d \to \R^d$
brings the set $X$ in (approximate) radial isotropic position.
Starting from the initial guess $A = I$, the algorithm iteratively
improves the current matrix $A$ 
until the desired approximation is obtained or
a proper subspace $W$ of $\R^d$ is found
such that $|X \cap W|/n \ge \dim(W)/ d$.




\begin{algorithm}[hbt!]
   \caption{Main algorithm for computing Forster Transform}
   \label{alg:forster}
\begin{algorithmic}[1]
    \Function{$\textsc{ForsterTransform}$ }{set $X \subset \R_{\ast}^d$ of $n$ points, accuracy parameter $\eps$}
    \State Let $A \leftarrow I$  \hspace{22pt} \(\triangleright\) Initialization  of transformation matrix $A$
        \State $M_A \leftarrow M_A(X) = (1/n) \sum_{x \in X} f_A(x) f_A(x)^\top$
    	\While{$\| M_A \|^2_F > \frac 1 d + \frac {\eps^2} {d^2}$}
                \State Set $A \leftarrow \textsc{ImproveTransform}(A,X)$
                \State Set $M_A \leftarrow M_A(X) = (1/n) \sum_{x \in X} f_A(x) f_A(x)^\top$
        \EndWhile

        \State \Return $A$
    \EndFunction
  \end{algorithmic}
\end{algorithm}


\begin{algorithm}[hbt!]
    \caption{Find Improved Transform Matrix}
    \label{alg:improvement}
 \begin{algorithmic}[1]
     \Function{$\textsc{ImproveTransform}$ }
     {current matrix $A \in \R^{d\times d}$, $X {\subset \R_{\ast}^d}$, {accuracy parameter $\eps$}}

          \State {Set $M_A \leftarrow M_A(X) = (1/n) \sum_{x \in X} f_A(x) f_A(x)^\top$}
          \State Compute the set of eigenvalues, {$\Lambda = \Lambda(M_A)$, and eigenvectors,
          $Q = Q(M_A)$,} of $M_A$.
          \State Set $\gamma \leftarrow O(\frac {\eps^2} {d^4 n^2 })$, where $n := |X|$
          \State\label{step:partition} Partition $(\Lambda, Q)$ into two sets {of
                eigenvalues and corresponding eigenvectors},
                $(\Lambda_\Bg, Q_\Bg)$ and $(\Lambda_\Sm, Q_\Sm)$,
                maximizing $\min(\Lambda_\Bg) - \max(\Lambda_\Sm).$

          \Statex  {\vspace{5pt} \hspace{22pt} \(\triangleright\) Consider the Following Two Cases \vspace{5pt}}

          \If{there exists $x \in X$ such that $\left\| \proj_{Q_\Bg} f_A(x) \right\|_2, \left\| \proj_{Q_\Sm} f_A(x) \right\|_2 \ge \gamma$}

                \State Set $U \leftarrow \spn({Q_\Sm})$.
                \State Set $\alpha \leftarrow \frac {\eps}{8 n d^3}$. \label{alpha-case1}
	  \Else

                \State Set $X^\Bg \leftarrow \{x \in X: \| \text{proj}_{Q_\Bg} f_A(x)\|_2 \ge \gamma \}$.
                \State Set $M_A^\Bg \leftarrow M_A(X^\Bg) \eqdef  (1/n)  \sum_{x \in X^\Bg} f_A(x) f_A(x)^\top$.
	        \State Let $Q_\bg$ and $Q_\sm$ be the sets of top $|Q_\Bg|$
	                  and bottom $|Q_\Sm|$ eigenvectors of $M_A^\Bg$ respectively.
		\State Set $U \leftarrow \spn({Q_\sm})$.
		\State Set $\beta  \leftarrow \max_{x \in X^\Bg} \| f^{(U)}_A(x) \|_2 $.
                		
                 \If{$\beta = 0$}
                       \Statex {\vspace{5pt} \hspace{22pt}  \(\triangleright\)   No Forster Transform Exists  \vspace{5pt}}
		         \State Output the subspace $\spn(Q_\bg)$.
		         \hspace{22pt}
		 \Else     {\vspace{5pt} \hspace{22pt}  \(\triangleright\)   Case where $\beta>0$  \vspace{5pt}}
		
			\State Set $\alpha \leftarrow \frac 1 \beta \eps/(3d^2 n) - 1$ \label{alpha-case2}
		\EndIf
         \EndIf
        \State \Return ${A' :=} \left( I + \alpha I_U \right) A$
         \EndFunction
   \end{algorithmic}
 \end{algorithm}

\subsection{Analysis of Algorithm~\ref{alg:forster}}  \label{ssec:analysis-main}

\vspace{-0.1cm}

\paragraph{Our Potential Function}
Our algorithm measures the improvements between consecutive iterations
using the potential function
\begin{equation} \label{eqn:potential}
\Phi_X(A) \eqdef \| M_A \|^2_F
\end{equation}
corresponding to the squared Frobenius norm of the matrix
$$M_A \eqdef M_A(X) \eqdef (1/n)  \sum_{x \in X} f_A(x) f_A(x)^\top \;.$$
Recall that approximate radial isotropy condition amounts to the condition
$\frac{1-\eps}{d} I \preceq M_A \preceq \frac {1+\eps}{d} I$.
{Equivalently,} we want that $\|M_A - \frac 1 d I\|_2 \le \frac \eps d$ or that the eigenvalues of $M_A$
lie in $[\frac {1-\eps}  d , \frac {1+\eps} d ]$.
This is guaranteed to hold when the potential function becomes
less than $1/d + \eps^2/d^2$, as shown in the following lemma.

\begin{lemma} \label{lem:phi-prop}
Consider any dataset $X \subseteq \R^d_\ast$ and any full-rank matrix $A \in \R^{d\times d}$.
The following properties hold for the potential $\Phi_X(A) {= \| M_A \|^2_F}$.
\begin{enumerate}
\item $ 1/d \leq \Phi_X(A) \leq 1$.

\item If $\Phi_X(A) \le 1/d + \eps^2/d^2$ for some $\eps \in (0,1)$,
then for every eigenvalue $\lambda$ of $M_A$
it holds that $|\lambda - 1/d | \le \eps/d$.

\item \label{case:phi-large}
If $\Phi_X(A) > 1/d + \eps^2/d^2$ for some $\eps \in (0,1)$,
then (a) there exists an eigenvalue $\lambda$ of $M_A$ such that
$|\lambda- 1/d| > \eps/d^2$ and
(b) there exists a pair of consecutive eigenvalues $\lambda_i$ and $\lambda_{i+1}$ of $M_A$
such that $\lambda_i - \lambda_{i+1} > \eps/d^3$.
\end{enumerate}
\end{lemma}
\begin{proof}
Note that for any set $X$ the matrix $M_A$
is {PSD (as an autocorrelation matrix)}.
Let $\lambda_i = \lambda_i(M_A)$, $i \in [d]$,
with $\lambda_1 \geq \lambda_2 \geq \ldots \geq \lambda_d \geq 0$,
{be its eigenvalues}. Then we have that
$$\littlesum_{i=1}^d \lambda_i = \tr(M_A) = (1/n)  \littlesum_{x \in X} \tr \left( f_A(x) f_A(x)^{\top} \right) = 1 \;,$$
{where we used the linearity of the trace and the fact that}
all points $f_A(x)$ have unit $\ell_2$-norm.
Moreover, for the squared Frobenius norm of $M_A$
it holds that $\|M_A\|^2_F = \tr(M_A^2) = \sum_{i=1}^{d} \lambda_i^2$.
Below we prove each of the stated properties.
\begin{enumerate}[leftmargin=*]
\item Given that $\sum_{i=1}^d \lambda_i = 1$, the maximum possible value of
$\|M_A\|^2_F = \sum_{i=1}^d \lambda_i^2$ is equal to $1$,
and the minimum possible value is equal to $1/d$
(which is achieved when $\lambda_i = 1/d$, for all $i \in [d]$,
i.e., when $M_A = (1/d) I$, as desired).

\item Since $\sum_{i=1}^d \lambda_i = 1$, it holds that
$\sum_{i=1}^d \lambda_i^2 = 1/d + \sum_{i=1}^d (\lambda_i - 1/d)^2$.
By the {assumed upper} bound on $\Phi_X(A)$, we get that
$\sum_{i=1}^d (\lambda_i - 1/d)^2 \le \eps^2/d^2$,
and thus $\max_{i \in [d]} |\lambda_i - 1/d | \le \eps/d$.

\item By the {assumed lower bound on $\Phi_X(A)$, we get that}
$\sum_{i=1}^d (\lambda_i - 1/d)^2 > \eps^2/d^2$.
By an averaging argument,
there exists {$j \in [d]$ such that}
with $(\lambda_j - 1/d)^2 > \eps^2/d^3$,
and thus $|\lambda_j - 1/d | > \eps/d^2$.
{This proves (a).}
Since $\lambda_1 \geq {\max\{\lambda_j, 1/d\}}$ and
$\lambda_d \leq {\min\{\lambda_j, 1/d\}}$,
it follows that $\lambda_1  - \lambda_d > \eps/d^2$.
This implies that there is a gap between consecutive eigenvalues {of $M_A$, namely
there exists $i \in [d-1]$ such that} $\lambda_i - \lambda_{i+1} > \eps/d^3$.
{This proves (b).}
\end{enumerate}
\end{proof}

\paragraph{Bounding the Decrease in Potential}
We now proceed with the analysis.
We show that the algorithm \textsc{ImproveTransform}
{either correctly determines that no Forster transform exists or
computes} a transformation matrix with significantly reduced potential {value}.
{This statement implies correctness and simultaneously allows us to bound
the running time of our algorithm.

The main result of this section is the following proposition.}

\begin{proposition} \label{prop:main-iteration}
Let $A \in \R^{d\times d}$ be {a full-rank matrix}
and $X$ be a set of $n$ points in ${\R_{\ast}^d}$ such that
$\Phi_X(A) > \frac {1}{d} + \frac {\eps^2} {d^2}$ for some $\eps \in (0,1)$.
The algorithm \textup{\textsc{ImproveTransform}} returns a matrix $A'$ such that
\begin{equation}\label{eqn:potential-dec}
\Phi_X(A) - \Phi_X(A')  \ge  \Omega(\eps^5/(n^5d^{11}))
\end{equation}
{or correctly determines that no Forster Transform of $X$ exists,
in which case it returns} a subspace $W$ such that $|X \cap W| > {(n/d)} \,  \dim(W)$.
\end{proposition}

{In the rest of this section, we provide a proof of Proposition~\ref{prop:main-iteration}.}

Assuming that {a Forster transform of $X$ exists},
the algorithm \textsc{ImproveTransform} returns the
matrix $A' = \left( I + \alpha \, I_V \right) A$,
where $V \subset \R^d$ is an appropriate proper subspace of $\R^d$ and
$\alpha \in \R_{>0}$ is a carefully selected parameter (that depends on the structure of the dataset $X$).
The algorithm distinguishes two cases: In the first case, $\alpha$ is a small positive quantity,
equal to $\eps/(8nd^3)$, see Line~\ref{alpha-case1} in Algorithm~\ref{alg:improvement}. 
In the second case, $\alpha$ is set to $\frac 1 \beta \eps / (3 d^2 n) - 1$, 
and can be significantly larger than $1$ as it depends 
on a small parameter $\beta$ which is a function of the dataset $X$.
See Line~\ref{alpha-case2} in Algorithm~\ref{alg:improvement}.

\subsubsection{A Useful Structural Result} \label{sssec:impr-lem}

{We will use the notation $M_A = M_A(X)$ and $M_{A'} = M_{A'}(X)$.
To bound the desired quantity, $\Phi_X(A) - \Phi_X(A') = \|M_A\|^2_F - \|M_{A'}\|^2_F$,
we will make essential use of the following key lemma:}

\begin{lemma}\label{lem:impr}
For any $X \subset {\R_{\ast}^d}$ and any full-rank matrix $A \in \R^{d\times d}$
the following holds.
For any subspace $V \subset \R^d$ and any scalar $\alpha > 0$, {for $A' \eqdef (I+\alpha I_V) A$},
we have that
\begin{equation}\label{eqn:help-impr}
\Phi_X(A) - \Phi_X(A') \geq
2 \left(\lambda_k( M^{V^{\perp},V^{\perp}}_{A} ) - \lambda_1( M^{V,V}_{A} )- 2D_f  \right) D_f - 2 \|M^{V,V^\perp}_{A'}\|^2_F \;,
\end{equation}
where {$k = \dim(V^{\perp})$} and
$D_f \eqdef \frac 1 n \sum_{x \in X} \left(\|f^{{(V)}}_{A'}(x)\|_2^2 - \|f^{{(V)}}_{A}(x)\|_2^2 \right)$.
\end{lemma}

Lemma~\ref{lem:impr} bounds the improvement in potential in terms of two opposing contributions. 
On the one hand, there is a decrease in the potential proportional to the amount of mass $D_f$ 
transferred from the subspace $V^\perp$ to the subspace $V$ 
times the eigenvalue gap between the subspaces $V$ and $V^\perp$. 
On the other hand, there is an increase in potential due to the cross terms $V\times V^\perp$ 
that get created after the transformation by $A'$.

\begin{proof}[Proof of Lemma~\ref{lem:impr}]
By definition, we have that
$\Phi_X(A) - \Phi_X(A') = \|M_A\|^2_F - \|M_{A'}\|^2_F$.
We decompose each matrix into  
block matrices specified by the subspaces $V$ and $V^{\perp}$.
We write $M_A$ as $M^{V,V}_A + M^{V,V^\perp}_A + M^{V^\perp,V}_A + M^{V^\perp,V^\perp}_A$, where 
we recall that for subspaces $S$ and $T$ we defined $M^{S,T}_A$ as $I_S M_A I_T$.
Since $V$ and $V^\perp$ are orthogonal subspaces, we also have that
$$\|M_A\|_F^2 = \|M^{V,V}_A\|_F^2 + \|M^{V,V^\perp}_A\|_F^2 + \|M^{V^\perp,V}_A\|_F^2 + 
\|M^{V^\perp,V^\perp}_A\|_F^2.$$
We can thus express $\|M_A\|^2_F - \|M_{A'}\|^2_F$ as the sum of the following three terms:
\begin{itemize}
\item[(i)] $\|M^{V,V}_A\|^2_F - \|M^{V,V}_{A'}\|^2_F$
\item[(ii)] $\|M^{V^\perp,V^\perp}_A\|^2_F - \|M^{V^\perp,V^\perp}_{A'}\|^2_F$
\item[(iii)] $2 \|M^{V,V^\perp}_A\|^2_F - 2 \|M^{V,V^\perp}_{A'}\|^2_F$
\end{itemize}
By Fact~\ref{fact:transform-properties} part~\ref{fact:dom},
for any $x \in {\R_{\ast}^d}$,  we have that
$f^{{(V)}}_{A'}(x) (f^{{(V)}}_{A'}(x))^\top \succeq f^{{(V)}}_{A}(x)  (f^{{(V)}}_{A}(x))^\top $ and
$f^{{(V^\perp)}}_{A'}(x) (f^{{(V^\perp)}}_{A'}(x))^\top
\preceq f^{{(V^\perp)}}_{A}(x) (f^{{(V^\perp)}}_{A}(x))^\top$.
Since $M^{V,V}_{A}$ is equal to $(1/n)$ times the sum of $f^{{(V^\perp)}}_{A}(x) (f^{{(V^\perp)}}_{A}(x))^\top$ over $x \in X$, 
we obtain
$M^{V,V}_{A'} \succeq M^{V,V}_{A}$ and $M^{V^\perp,V^\perp}_{A'} \preceq M^{V^\perp,V^\perp}_{A}$.

We will use the following linear-algebraic fact to bound from below the first two terms.

\begin{fact}\label{fact:psd}
Let $A, B \in \R^{d \times d}$ be symmetric PSD matrices.
If $A \succeq B$, then it holds that
$$2 \, \tr(A-B) \lambda_{k}(B)   \leq \|A\|_F^2 - \|B\|_F^2 \leq  2 \tr(A-B) \,  \lambda_{1}(A) \;,$$
where $k = \rank(A)$
and $\lambda_{k}(B)$ is the $k$-th largest eigenvalue of $B$.
\end{fact}
\begin{proof}
We use $\bullet$ to denote the entrywise inner product
between two matrices. Note that $A \bullet B = \tr(A^\top B)$ which is equal to $\tr(A B)$ if the matrices are symmetric.
We have that
$$\|A\|_F^2 - \|B\|_F^2 = (A-B) \bullet (A+B) = \tr( (A+B) (A-B) ).$$
Since both $A+B$ and $A-B$ are PSD, multiplication by $A+B$ increases the eigenvalues of
$A-B$ by at least a factor of $2\lambda_{k}(B)$ and at most $2\lambda_{1}(A)$.
Thus, $\tr( (A+B) (A-B) )$ is  bounded between
$2 \tr(A-B) \lambda_{k}(B)$ and $ 2 \tr(A-B) \lambda_{1}(A)$.
\end{proof}

\noindent
{We start by bounding term (i) from below.}
Using Fact~\ref{fact:psd} {applied to the matrices $M^{V,V}_{A'}$ and $M^{V,V}_{A}$},
we get that
\begin{align}
\|M^{V,V}_A\|^2_F - \|M^{V,V}_{A'}\|^2_F
&\geq 2 \lambda_1 ( M^{V,V}_{A'} ) \, \tr({M^{V,V}_A - M^{V,V}_{A'}}) \nonumber \\
&= \frac 2 n \lambda_1( M^{V,V}_{A'} ) \sum_{x \in X} \left( \|f^{{(V)}}_A(x)\|_2^2 - \|f^{{(V)}}_{A'}(x)\|_2^2 \right) \;. \label{lb:i}
\end{align}
Similarly, {we can bound below term (ii).}
Using Fact~\ref{fact:psd} {applied to the matrices
$M^{V^\perp,V^\perp}_A$ and $M^{V^\perp,V^\perp}_{A'}$},
we get that
\begin{align}
\|M^{V^\perp,V^\perp}_A\|^2_F - \|M^{V^\perp,V^\perp}_{A'}\|^2_F
&\geq 2 \lambda_k(M^{V^\perp,V^\perp}_{A'}) \tr({M^{V^\perp,V^\perp}_A - M^{V^\perp,V^\perp}_{A'}}) \nonumber \\
&= \frac{2}{n} \lambda_k(M^{V^\perp,V^\perp}_{A'}) \sum_{x \in X} \left(\|f^{{(V^\perp)}}_A(x)\|_2^2 - \|f^{{(V^\perp)}}_{A'}(x)\|_2^2\right) \nonumber \\
&= \frac{2}{n} \lambda_k( M^{V^\perp,V^\perp}_{A'} ) \sum_{x \in X} \left(\|f^{{(V)}}_{A'}(x)\|_2^2 - \|f^{{(V)}}_{A}(x)\|_2^2\right) \;, \label{lb:ii}
\end{align}
where $k = \rank (M^{V^\perp,V^\perp}_{A'})$ and the last equality follows since
$\|f^{{(V)}}_{A}(x)\|_2^2 + \|f^{{(V^\perp)}}_{A}(x)\|_2^2 = 1$.

\noindent Finally, we straightforwardly bound from below the third term as follows:
\begin{equation}\label{lb:iii}
2 \|M^{V,V^\perp}_A\|^2_F - 2 \|M^{V,V^\perp}_{A'}\|^2_F \geq - 2 \|M^{V,V^\perp}_{A'}\|^2_F \;.
\end{equation}
\noindent Overall, recalling that $D_f = \frac 1 n \sum_{x \in X} (\|f^{(V)}_{A'}(x)\|_2^2 - \|f^{(V)}_{A}(x)\|_2^2)$,
\eqref{lb:i}, \eqref{lb:ii}, \eqref{lb:iii} give that
$$\Phi_X(A) - \Phi_X(A') \geq
2 \left(\lambda_k( M^{V^\perp,V^\perp}_{A'} ) - \lambda_1( M^{V,V}_{A'} ) \right) D_f - 2 \|M^{V,V^\perp}_{A'}\|^2_F \;.$$
To complete the proof, we note that
$$\lambda_1( M^{V,V}_{A'} ) 
\leq \lambda_1( M^{V,V}_{A})+  \| M^{V,V}_{A} - M^{V,V}_{A'} \|_2$$
and that
$$\lambda_k( M^{V^\perp,V^\perp}_{A'} ) 
\leq \lambda_k( M^{V^\perp,V^\perp}_{A'} )  +  \| M^{V^\perp,V^\perp}_{A} - M^{V^\perp,V^\perp}_{A'} \|_2
= \lambda_k( M^{V^\perp,V^\perp}_{A'} )  +  \| M^{V,V}_{A} - M^{V,V}_{A'} \|_2$$

Finally, we have that
\begin{align*}
  \| M^{V,V}_{A} - M^{V,V}_{A'} \|_2
&\leq \frac{1}{n} \sum_{x \in X} \left\| f^{(V)}_{A'}(x) f^{(V)}_{A'}(x)^\top - f^{(V)}_{A}(x) f^{(V)}_{A}(x)^\top \right\|_2 \\
&= \frac{1}{n} \sum_{x \in X} \left(\|f^{(V)}_{A'}(x)\|_2^2 - \|f^{(V)}_{A}(x)\|_2^2 \right) \;,
\end{align*}
{where the last equality follows from Fact~\ref{fact:transform-properties}\ref{fact:dom}.}
{Combining the above completes the proof of Lemma~\ref{lem:impr}.}
\end{proof}

In the following two subsections, we analyze the two cases of \textsc{ImproveTransform} separately. 
{Note that \textsc{ImproveTransform} requires that we be able to do exact singular value decompositions
in order to compute the subspace $U$. Our final algorithm will not be able to do this exactly 
and will need to make do with an approximate singular value decomposition (see Section~\ref{sec:eigen}). 
In order to make our extension easier, we will show that the potential decrease holds even when $U$ 
is replaced by some $V$ which satisfies some approximation of the properties that $U$ does.}

\subsubsection{Case I: 
There exists \texorpdfstring{$x \in X$}{x} such that 
\texorpdfstring{$\|\proj_{Q_\Bg} f_A(x)\|_2, \|\proj_{Q_\Sm} f_A(x)\|_2 \ge \gamma$}{...}}

In order to analyze this case, we prove the following proposition:
\begin{proposition}\label{case-1-prop}
Suppose that $X$ is a set of $n$ points in $\R^d_\ast$ and $A$ an invertible $d\times d$ matrix. Suppose that $V\subset \R^d$ is a subspace so that for $\alpha, \rho >0$ with $\alpha \leq \eps/(64 nd^3)$:
\begin{enumerate}
\item The maximum over $x\in X$ of $\min(\|f_A^{(V)}(x)\|_2,\|f_A^{(V^\perp)}(x)\|_2)$ equals $\rho$.\label{large-deviation-property}
\item $\lambda_{\min} (M_A^{V^\perp V^\perp}(X)) - \lambda_{\max}(M_A^{VV}(X)) \geq \frac{\eps}{2d^3}.$\label{ev-gap-property}
\item $\|M_A^{VV^\perp}(X)\|_F \leq \alpha\rho$.\label{small-off-diagonal-property}
\end{enumerate}

Then for $C = (I+\alpha I_V)A$ we have that $\Phi_X(C) \leq \Phi_X(A) - \rho^2 \eps/(8nd^3)$.
\end{proposition}

We note that if $\Phi_X(A) > \frac{1}{d} + \frac {\eps^2} {d^2}$, then by Lemma~\ref{lem:phi-prop} part~\ref{case:phi-large}, the difference between the largest and smallest eigenvalues of $M_A(X)$ will be at least $\eps/d^2$, and therefore the largest eigenvalue gap will be at least $\eps/d^3$. Thus, for $V$ taken to be the $U$ given in Algorithm \ref{alg:improvement}, Property \ref{ev-gap-property} will hold. Furthermore, for as $U$ is an eigenspace of $M_A(X)$, $M^{U,U^\perp}_A(X) = \mathbf{0}$ and Property \ref{small-off-diagonal-property} will hold.

The rest of this section will be devoted to proving Proposition \ref{case-1-prop}.

{To bound below the improvement in potential, we will make essential use of Lemma~\ref{lem:impr}. We bound the relevant quantities in the following lemmas.}

\begin{lemma}\label{lem:c1-Df}
Letting $D_f = \frac{1}{n} 
\sum_{x\in X} \left( 
    \| f_C^{(V)}(x)\|_2^2 - 
    \| f_A^{(V)}(x)\|_2^2
\right)$, 
we have that $\frac {\alpha \rho^2} {2 n} \leq D_f \leq 2\alpha$.
\end{lemma}
\begin{proof}
By Fact~\ref{fact:transform-properties}\ref{fact:notfar}, it follows that
$D_f \le 2\alpha$, as $C = (I+\alpha I_V) A$ and $\|(I+\alpha I_V)-I\|_2 \le \alpha$. This implies that $\| f_C^{(V)}(x)- f_A^{(V)}(x)\|_2 \leq \| f_C(x)- f_A(x)\|_2 \leq \alpha$ for all $x$. Thus (since $x\rightarrow x^2$ is $2$-Lipschitz on $[0,1]$), we have that $\| f_C^{(V)}(x)\|_2^2 - \| f_A^{(V)}(x)\|_2^2 \leq 2\alpha$ for all $x$, so $D_f \leq 2\alpha$.

To bound $D_f$ from below, consider an $x^{\ast} \in X$ that maximizes the quantity
$$\min\{\|\proj_{Q_\Bg} f_A(x)\|_2, \|\proj_{Q_\Sm} f_A(x)\|_2\}$$
over $x \in X$. Recall that the maximum value of the above quantity equals $\rho$.

By assumption, we must have that $\rho \ge \gamma$. As Fact~\ref{fact:transform-properties}\ref{fact:dom} implies that all terms in the sum defining $D_f$ are nonnegative, we have that
$$D_f \geq \frac{1}{n} \left( \|f^{(V)}_{(I+\alpha I_V) A}(x^{\ast})\|_2^2 - \|f^{(V)}_{A}(x^{\ast})\|_2^2 \right) \;.$$

Let $y = f_{A}(x^{\ast})$, and recall that we use $y^{(V)} = f^{(V)}_{A}(x^{\ast})$ and $y^{(V^\perp)} = f^{(V^\perp)}_{A}(x^{\ast})$
for the projections of $y$ onto $V$ and $V^{\perp}$ respectively.
With this notation, we can write
$$D_f \geq \frac{1}{n} \left( \|f^{(V)}_{I+\alpha I_V}(y)\|_2^2 - \|y^{(V)}\|_2^2 \right)
= \frac{1}{n} \left(\frac{(1+\alpha)^2 \|y^{(V)}\|^2_2}{(1+\alpha)^2 \|y^{(V)}\|^2_2 + \|y^{(V^\perp)}\|^2_2 } - \|y^{(V)}\|_2^2 \right) \;,$$
{where we used Fact~\ref{fact:transform-properties}\ref{fact:comp} and the definition of the transformation $f_{A}$.}
{The Pythagorean theorem and the definition of $y$ give that}
$\|y^{(V)}\|^2_2 + \|y^{(V^\perp)}\|^2_2 = \|y\|^2_2 = 1$.
We thus obtain
$$D_f \geq \frac{1}{n} \left( \frac {2 \alpha \, \|y^{(V)}\|^2_2 \, \|y^{(V^\perp)}\|^2_2} {(1+\alpha)^2 } \right)
\geq \frac {2 \alpha \rho^2 (1-\rho^2)} {(1+\alpha)^2 n} \geq \frac {\alpha \rho^2} {2 n} \geq \frac {\alpha \gamma^2} {2 n} \;,$$
where the last inequality follows since $\alpha$ is sufficiently smaller than $1$
and $\rho^2 \le \frac{1}{2}$.
\end{proof}

Finally, we bound $\|M^{V,V^\perp}_{C}\|^2_F$ from above in the following lemma:

\begin{lemma}\label{lem:c1-MVVp}
We have that $\|M^{V,V^\perp}_{C}\|^2_F \leq 4 \alpha^2 \rho^2$.
\end{lemma}
\begin{proof}
By the triangle inequality for the Frobenius norm, we have that
$$\|M^{V,V^\perp}_{C}\|^2_F \leq \left( \|M^{V,V^\perp}_{A} \|_F + \|M^{V,V^\perp}_{C} - M^{V,V^\perp}_{A}\|_F \right)^2 \;.$$
We have that $\|M^{V,V^\perp}_{A}\|_F \leq \alpha\rho$,
by assumption.
We bound above the second term using the following sequence of inequalities:
\begin{align*}
\|M^{V,V^\perp}_{C} - &M^{V,V^\perp}_{A}\|_F^2
= \left\| \frac{1}{n} \sum_{x \in X} \left( f^{(V)}_{C}(x) f^{{(V)}^\perp}_{C}(x)^\top -
f^{(V)}_{A}(x)f^{(V^\perp)}_{A}(x)^\top \right) \right\|_F^2 \\
&\leq \max_{x \in X} \left\| f^{(V)}_{C}(x)f^{(V^\perp)}_{C}(x)^\top - f^{(V)}_{A}(x)f^{(V^\perp)}_{A}(x)^\top \right\|_F^2\\
&= \max_{x\in X} \left\| \frac{1+\alpha}{\|(1+\alpha) f^{(V)}_{A}(x)\|^2+ \| f^{(V^\perp)}_{A}(x) \|^2}
f^{(V)}_{A}(x) f^{(V^\perp)}_{A}(x)^\top - f^{(V)}_{A}(x) f^{(V^\perp)}_{A}(x)^\top \right\|_F^2\\
&\leq \max_{x \in X} \left\| \alpha f^{(V)}_{A}(x) f^{(V^\perp)}_{A}(x)^\top \right\|_F^2\\
&\le \alpha^2 \rho^2 (1-\rho^2) \le \alpha^2 \rho^2 \;,
\end{align*}
{where the third line uses Fact~\ref{fact:transform-properties}\ref{fact:comp}
and the definition of the transformation $f_{A}$.}
\end{proof}

Combining the above lemmas, we obtain that
$$\Phi_X(A) - \Phi_X(C) \ge \left( \frac{\eps}{2d^3} - 4 \alpha \right) \frac {\alpha \rho^2} {n} - 8 \alpha^2 \rho^2 \;.$$
We note that so long as $\alpha \leq \eps/(64 nd^3)$ the above is at least
$$
\left(\alpha \rho^2 \right)\left( \left(\frac{\eps}{4d^3} \right)\frac{1}{n} - 8\alpha \right) \geq \alpha \rho^2 \eps/(8nd^3).
$$
This completes our proof of Proposition \ref{case-1-prop}.

\subsubsection{Case II: For all \texorpdfstring{$x \in X$}{x}, either \texorpdfstring{$\|\proj_{Q_\Bg} f_A(x)\|_2 \le \gamma$}{...}
or \texorpdfstring{$\|\proj_{Q_\Sm} f_A(x)\|_2 \le \gamma$}{...}}

In this case, all points $x\in X$ lie within a $\gamma$ margin from the subspaces spanned by the vectors $Q_\Bg$ and $Q_S$.
The algorithm updates the matrix $A$ by considering only the set of  ``big'' points $X^\Bg$, {i.e., the points in $X$ whose images under $f_A$
have sufficiently large projections on the subspace spanned by the large eigenvectors of $M_A$}. In more detail, instead of using the eigenvectors $Q_\Bg, Q_S$ of the matrix $M_A = M_A(X)$,
the algorithm uses the eigenvectors $Q_b, Q_s$ of the matrix $M_{A}(X^\Bg) = {(1/n)  \sum_{x \in X^\Bg} f_A(x) f_A(x)^\top}$,
setting $U = \spn(Q_s)$. This is done to ensure that the cross-terms $M^{U,U^\perp}_{A}(X^\Bg)$ start out at 0 initially and remain small despite significant rescaling of the subspace $U$.
Moreover, despite the change in the definition, we show (in Claim~\ref{claim:marginB} and Claim~\ref{claim:marginS}) that
the corresponding subspaces $U$ and $U^\perp$
satisfy a similar margin condition to the subspaces spanned by $Q_\Bg$ and $Q_S$ and that there is still a significant eigenvalue gap between $U$ and $U^\perp$.
The margin condition is shown in Claim~\ref{claim:marginB} and Claim~\ref{claim:marginS}, and the eigenvalue bounds are proven in Lemmas~\ref{lem:deltamss}
and~\ref{lem:deltambb}.

These properties will allow us to bound the decrease in potential in this case.
We will show the following result.

\begin{proposition}\label{case-2-prop}
Suppose that $X$ is a set of $n$ points in $\R^d_\ast$ and $A$ an invertible $d\times d$ matrix. Suppose that for some $k < n$ that $\lambda_k(M_A(X)) - \lambda_{k+1}(M_A(X)) \geq \eps/(2d^3)$.
Let $W$ be the span of the $d-k$ smallest eigenvalues of $M_A(X)$.

Suppose furthermore that for some $\gamma$
at most a sufficiently small multiple of $\eps^2/(d^4 n^2)$
that every $x\in X$ satisfies $\min(\|f_A^{(W)}(x)\|_2,\|f_A^{(W^\perp)}(x)\|_2) \leq \gamma$.
Let $X^\Bg$ denote the set of $x\in X$ so that $\|f_A^{(W)}(x)\|_2 \leq \gamma$
and $X^S = X\backslash X^\Bg$. Let $0 < \delta < \gamma$.
Suppose that $V\subset \R^d$ is a $(d-k)$-dimensional subspace so that:
\begin{enumerate}
\item $\tr(M_A^{V,V}(X^\Bg)) \leq \tr(M_A^{W,W}(X^\Bg))+\delta^2$, \label{V small property}
\item $\tr(M_A^{V^\perp,V^\perp}(X^\Bg))\geq \tr(M_A^{W^\perp,W^\perp}(X^\Bg)) - \delta^2$, \label{VT frob property}
\item $\lambda_{k}(M_A^{V^\perp,V^\perp}(X^\Bg)) \geq \lambda_{k}(M_A(X^\Bg)) - \delta$, \label{VT eigenvalue property}
\item $\|M_A^{V,V^\perp}(X^\Bg)\|_F \leq \beta\delta$, \label{cross term property}
\end{enumerate}
where $\beta = \max_{x\in X^\Bg} \|f_A^{(V)}(x)\|_2$.
Then if $\beta=0$, $V^\perp$ contains more than $kn/d$ elements of $X$.
Otherwise, setting $\alpha = \eps/(3\beta d^2 n) - 1$ and $C = (I+\alpha I_V)A$, we have that
$$
\Phi_X(C) \leq \Phi_X(A) - \Omega(\eps^3/(d^7n^3)) \;.
$$
\end{proposition}

We note that if $V$ is taken to be the space of the bottom $d-k$ eigenvalues of $M_A(X^\Bg)$ that the above properties trivially hold with $\delta = 0$. Properties \ref{V small property} and \ref{VT frob property} follow from the variational characterization of eigenspaces. Properties \ref{VT eigenvalue property} and \ref{cross term property} hold trivially.

\medskip

\noindent We know that elements of $X^\Bg$ are close to $W^\perp$ and elements of $X^S$ are close to $W$. 
We will need to claim that elements of $X^\Bg$ are also close to $V^\perp$ and elements of $X^S$ are close to $V$. 
We establish this in the next two claims.

\begin{claim}\label{claim:marginB}
We have that $\frac{1}{n} \sum_{x \in X^\Bg} \|f^{(V)}_A(x)\|_2^2 \le \gamma^2+\delta^2$.
In particular, for any $x \in X^\Bg$, it holds that $\|f^{(V)}_A(x)\|_2 \le \sqrt{2n} \gamma$.
\end{claim}
\begin{proof}
This follows from Property \ref{V small property} and the fact that
$$
\tr(M^{W,W}_A(X^\Bg)) = \frac{1}{n} \sum_{x\in X^\Bg} \|f_A^{(W)}(x)\|_2^2 \leq \gamma^2.
$$
\end{proof}

\begin{claim}\label{claim:marginS}
We have that $\frac{1}{n} \sum_{x \in X \setminus X^\Bg} \|f^{(V^\perp)}_A(x)\|_2^2 \leq \gamma^2+\delta^2$.
In particular, for any $x \in X \setminus X^\Bg$, it holds that
$\|f^{(V^\perp)}_A(x)\|_2 \le \sqrt{2n} \gamma$.
\end{claim}

\begin{proof}
Recalling that $W^\perp$ is the span of the principle eigenvectors of $M_A(X)$, by the variational characterization of eigenspaces,
it maximizes the quantity
$\tr(M^{Z,Z}_A(X)) = \frac{1}{n} \sum_{y \in f_A(X)} \|y^{(Z)}\|_2^2$
over all subspaces $Z$ with $\dim(Z) = \dim(W^\perp) = k$.
In particular, it holds that
$$\frac{1}{n}  \sum_{y \in f_A(X)} \| y^{(W^\perp)}\|_2^2 \geq
\frac{1}{n} \sum_{y \in f_A(X)} \|y^{(V^\perp)}\|_2^2 \;.$$
On the other hand, by Property \ref{VT frob property}, we have that
$$ \frac{1}{n} \sum_{y \in f_A(X^\Bg)} \| y^{(V^\perp)}\|_2^2 = \tr(M^{V^\perp,V^\perp}_A(X^\Bg)) \geq \tr(M^{W^\perp,W^\perp}_A(X^\Bg))-\delta^2 =
\frac{1}{n} \sum_{y \in f_A(X^\Bg)} \| y^{(W^\perp)}\|_2^2 -\delta^2 \;.$$
Subtracting the above two inequalities, we get that
$$\frac{1}{n} \sum_{x \in X \setminus X^\Bg} \|f^{(V^\perp)}_A(x)\|_2^2
\leq \frac{1}{n}  \sum_{x \in X \setminus X^\Bg} \|f^{(W^\perp)}_A(x)\|_2^2 + \delta^2
\leq \gamma^2+\delta^2 \;.$$
This gives the claim.
\end{proof}

\paragraph{The case that $\beta > 0$:}

Here we assume that $\beta>0$ and show that we can obtain an improvement in our potential function. 
To bound below the improvement in potential,
we will make essential use of Lemma~\ref{lem:impr}, 
and we bound the relevant quantities in a sequence of lemmas.

We begin by bounding below $\lambda_k( M_A^{V^\perp,V^\perp}(X) )$. 
In particular, we show that it is nearly as big as $\lambda_k(M_A(X))$. 
Morally, this holds because
$$
\lambda_k(M_A(X)) = \lambda_k( M_A^{W^\perp,W^\perp}(X) ) \approx \lambda_k( M_A^{V^\perp,V^\perp}(X) ) \;.
$$
Formally, we have the following.

\begin{lemma}\label{lem:deltambb}
We have that $\lambda_k( M_A^{V^\perp,V^\perp}(X) ) \geq  \lambda_k(M_A(X)) - 4 \gamma$,
for $k = \dim (V^\perp)$.
\end{lemma}
\begin{proof}
We bound the $k$-th largest eigenvalue of $M_A^{V^\perp,V^\perp}$ via the following sequence of inequalities.
\begin{align}
    \lambda_k( M_A^{V^\perp,V^\perp}(X) )
&\ge \lambda_k( M_A^{V^\perp,V^\perp}(X^\Bg) ) \label{deltambb1} \\
&\ge \lambda_k( M_A(X^\Bg) ) -\gamma   \label{deltambb2} \\
&\ge \lambda_k( M^{W^\perp,W^\perp}_A(X^\Bg) ) - 3 \gamma   \label{deltambb3}
\\
&\ge \lambda_k( M^{W^\perp,W^\perp}_A ) - 3 \gamma - \gamma^2 \label{deltambb4} \\
&= \lambda_k( M_A ) - 3 \gamma - \gamma^2 \;. \label{deltambb5}
\end{align}
Inequality~\eqref{deltambb1} follows since
$$M_A^{V^\perp,V^\perp}(X) = M_A^{V^\perp,V^\perp}(X^\Bg) +
M_A^{V^\perp,V^\perp}(X \setminus X^\Bg) \succeq M_A^{V^\perp,V^\perp}(X^\Bg) \;.$$
Inequality~\eqref{deltambb2} follows from Property \ref{VT eigenvalue property} and $\delta \leq \gamma$.

\noindent Inequality~\eqref{deltambb3} follows since
Inequality~\eqref{deltambb4} follows since
$$\|M^{W^\perp,W^\perp}_A(X) - M^{W^\perp,W^\perp}_A(X^\Bg)\|_2
= \|M^{W^\perp,W^\perp}_A(X \setminus X^\Bg)\|_2.$$
This is at most $\gamma^2$ as it is bounded by
$$
\sup_{\|v\|_2=1} v^\top \left(\frac{1}{n}\sum_{x\in X^S} f_A(x)^{(W^\perp)}(f_A(x)^{(W^\perp)})^\top \right)v 
\leq \sup_{\|v\|_2=1}\max_{x\in X^S} |v\cdot f_A(x)^{(W^\perp)}|^2 \leq \gamma^2 \;.
$$

Equality~\eqref{deltambb5} follows since
$M_A^{W^\perp,W^\perp}$ and $M_A$ agree in the top $k$ eigenvalues
by definition of the subspace $W$, as the span of the top $k$ eigenvectors of $M_A$.

Lemma~\ref{lem:deltambb} now follows as $\gamma^2 \le \gamma \le 1$.
\end{proof}

Next we bound from above $\lambda_1( M_A^{V , V}(X) )$. 
In particular, we show that it is not much larger than $\lambda_{k+1}(M_A(X))$. 
Morally, this holds because
$$
\lambda_{k+1}(M_A(X)) =\lambda_1( M_A^{W,W}(X) ) \approx \lambda_1( M_A^{V,V}(X) ) \;.
$$
Formally, we have the following. 

\begin{lemma}\label{lem:deltamss}
We have that $\lambda_1( M_A^{V, V}(X) ) \le \lambda_{k+1}(M_A(X)) + 8 \gamma$.
\end{lemma}
\begin{proof}
We bound the largest eigenvalue of $M^{V,V}_{A}$
via the following sequence of inequalities.
\begin{align}
\lambda_1( M^{V,V}_{A} )
&\le \lambda_1( M^{V,V}_{A}(X \setminus X^\Bg) ) + 2\gamma^2 \label{deltamss1} \\
&\le \lambda_1( M_{A}(X \setminus X^\Bg) ) + 4 \gamma + 2\gamma^2    \label{deltamss2} \\
&\le \lambda_1( M^{W,W}_{A}(X \setminus X^\Bg) ) + 6 \gamma + 2\gamma^2  \label{deltamss3}  \\
&\le \lambda_1( M^{W,W}_{A}(X) ) + 6 \gamma + 2\gamma^2 \label{deltamss4} \\
&\le \lambda_{k+1}(M_A(X)) + 8 \gamma  \;. \label{deltamss5}
\end{align}
Inequality~\eqref{deltamss1} follows since
$\lambda_1( M^{V,V}_{A} ) \leq \lambda_1( M^{V,V}_{A}(X \setminus X^\Bg) )
+ \lambda_1( M^{V,V}_{A}(X^\Bg))$
and
$$\lambda_1( M^{V,V}_{A}(X^\Bg) ) = \| M^{V,V}_{A}(X^\Bg) \|_2
\leq \frac{1}{n} \sum_{x \in X^\Bg} \|f_A^{(V)}(x)\|^2_2
\leq 2 \gamma^2\;,$$
{where the last inequality follows from Claim~\ref{claim:marginB}} and since $\delta \le \gamma$.

\noindent
Inequality~\eqref{deltamss2} follows since
$\|M_{A}(X \setminus X^\Bg) - M^{V,V}_{A}(X \setminus X^\Bg)\|_2 \le 4\gamma$.
Indeed, note that
\begin{align*}
\|M_{A}(X \setminus X^\Bg) &- M^{V,V}_{A}(X \setminus X^\Bg)\|_2
\leq \frac{1}{n} \sum_{x \in X \setminus X^\Bg} \| f^{(V)}_A(x) f^{(V)}_A(x)^\top - f_A(x) f_A(x)^\top\|_2 \\
&\leq \frac{2}{n}  \sum_{x \in X \setminus X^\Bg} \| f^{(V)}_A(x) - f_A(x)\|_2
= \frac{2}{n}  \sum_{x \in X \setminus X^\Bg} \| f^{(V^\perp)}_A(x)\|_2
\leq 4  \gamma \;,
\end{align*}
{where the last inequality follows from Claim~\ref{claim:marginS}}.

\noindent
Inequality~\eqref{deltamss3} follows since
\begin{align*}
\| M^{W,W}_{A}(X \setminus X^\Bg) &- M_{A}(X \setminus X^\Bg) \|_2
\leq \frac{1}{n} \sum_{x \in X \setminus X^\Bg} \| f^{(W)}_A(x) f^{(W)}_A(x)^\top - f_A(x) f_A(x)^\top\|_2 \\
&\leq \frac{2}{n}  \sum_{x \in X \setminus X^\Bg} \| f^{(W)}_A(x) - f_A(x)\|_2
= \frac{2}{n}  \sum_{x \in X \setminus X^\Bg} \| f^{(W^\perp)}_A(x)\|_2
\leq 2 \gamma \;,
\end{align*}
{where the last inequality follows from the definition of $X \setminus X^\Bg$}.

\noindent
Inequality~\eqref{deltamss4} follows  since
$M^{W,W}_{A}(X) \succeq M^{W,W}_{A}(X \setminus X^\Bg).$

\noindent
Inequality~\eqref{deltamss5} follows since
$\lambda_1(M^{W,W}_{A} ) = \lambda_{k+1}(M_A(X))$ and
$\gamma^2 \le \gamma \le 1$.

This completes the proof of Lemma~\ref{lem:deltamss}.
\end{proof}

Note that together Lemmas~\ref{lem:deltambb} and~\ref{lem:deltamss} show that 
$\lambda_k(M_A^{V^\perp,V^\perp}(X))-\lambda_1(M_A^{V,V}(X))$
is nearly as large as $\lambda_k(M_A(X))-\lambda_{k+1}(M_A(X))\geq \eps/(2d^3).$

Next we bound the off-diagonal terms of the transformed vectors.

\begin{lemma}\label{lem:crossterm}
We have that $\|M^{V,V^\perp}_{C}\|^2_F \leq ((1+\alpha)^3\beta^3 + (1+\alpha) \beta \delta + 2\gamma)^2$.
\end{lemma}
\begin{proof}
By the triangle inequality, we have that
$$\|M^{V,V^\perp}_{C}\|_F \leq \|M^{V,V^\perp}_{C}(X^\Bg)\|_F
+ \|M^{V,V^\perp}_{C}(X \setminus X^\Bg)\|_F \;.$$
We use this to bound the contributions from the vectors in $X^\Bg$ 
separately from the contributions from $X\setminus X^\Bg$.

For the contribution from $X^\Bg$, we note that if $x\in X^\Bg$ and $y=f_A(x)$, 
then one obtains $f_C(x)$ by multiplying the $V$-part of $y$ by $(1+\alpha)$ and rescaling slightly. 
Thus, the $V \setminus V^\perp$ component of $f_C(x)f_C(x)^\top$ is roughly 
$y^{(V)}(1+\alpha)y^{(V^\perp)T}.$ Summing over all $y\in f_A(X^\Bg)$ 
gives roughly $M_A^{V,V^\perp}(X^\Bg)$, which is small by assumption.
 
In particular, we have that
\begin{align*}
\| M^{V,V^\perp}_{C}(&X^\Bg) \|_F
\leq  \| M^{V,V^\perp}_{C}(X^\Bg) - (1+\alpha) M^{V,V^\perp}_{A}(X^\Bg) \|_F + (1+\alpha) \beta \delta\\
&= \left\| \frac 1 n \sum_{y \in f_A(X^\Bg)} \frac {(1+\alpha) y^{(V)} y^{(V^\perp)T}} {\|y^{(V^\perp)}+ (1+\alpha) y^{(V)}\|_2^2}
- (1+\alpha) \frac 1 n \sum_{y \in f_A(X^\Bg)} y^{(V)} y^{(V^\perp)T} \right\|_F +(1+\alpha) \beta \delta \\
&= \left\| \frac {(1+\alpha)} n \sum_{y \in f_A(X^\Bg)}
        \left( \frac 1 {\|y^{(V^\perp)}+ (1+\alpha) y^{(V)}\|_2^2} - 1 \right) y^{(V)} y^{(V^\perp)T} \right\|_F +(1+\alpha) \beta \delta \\
        &\le \frac {(1+\alpha)} n \sum_{y \in f_A(X^\Bg)}
        \left( 1 - \frac 1 {\|y^{(V^\perp)}+ (1+\alpha) y^{(V)}\|_2^2} \right) \|  y^{(V)} y^{(V^\perp)T} \|_F +(1+\alpha) \beta \delta \\
        &\le ( (1+\alpha) \beta )^3 +(1+\alpha) \beta \delta  \;.
    \end{align*}
The first inequality follows from Property \ref{cross term property}.
The second equality follows from the fact that
$f_{C}(x) = f_{I+\alpha I_V}(f_{A}(x))$, {by Fact~\ref{fact:transform-properties}\ref{fact:comp}},
and thus for $y = f_A(x)$ this is equal to
$\frac{y^{(V)}+ \alpha y^{(V^\perp)}}{\|y^{(V)}+ (1+\alpha) y^{(V^\perp)}\|_2}$.
The last inequality follows since
$\|  y^{(V)} {(y^{(V^\perp)})^\top} \|_F \le \beta$ and
$\|y^{(V)}+ (1+\alpha) y^{(V^\perp)}\|_2^2 \le 1+(1+\alpha)^2 \beta^2$.

To bound the contribution from $X\setminus X^\Bg$, 
we note that applying $f_{I+\alpha I_V}$ can only decrease 
the size of the $V^\perp$-component of a vector. 
Thus, for each $x\in X\setminus X^\Bg$, $f_C(x)$ 
will have a small $V^\perp$-component. In particular, we have
\begin{align*}
\| M^{V,V^\perp}_{C}({X\setminus} X^\Bg) \|_F
    &= \left\| \frac 1 n \sum_{y \in f_{C}(X \setminus X^\Bg)} y^{(V)} {(y^{(V^\perp)})^\top}  \right\|_F \\
    &\le \frac 1 n \sum_{y \in f_{C}(X \setminus X^\Bg)} \|  y^{(V)} {(y^{(V^\perp)})^\top}  \|_F \\
    &\le \frac 1 n \sum_{y \in f_{C}(X \setminus X^\Bg)} \| {(y^{(V^\perp)})^\top}  \|_2 \le 2\gamma \;,
\end{align*}
where the last inequality follows since for any
$x \in X \setminus X^\Bg$, it holds that
$\|f_{C}^{(V^\perp)}(x) \|_2 \le \|f_{A}^{(V^\perp)}(x)\|_2$,
{as follows from Fact~\ref{fact:transform-properties}\ref{fact:dom}}.

Combining this with the above proves our lemma.
\end{proof}

Finally, we need to bound $D_f$, showing that it is neither too big nor too small. 
This follows by noting that the greatest amount that any vector 
was modified is on the order of $\alpha\beta$. 
In particular, we have that:
\begin{lemma}\label{lem:massmove}
We have that
$\frac 1 n \left( \frac {(1+\alpha)^2\beta^2} {1 + (1+\alpha)^2\beta^2} - \gamma^2\right) \le D_f \le (1+\alpha)^2\beta^2 + 2 \gamma^2$, 
where
$$D_f = \frac 1 n \sum_{y \in f_{C}(X)} \|y^{(V)}\|_2^2 - \frac 1 n \sum_{y \in f_{A}(X)} \|y^{(V)}\|_2^2 \;.$$
\end{lemma}
\begin{proof}
By the definition of $\beta$, there exists $x^{\ast} \in X^\Bg$ with $\|f_A^{(V)}(x^{\ast})\|_2 = \beta$.
Such a point satisfies
$$\|f^{(V)}_{C}(x^{\ast})\|_2^2 = \| f^{(V)}_{I+\alpha I_V}(f_A(x^{\ast}))\|_2^2 = \frac {(1+\alpha)^2 \|f_A^{(V)}(x^{\ast})\|_2^2} {\|f_A^{(V^\perp)}(x^{\ast})\|_2^2 + (1+\alpha)^2\|f_A^{(V)}(x^{\ast})\|_2^2} = \frac {(1+\alpha)^2\beta^2} {1-\beta^2 + (1+\alpha)^2\beta^2}.$$
This point will contribute at least 
$$\frac 1 n \left(\|f^{(V)}_{C}(x^{\ast})\|_2^2 - \|f^{(V)}_{A}(x^{\ast})\|_2^2 \right) 
\geq \frac 1 n \left( \frac {(1+\alpha)^2\beta^2} {1 + (1+\alpha)^2\beta^2} - \gamma^2\right) $$
to $D_f$.

Moreover, for any other $x \in X^\Bg$, $\|f^{(V)}_{C}(x)\|_2^2 \le \|f^{(V)}_{C}(x^{\ast})\|_2^2$.

Thus, for points $x \in X^\Bg$, we have that
$$
0 \leq \|f_{C}^{(V)}(x)\|_2^2 - \|f_{A}^{(V)}(x)\|_2^2 \leq  {(1+\alpha)^2\beta^2} \;.$$
Finally, for all points $x \in X \setminus X^\Bg$, we have that
$0\leq \|f_{C}^{(V)}(x)\|_2^2 - \|y_{A}^{(V)}(x)\|_2^2 \le 1 - (1-2\gamma^2) \le 2\gamma^2$.
This implies the required bounds.
\end{proof}

Combining the Lemmas~\ref{lem:deltambb}, \ref{lem:deltamss}, \ref{lem:crossterm}, \ref{lem:massmove} with Lemma~\ref{lem:impr} and setting $\eta=(1+\alpha)\beta = \eps/(3d^2 n)$, we get that
\begin{align*}
&\Phi_X(A) - \Phi_X(C) \geq\\
&\ge
    2 \left(\lambda_k(M_A(X)) - 4\gamma - \lambda_{k+1}(M_A(X)) - 8 \gamma - 2\eta^2 - 4 \gamma^2 \right)
    \frac 1 n \left(\frac {\eta^2} {1 + \eta^2} - \gamma^2 \right) - 2(\eta^3 + \eta\delta + 2\gamma)^2 \\
&\ge
    \left( \frac \eps {2d^3} - 16 \gamma -2 \eta^2 \right) \left(
        \frac{\eta^2 - 2\gamma^2}{n} \right) -
        2(\eta^3 +\eta\delta + 2\gamma)^2 \;\\
& \geq \left( \frac \eps {2d^3} - 16 \gamma - 2\eta^2 \right) \left( \frac{\eta^2- 2\gamma^2}{n} \right) - 6\eta^6 - 6(\eta \delta)^2 - 24\gamma^2.
\end{align*}
Given that both $\gamma$ and $\eta^2$ are less than a sufficiently small multiple of $\frac{\eps}{d^3}$, the above is at least
$$
\frac{\eps\eta^2}{3d^3 n} - 6\eta^6 - 6(\eta \delta)^2 - 24\gamma^2.
$$
Given that $\delta$ is less than a sufficiently small multiple of $\frac{\eps}{dn}$, and $\gamma$ a small multiple of $\eps^2/(d^4n^2)$, this is
$$
\Omega( \eps \eta^2/(d^3n )) = \Omega(\eps^3/(d^7n^3)).
$$

\paragraph{The case of $\beta=0$:}
We now argue that in the case that $\beta = 0$,
no Forster transform exists, as
the algorithm correctly identifies a subspace
of dimension $k$ containing more than a $k/d$
fraction of the points of $X$.

Since we have that $\lambda_k(M_A(X)) - \lambda_{k+1}(M_A(X)) \ge \eps /(2d^3)$ by assumption,
either $\lambda_k(M_A(X)) \ge \frac 1 d + \frac {\eps} {4 d^3}$
or $\lambda_{k+1}(M_A(X)) \le \frac{1}{d} - \frac {\eps} {4 d^3}$.
The algorithm returns the subspace $V^\perp$ of dimension $k$,
which contains all points $X^\Bg$ as $\beta = 0$. 
We claim that $|X^\Bg|/n > k/d$, which would complete our analysis. 
This is essentially because the large eigenvalues of $M_A(X)$ 
on $V^\perp$ imply that $X^\Bg$ must have many points.

We consider two subcases below.
\begin{itemize}[leftmargin=*]
    \item[] Case 1: $\lambda_k(M_A(X)) \ge \frac 1 d + \frac {\eps} {4 d^3}$.
    In this case, we have that
    \begin{align*}
    \frac {|X^\Bg|} {n} &= \tr( M_A(X^\Bg) ) \geq  k \, \lambda_k ( M_A(X^\Bg) ) \\
    & \geq k\lambda_k( M_A^{V^\perp,V^\perp}(X^\Bg) )- k\delta\\
    & \geq k\lambda_k( M_A^{V^\perp,V^\perp}(X) )- k\delta-2k\gamma^2\\
    & \geq k\lambda_k(M_A(X))-7k\gamma\\
    & \geq k/d +k( \eps/(4d^3) - 7\gamma) > k/d \;,
    \end{align*}
where the second line above follows from Property \ref{VT eigenvalue property}, 
the third line from Claim \ref{claim:marginS}. 
The fourth line follows from Lemma \ref{lem:deltambb}, and the rest from $\eps/d^3 \gg \gamma > \delta$.

    \item[]Case 2: $\lambda_{k+1}(M_A(X)) \le \frac{1}{d} - \frac {\eps} {4 d^3}$.
    In this case,  we have that
    $$
    \frac{|X^S|}{n} = \tr(M_A(X^S)) = \tr(M_A^{V,V}(X^S)) + \tr(M_A^{V^\perp,V^\perp}(X^S)).
    $$
    By Claim \ref{claim:marginS} we have that $\tr(M_A^{V^\perp,V^\perp}(X^S)) \leq 2\gamma^2$. On the other hand since $\beta=0$, all elements of $X^\Bg$ are orthogonal to $V$ and thus
    $$
    \tr(M_A^{V,V}(X^S)) = \tr(M_A^{V,V}(X)).
    $$
    This is at most $(k-d)\lambda_1(M_A^{V,V}(A))$, 
    which by Lemma \ref{lem:deltamss} is at most $(k-d)\lambda_{k+1}(M_A(X))+8d\gamma$. 
    Combining with the above, we get that
    $$\frac{|X^S|}{n} \leq (k-d)\lambda_{k+1}(M_A(X))+10d\gamma \leq (k-d)/d - (k-d)\eps/(4d^3) + 10d\gamma < (k-d)/d \;.$$
    Hence in this case as well $|X^\Bg|/n = 1-|X^S|/n > k/d.$
\end{itemize}
This completes the proof of Proposition \ref{case-2-prop}.

{This completes the proof of Proposition~\ref{prop:main-iteration}. \qed}


\section{Approximate Eigendecomposition in Strongly Polynomial time} \label{sec:eigen}

In this section, we give a simple algorithm that computes an approximate eigendecomposition
with multiplicative error guarantees in strongly polynomial time.

\begin{proposition}\label{eigendecomposition proposition} \label{prop:svd}
Given a $d \times d$ PSD matrix $M$, an accuracy parameter $\eps > 0$ 
and a failure probability $\delta >0$, there is an algorithm that computes 
orthogonal vectors $q_1, \ldots, q_d$ and scalars $a_i$ 
such that the matrix $\hat M = \sum_{i=1}^d a_i q_i q_i^\top$ satisfies
the following: for all $v \in \R^d$, it holds that
$$|v^\top (M-\hat M) v| \le \eps \,  (v^\top M v) \;.$$
The algorithm performs $\poly(d/\eps,\log(1/\delta ))$ arithmetic operations 
on $\poly(d/\eps,\log(1/\delta),b)$-bit numbers, 
where $b$ is the bit complexity of the entries of $M$.
\end{proposition}

\begin{proof}
We assume throughout that $d$ is sufficiently large and $\eps$ sufficiently small.

Our algorithm is based on the power method. 
In more detail, by taking a large power of $M$ times a random vector, 
we can obtain an approximation of the principal eigenvalue. 
Taking a large power of $M$ times another random vector and projecting 
onto the orthogonal complement of the first, gives us an approximation 
to the second eigenvector. Repeating this process, we obtain 
an approximation to the full eigendecomposition. Unfortunately, 
we cannot quite hope to learn the eigenvectors themselves in general. 
Specifically, in the case where some of the eigenvalues are close, 
we would require very large powers of $M$ in order to distinguish them. 
However, in this case, it is not necessary to learn the eigenvalues exactly 
in order to obtain a good approximation. 

It is well-known that the aforementioned standard approach 
can be used to get an approximate eigendecomposition 
such that $\|M-\hat M\|_2 \leq \eps \|M\|_2$. 
Unfortunately, this guarantee is weaker than the result that we require. 
For our application, we require that if $M$ has a large eigenvalue gap somewhere, 
then $\hat M$ very precisely finds this gap. Fortunately, if there is a large eigenvalue gap, 
this makes the power method that much stronger. Indeed, multiplying by a suitable power of $M$
will cause the components of the small eigenvector 
to shrink by an amount proportional to a power of the gap size.

Our algorithm is presented in pseudocode below.

\begin{algorithm}[hbt!]
   \caption{Computing the approximate eigendecomposition of a matrix $M$}
   \label{alg:eigen}
\begin{algorithmic}[1]
    \Function{$\textsc{EigenDecomposition}$}{Matrix $M_{d\times d}$, accuracy parameter $\eps$, error probability $\delta$}
    \State Let $A$ be a random $d\times d$ matrix where the entries are i.i.d.\ uniform samples 
    from $\{1,2,\ldots,N\}$, for $N$ at least a sufficiently large constant multiple of $d/\delta$.
    \State Let $w_1,w_2,\ldots,w_d$ be the column vectors of $M^t A$, 
    for $t$ a sufficiently large constant multiple of $d^6/\eps^2 \log(d/\delta)$.
    \For{$i=1$ to $d$}
         \State Let $q_i$ be the projection of $w_i$ onto the orthogonal complement of $w_1,w_2,\ldots,w_{i-1}$.
         \State Let $a_i=0$ if $q_i=0$ and $a_i = q_i^\top M q_i/(q_i \cdot q_i)$ otherwise.
    \EndFor
    \State \Return $\{a_i,q_i\}$
    \EndFunction
  \end{algorithmic}
\end{algorithm}

It is easy to see that this algorithm runs in the appropriate time and bit-complexity bounds. 
The difficulty is in showing that the resulting $\hat M$ satisfies the desired error bounds. 
We begin by giving this analysis under the assumption that $M$ is non-singular.

We start by noting that if we take $q_i'$ to be the normalization of $q_i$ 
and take $a_i'$ to be $(q_i')^\top M (q_i')$, 
then we have that $a_i' (q_i')(q_i')^\top = a_i q_i q_i^\top$, 
leading to the same matrix $\hat M$. 
(Note that we cannot use $a_i'$ and $q_i'$ in our algorithm 
only because normalizing $q_i$ requires taking square roots; 
an operation that is not efficiently implementable in our model). 
We note that the $q_i'$ are obtained from $w_i$ by applying Gram-Schmidt. 
From here on, we will consider the equivalent algorithm, 
where the $q_i$ are obtained from the $w_i$ by applying Gram-Schmidt.

Let the eigendecomposition of $M$ be given by $M=\sum_{i=1}^d \lambda_i v_i v_i^\top$, 
where the $v_i$ are an orthonormal basis and 
where $\lambda_1\geq \lambda_2\geq \ldots \geq \lambda_d \geq 0$. 
Let $\eta = \eps^2/d^3$. 

We say that two consecutive eigenvectors $v_i$ and $v_{i+1}$ are in the same block 
if $\lambda_{i+1} \geq \lambda_i/(1+\eta)$. We say that $\lambda_i$ and $\lambda_j$ 
are in the same block for $i\leq j$, 
if $\lambda_k$ and $\lambda_{k+1}$ are in the same block for all $i\leq k < j$. 
Note that if $\lambda_i$ and $\lambda_j$ are in the same block, 
their ratio is at most $(1+\eta)^d$; 
and that if they are in different blocks, their ratio is at least $(1+\eta)$.

Let $\alpha_1,\ldots,\alpha_d$ be the columns of $A$. 
Let $\beta_i$ be the unique vector in $\spn(\alpha_1,\alpha_2,\ldots,\alpha_i)$ 
such that $v_j \cdot \beta_i = 0$ if $j<i$, 
and $v_i\cdot \beta_i = 1$. We note that if $B$ is the matrix with columns $\beta_i$, 
applying Gram-Schmidt to the columns of $M^t B$ yields the same result 
as applying it to the columns of $M^t A$. We will need the following claim:

\begin{claim}\label{non singularity claim}
With probability at least $1-\delta$ over the choice of $A$, 
we have that for all $1\leq i \leq d$ it holds that 
$\|\beta_i\|_2 \leq O(d^2/\delta)^i$.
\end{claim}
\begin{proof}
We note that $\beta_m = \sum_{i=1}^m t_i \alpha_i$, 
so that $v_j \cdot \left(\sum_{i=1}^m t_i \alpha_i \right) = \sum_{i=1}^m t_i (v_j\cdot \alpha_i)$ 
is equal to $0$ for $j<m$, and equal to $1$ for $j=m$. 
In particular, the $t_i$ form the unique vector such that 
$D_m [t_1,t_2,\ldots,t_m]^\top = [0,0,\ldots,0,1]^\top$, 
where $D_m$ is the $m\times m$ matrix defined by 
$(D_m)_{j,i}=[v_j \cdot \alpha_i]_{1\leq i,j\leq m}$. 
Using Cramer's rule, we find that each $t_i$ is a sub-determinant of $D_m$ 
divided by $\det(D_m)$. Since the sub-determinant has absolute value 
at most the product of the norms of its rows or 
$O(N\sqrt{m})^{m-1} = O(N\sqrt{d})^m$, 
it suffices to show that $|\det(D_m)| = \Omega(N\delta/d^{3/2})^m$.

In particular, we will prove that with probability at least $1-\delta$ over the choice of $A$ 
we have that $|\det(D_k)| = \Omega(N \delta/(d^{3/2}))^k$ for all $k$. 
In fact we show that conditioning on the values of $\alpha_1,\alpha_2,\ldots,\alpha_{m-1}$, 
the probability that $|\det(D_m)| > \Omega(N\delta/d^{3/2})|\det(D_{m-1})|$ 
is at least $1-\delta/d$. If this holds for all $m$, our result will follow.

To show this, we note that conditioning on $\alpha_1,\alpha_2,\ldots,\alpha_{m-1}$ 
fixes all but the last column of $D_m$, which is linear in $\alpha_m$. 
The determinant of $D_m$ equals a sum over the last column of the relevant entry 
times an appropriate sub-determinant. In particular, 
$\det(D_m) = \sum_{i=1}^m C_i (v_i \cdot \alpha_m) = u \cdot \alpha_m$, 
where $v = \sum_{i=1}^m C_i v_i$ and $C_i$ are the appropriate sub-determinants. 
Note that $v_m\cdot u = C_m = \det(D_{m-1})$. Therefore, $|u| \geq |\det(D_{m-1})|$. 
This implies that $u$ must have an entry of size at least $|\det(D_{m-1})|/\sqrt{d}$. 
Fixing all other entries of $\alpha_m$, we note that there is a probability 
of at least $1-(\delta/d)$ that $|u\cdot \alpha_m| \gg N \delta/d^{3/2}$. 
If this holds, then $|\det(D_m)| \geq \Omega(N\delta/d^{3/2}) |\det(D_{m-1})|$, as desired.
This completes our proof.
\end{proof}

We will show that so long as the conclusion of Claim \ref{non singularity claim} holds, 
our algorithm will produce an appropriate error guarantee.

We begin by defining
$$
\gamma_i := M^t \beta_i / \lambda_i^t \;,
$$
where $\gamma_i=0$ if $\lambda_i=0$. 
Note that in this case $\lambda_j$ will be $0$ for all $j>i$, 
and thus $M^t \beta_i = 0$. 
Applying Gram-Schmidt to $M^t \beta_i$ or to $\gamma_i$ gives the same result, 
and thus $q_i$ can be thought of as the result of applying Gram-Schmidt to the $\gamma_i$. 
Furthermore, we note that
$$
v_j \cdot \gamma_i = 
\begin{cases}
0 &\textrm{, if } j < i \\ 
1 &\textrm{, if } j = i \\ 
O(d^2/\delta)^d (\lambda_j/\lambda_i)^t & \textrm{, otherwise}. 
\end{cases}
$$
This implies that if $\lambda_i$ and $\lambda_j$ are in different blocks, then
$$
|v_j \cdot \gamma_i| \leq s \, \min\left(\frac{\lambda_j}{\lambda_i},\frac{\lambda_i}{\lambda_j} \right) \;,
$$
where
$$
s = (d/\delta)^{-3d^3} \;.
$$

We require the following claim.

\begin{claim}\label{off block entry claim}
Letting $q_i$ be obtained from $\gamma_i$ by Gram-Schmidt, 
we have that if $1\leq m \leq d$ and if $\lambda_i$ and $\lambda_m$ are in different blocks, then
$$
|v_i \cdot q_m| \leq s \, \min\left(\frac{\lambda_m}{\lambda_i},\frac{\lambda_i}{\lambda_m} \right) M^m \;,
$$
where $M$ is $(C d^2/\delta)^{d+d^2}$, for a sufficiently large universal constant $C>0$.
\end{claim}
\begin{proof}
We proceed by induction on $m$. 
The case of $m=1$ follows immediately from the above observation 
and the fact that $q_1 = \gamma_1/\|\gamma_1\|_2$ 
and that $\|\gamma_1\|_2 \geq v_1 \cdot \gamma_1 = 1$.

For the inductive step, we begin by assuming that our claim is true 
for all smaller values of $m$. We note that $q_m = r_m/\|r_m\|_2$, where
$$
r_m := \gamma_m - \sum_{j=1}^{m-1} (\gamma_m \cdot q_j) q_j \;.
$$
We note that if $m > j$ and $\lambda_m$ and $\lambda_j$ 
are in different blocks, then
\begin{align*}
|\gamma_m \cdot q_j| & \leq \sum_{k=1}^d |v_k \cdot \gamma_m| |v_k \cdot q_j|\\
& \leq d \, s \, \max_k \min\left(\frac{\lambda_m}{\lambda_k},\frac{\lambda_k}{\lambda_m} \right) 
\min\left(\frac{\lambda_j}{\lambda_k},\frac{\lambda_k}{\lambda_j} \right) \, O(d^2/\delta)^d M^{m-1}\\
& \leq s \, O(d^2/\delta)^d \min\left(\frac{\lambda_m}{\lambda_j},\frac{\lambda_j}{\lambda_m} \right) M^{m-1} \;,
\end{align*}
where the second line above follows from the inductive hypothesis, 
the fact that $\lambda_k$ cannot be in the same block 
as both $\lambda_m$ and $\lambda_j$, 
and the fact that $\|\gamma_m\|_2 \leq \|\beta_m\|_2 = O(d^2/\delta)^d$.

From this we conclude that if $\lambda_m$ and $\lambda_k$ are in different blocks, then
\begin{align*}
|r_m \cdot v_k| & \leq |\gamma_m\cdot v_k| + \sum_{j=1}^{m-1} |q_j \cdot v_k| |q_j \cdot \gamma_m| \;.
\end{align*}
We know that
$$
|\gamma_m\cdot v_k| \leq s \min\left(\frac{\lambda_m}{\lambda_k},\frac{\lambda_k}{\lambda_m} \right) \;.
$$
For $\lambda_j$ in the same block as $\lambda_m$, 
we have that $|q_j \cdot v_k| |q_j \cdot \gamma_m|$ is at most
$$
|\gamma_m| |q_j \cdot v_k| \leq 
s \, O(d^2/\delta)^d \min\left(\frac{\lambda_m}{\lambda_k},\frac{\lambda_k}{\lambda_m} \right) M^{m-1} \;.
$$
For $\lambda_j$ and $\lambda_m$ in different blocks, 
using the above bound on $|\gamma_m \cdot q_j|$, it is at most
$$
s \, O(d^2/\delta)^d \min\left(\frac{\lambda_m}{\lambda_j},\frac{\lambda_j}{\lambda_m} \right) 
\min\left(\frac{\lambda_k}{\lambda_j},\frac{\lambda_j}{\lambda_k} \right) M^{m-1} \;,
$$
which means that
$$
|q_j \cdot v_k| |q_j \cdot \gamma_m| \leq 
s \, O(d^2/\delta)^d  \, \min\left(\frac{\lambda_m}{\lambda_k},\frac{\lambda_k}{\lambda_m} \right) M^{m-1} \;.
$$
Summing over $j$, we find that for $\lambda_m$ and $\lambda_k$ in different blocks, 
we have that
\begin{equation}\label{r bound equation}
|r_m \cdot v_k| \leq s \, O(d^2/\delta)^d \min\left(\frac{\lambda_m}{\lambda_k},\frac{\lambda_k}{\lambda_m} \right) M^{m-1} \;.
\end{equation}
We now just need to show that $\|r_m\|_2$ is not too small. 
To achieve this, we will show that $q_m$ has a reasonably large projection 
onto the space orthogonal to $q_1,\ldots,q_{m-1}$. 
To show this, let $\ell$ be the smallest number such that 
$\lambda_\ell$ and $\lambda_m$ are in the same block. 
For $\ell \leq i \leq m$, let $r_i'$ denote the projection of $\gamma_i$ 
onto the space orthogonal to $q_1,q_2,\ldots,q_{\ell-1}$. Namely, we define
$$
r_i' := \gamma_i - \sum_{j=1}^{\ell-1} (\gamma_i \cdot q_j) q_j \;.
$$
We note that since each 
$|\gamma_i \cdot q_j| \leq s \, O(d^2/\delta)^d M^d$, 
we have that $\|r_i' -\gamma_i\|_2 \leq s \, O(d^2/\delta)^d M^d$. 
We note that $r_m$ is the component of $r_m'$ orthogonal to 
$r'_\ell,\ldots, r'_{m-1}$. This equals the ratio of the volumes 
of the parallelepiped with sides $r'_\ell,\ldots, r'_{m}$ 
to the volume of the one with sides $r'_\ell,\ldots, r'_{m-1}$. 
The latter volume is at most 
$\|r'_\ell\|_2 \, \|r'_{\ell+1}\|_2 \cdots \|r'_{m-1}\|_2 \leq O(d^2/\delta)^{d^2}.$ 
We can bound the former from below by the determinant of the matrix 
with entries $v_i\cdot r'_j$, for $\ell \leq i,j \leq m$. However, we note that
$$
|v_i \cdot r_j' - v_i\cdot r_j| \leq \|r_j - r_j'\|_2 \leq s \, O(d^2/\delta)^d M^d \;.
$$
On the other hand, the matrix with entries $v_i\cdot r_j$ 
is a lower diagonal matrix with $1$'s on the diagonal 
and entries of size at most $O(d^2/\delta)^d$. 
This matrix has determinant $1$, 
and the difference between its determinant and that of the matrix 
with entries $v_i\cdot r_j'$ is at most 
$O(d^2/\delta)^{d^2} \, s \, M^d \leq 1/2$. 
Thus, the matrix with entries $v_i \cdot r_j'$ has determinant at least $1/2$.
This implies that $\|r_m\|_{2} \geq O(d^2/\delta)^{-d^2}$. 
Combining this with Equation~\eqref{r bound equation} yields
$$
|q_m \cdot v_k| \leq s \, O(d^2/\delta)^{d+d^2} 
\min\left(\frac{\lambda_m}{\lambda_k},\frac{\lambda_k}{\lambda_m} \right) M^{m-1} 
\leq s \, \min\left(\frac{\lambda_m}{\lambda_k},\frac{\lambda_k}{\lambda_m} \right) M^m \;,
$$
whenever $\lambda_m$ and $\lambda_k$ are in different blocks. 
This completes our inductive step.
\end{proof}

\noindent Claim \ref{off block entry claim} implies that if $\lambda_i$ and $\lambda_j$ are in different blocks 
we have that
$$
|v_i \cdot q_j| \leq \min\left(\frac{\lambda_m}{\lambda_i},\frac{\lambda_i}{\lambda_m} \right)(\eps/d)^3 \;.
$$
We now will try to understand the size of the $a_i$'s. 
We have the following sequence of (in)equalities:
\begin{align*}
a_i & = q_i^\top M q_i \\
& = \sum_{j=1}^d \lambda_j |v_j \cdot q_i|^2\\
& = \sum_{\lambda_j\textrm{ in same block as }\lambda_i} \lambda_i(1+O(d\eta))|v_j \cdot q_i|^2 + O\left(\sum_{\lambda_j\textrm{ in different block from }\lambda_i}(\eps/d)^6 \lambda_j (\lambda_i/\lambda_j)\right)\\
& = \lambda_i O(d\eta+\eps^6/d^6) + \lambda_i\sum_{\lambda_j\textrm{ in same block as }\lambda_i} |v_j \cdot q_i|^2\\
& = \lambda_i O(d\eta+\eps^6/d^6) + \lambda_i - \lambda_i\sum_{\lambda_j\textrm{ not in same block as }\lambda_i} |v_j \cdot q_i|^2\\
& = \lambda_i ( 1 + O(\eps^2/d^2)) \;,
\end{align*}
where we used the fact that the $v_j$'s form an orthonormal basis, 
and therefore $1=\|q_i\|_2^2 = \sum_j |v_j \cdot q_i|^2$.

Our result will now follow from the proceeding claim:
\begin{claim}
For any $1\leq i,j\leq d$, we have that
$
|v_i^\top (M - \hat M) v_j| < (\eps/d^2) \sqrt{\lambda_i \lambda_j} \;.
$
\end{claim}
\begin{proof}
We begin with the case where $\lambda_i$ and $\lambda_j$ are not in the same block. 
We have that $v_i^\top M v_j = 0$ and that
$$
|v_i^\top \hat M v_j| \leq \sum_{k=1}^d a_k |v_i\cdot q_k||v_j \cdot q_k| 
= \sum_{k=1}^d O(\lambda_k)(\eps/d)^3 
\min\left(\frac{\lambda_k}{\lambda_i},\frac{\lambda_i}{\lambda_k} \right) 
\min\left(\frac{\lambda_k}{\lambda_j},\frac{\lambda_j}{\lambda_k} \right) \;.
$$
This in turn is at most
$$
\sum_{k=1}^d O(\lambda_k)(\eps/d)^3 (\sqrt{\lambda_i/\lambda_k})(\sqrt{\lambda_j/\lambda_k}) 
< (\eps/d^2)\sqrt{\lambda_i \lambda_j} \;.
$$
If $\lambda_i$ and $\lambda_j$ are in the same block, then we have that
$$
v_i^\top \hat M v_j = \sum_{k=1}^d a_k (v_i\cdot q_k)(v_j \cdot q_k) \;.
$$
The contribution from $\lambda_k$ not in the same block 
is once again $O(\lambda_i \eps/d^2)$. 
The contribution from $\lambda_k$ in the same block 
can be bounded above as follows:
\begin{align*}
\sum_{\lambda_k\textrm{ in the same block as }\lambda_i} a_k (v_i\cdot q_k)(v_j \cdot q_k) 
&= \sum_{\lambda_k\textrm{ in the same block as }\lambda_i} \lambda_i(1+O(\eps^2/d^2)) (v_i\cdot q_k)(v_j \cdot q_k)\\
& = O(\lambda_i \eps^2/d^2) + \lambda_i \sum_{\lambda_k\textrm{ in the same block as }\lambda_i}(v_i\cdot q_k)(v_j \cdot q_k)\\
& = O(\lambda_i \eps^2/d^2) + \lambda_i (v_i \cdot v_j) - \lambda_i \sum_{\lambda_k\textrm{ not in the same block as }\lambda_i}(v_i\cdot q_k)(v_j \cdot q_k)\\
& = O(\sqrt{\lambda_i\lambda_j} \eps^2/d^2) + \lambda_i \delta_{i,j}\\
& = O(\sqrt{\lambda_i\lambda_j} \eps^2/d^2) + v_i^\top M v_j \;.
\end{align*}
This completes the proof of the claim.
\end{proof}

\noindent To complete our analysis, let $v = \sum_{i=1}^d c_i v_i$. 
Then, we have that $v^\top M v = \sum_{i=1}^d c_i^2 \lambda_i \geq \max_i(|c_i| \sqrt{\lambda_i})^2$. 
On the other hand, we have that
$$
|v^\top(M-\hat M)v| \leq \sum_{i,j=1}^d |c_i| \, |c_j| \, |v_i^\top (M-\hat M) v_j| 
\leq \sum_{i,j=1}^d (\eps/d^2) \, |c_i| \, |c_j| \, \sqrt{\lambda_i \lambda_j} 
\leq \eps \max_i(|c_i| \sqrt{\lambda_i})^2 \leq \eps (v^\top M v) \;.
$$
This completes our analysis in the case where $M$ is non-singular. 
When $M$ is singular and rank $k$, 
then assuming that the conclusion of Claim \ref{non singularity claim} holds, 
we note that applying the same analysis to the vectors 
$M^t \alpha_1,\ldots, M^t \alpha_k$ 
as elements of the $k$-dimensional vector space $\mathrm{Image}(M)$, 
we get the desired result.
\end{proof}

\section{Matrix Rounding} \label{sec:rounding}

In this section, we establish our efficient rounding procedure, establishing the following:

\begin{theorem}[Matrix Rounding]\label{thm:rounding}
There is an algorithm that given 
(i) a set of $n$ points $X \subseteq \{-2^b, \ldots, 2^b\}^d \setminus \{\0\}$
with $b \in \Z_+$, so that $X$ spans $\R^d$,
(ii) a full-rank $d\times d$ matrix $A \in \{-2^{rb}, \ldots, 2^{rb}\}^{d\times d}$, with $r \in \Z_+$,
and (iii) an accuracy parameter $\eps \in (0,1)$,
outputs a matrix $A'$ with integer entries of magnitude at most $(\frac d \eps)^{O(d^3b)}$ such that
for all points $x \in X$ it holds $\|f_A(x) - f_{A'}(x)\|_2 \le \eps$.
The algorithm performs $\poly(d, n, r)$ arithmetic operations on $\poly(d, n, r, b, \log(1/\eps))$-bit numbers.
\end{theorem}

This theorem will allow us to avoid having the matrices $A$ in our main algorithm 
blow up in bit complexity since every round we can replace $A$ by $A'$ 
to reduce the bit complexity with at most a small loss of potential.

\medskip

\noindent {\bf Notation} For a matrix $A$ and a subspace $W$,
we let $A^{(W)} = A I_W$ be the matrix whose $i$-th row is the $i$-th row
of $A$ projected onto the subspace $W$.
We let $\sigma_{\max}(A)$ to be the maximum singular value of a matrix $A$ and
$\sigma_{\min}(A)$ to be the minimum non-negative singular value.

We also use $\nearest{x}$ to denote the integer closest to the real number $x$, 
and use $\nearest{A}$ to denote the matrix obtained by applying $\nearest{\cdot}$ 
to each entry of the matrix $A$.

The rounding process is presented in Algorithm~\ref{alg:forsterrounding}.
The main idea of the algorithm is to iteratively reduce the condition
number of matrix $A$ without significantly affecting the transformation on the pointset $X$.
To achieve this, the algorithm identifies a subspace $V$ such that $V$ and $V^\perp$ have
large multiplicative singular value gap $\sigma_{\min} (A^{(V^\perp)}) / \sigma_{\max} (A^{(V)})$.
It then aims to rescale the subspace $V^\perp$ so that the condition number decreases.
As this could significantly affect the transformation for some points in $X$,
it rescales instead a different subspace $R$ that is very close to $V^\perp$,
but at the same time leaves unaltered the transformation 
for the set of problematic points in $X$ lying in a subspace $W$.
Applying this technique iteratively, we reduce to the case where $A$ 
has bounded condition number. At this point, 
we can simply replace $A$ by the matrix obtained by an appropriate rounding of $A$'s entries.

\begin{algorithm}[hbt!]
   \caption{Rounding the matrix}
   \label{alg:forsterrounding}
\begin{algorithmic}[1]
    \Function{$\textsc{Round}$}{Matrix $A\in \{-2^{rb}, \ldots, 2^{rb}\}^{d\times d}$, Set $X$ of $n$ points in $\{0,\ldots,2^b\}^d$, accuracy parameter $\eps$}
    \State Let $N \leftarrow (\frac d \eps)^{O(d^3b)}$

    		\While{not terminated}
        \State Compute estimates $\bar \sigma_1$ and $\bar \sigma_d$ such that $\sigma_1(A) \le \bar \sigma_1 \le 2 \sigma_1(A)$ and $\frac 1 2 \sigma_d(A) \le \bar \sigma_d \le \sigma_d(A)$.
        \If{the condition number $\bar \sigma_1(A) / \bar \sigma_d(A) \ge N$} \label{step:return} \Return $\lceil \frac d \eps A / \bar \sigma_d \rfloor$,
        \EndIf
        \State Round the entries of $A$ setting $A \leftarrow \lceil \frac { 2^{O(rdb)} } \eps A / \bar \sigma_d \rfloor$
        \State Using approximate eigendecomposition find a subspace $V$ and a parameter $G$, such that
                $$
                \frac 1 2 \max_{1\leq i\leq d-1}  \frac {\sigma_{i} (A)} { \sigma_{i+1} (A) } \le G \le \frac {\sigma_{\min} (A^{(V^\perp)})} {\sigma_{\max} (A^{(V)})} \le \max_{1\leq i\leq d-1}  \frac {\sigma_{i} (A)} { \sigma_{i+1} (A) }. $$
                
        \State Obtain $(w_i,p_i)_{i=1}^d$ by running $\textsc{EigendecompositionFromSet}(A,X)$

		        \State Find the smallest index $m \ge 0$ such that $p_{m+1} \ge \frac 1 d G^{(m+1)/d} \sigma_{\max} (A^{(V)})$\label{m-line}
        \State Let $W \leftarrow \spn \{ w_1,\ldots, w_m\}$,  $g \leftarrow \frac { \min \{p_{m+1}, \sigma_{\min}(A^{(V^\perp)}) \}}{ \max \{\sigma_{\max}(A^{(W)}), \sigma_{\max}(A^{(V)}) \}}$
        \State Consider the subspace $R \leftarrow \spn(W\cup V^\perp) \cap W^\perp$,  
        \State \,\,\, and note that by Claim~\ref{claim:subspace partition} it holds that $I = I_R + I_W + I_{W^\perp \cap V}$
		    \State Define a matrix $T$ that rescales the subspace $R$ by $\delta = \Theta(2^{-b})
        $ ,
        \State \,\,\, i.e. $T = \delta I_{R} + I_W + I_{W^\perp \cap V}$
		    \State Set $A \leftarrow A T$
		     	\EndWhile
	    \EndFunction
  \end{algorithmic}
\end{algorithm}

\begin{algorithm}[hbt!]
  \caption{Eigendecomposition with respect to a Set}
  \label{alg:eigendecompositionwrtset}
\begin{algorithmic}[1]
   \Function{$\textsc{EigendecompositionFromSet}$}{Matrix $A$, Set $X$ of $n$ points}
      \State Let $w_1 = \arg\min_{x \in X} \frac {\|A x\|_2} {\|x\|_2}$
      and $p_1 = \min_{x \in X} \frac {\|A x\|_2} {\|x\|_2}$
         \For{$i=2$ to $d$}
           \State Set $W_{i-1} = \spn\{w_1,\ldots, w_{i-1}\}$
           \State Set $w_i = \arg\min_{x \in X \setminus W_{i-1} } \frac {\|A x^{(W_{i-1}^\perp)}\|_2} {\| x^{(W_{i-1}^\perp)}\|_2}$
           \State Set $p_i = \min_{x \in X \setminus W_{i-1} } \frac {\|A x^{(W_{i-1}^\perp)}\|_2} {\| x^{(W_{i-1}^\perp)}\|_2}$
         \EndFor
         \State \Return $(w_i,p_i)_{i=1}^d$
     \EndFunction
 \end{algorithmic}
\end{algorithm}

This works by applying an iterative process to decrease the condition number of $A$ 
followed by appropriate rounding. In this iteration, we take $V$ to be a subspace 
spanned by the small singular vectors of $A$ right before a large eigenvalue gap. 
We then define a sequence of vectors $w_1,w_2,\ldots,w_n$, 
where $w_i$ is (roughly) the element of $X$ not in the span of $w_1,w_2,\ldots,w_{i-1}$ 
with $\|A w_i\|_2$ minimal; and take $W$ to be the span of $w_1,w_2,\ldots,w_m$, 
where $\|A w_{m+1}\|_2$ is substantially larger than any of the previous values. 
This gives us a subspace $W$ where any $x\in X$ is either in $W$ or has $\|Ax\|_2$ 
substantially larger than the corresponding value for any small element of $W$. 
We then replace $A$ by $AT$, where $T$ acts as the identity on $W$ 
and $W^\perp \cap V$, but shrinks things considerably in orthogonal directions. 
As elements in $x\in W$ are preserved by $T$ and $x\in X \setminus W$ 
have most of their contributions to $Ax$ coming from parts orthogonal to $W$ and $V$, 
we show (see Proposition \ref{prop:approxtransform}) that $f_A(x)\approx f_{AT}(x)$ for all $x\in X$, 
and thus this operation does not substantially change the potential. 
Furthermore, since $T$ scales down the directions orthogonal to $V$ 
(which are the large singular directions of $A$), 
we show (see Proposition \ref{prop:conditionnumber}) that the condition number of $AT$ 
is substantially smaller than the condition number of $A$, which implies that 
repeating this process enough times will eventually terminate with a matrix 
with not-too-large condition number. Finally, in Proposition \ref{prop:rounding}, 
we show that once we have reduced the condition number, 
rounding the appropriate matrix entries will not substantially 
change the Forster transform on any given vector.

The analysis of Algorithm~\ref{alg:forsterrounding} is based on three key propositions.

\begin{proposition}\label{prop:rounding}
Let $A$ be any matrix with smallest singular value $\sigma_d(A) > 0$,
Given $\eps \in (0,1)$ and $\bar \sigma$ such that $\frac 1 2 \sigma_d(A) \le \bar \sigma \le \sigma_d(A)$, the integer matrix $\hat A = \lceil \frac d {\bar \sigma \eps} A \rfloor$ satisfies
\begin{itemize}
  \item[-] $\kappa(\hat A) \le \kappa(A) (1+4 \eps)$,
  \item[-] for all $x \in \R^d$, $\|f_A(x) - f_{\hat A}(x)\|_2 \le 2 \eps$, and
  \item[-] the entries of $\hat A$ have magnitude $O(d \kappa(A) / \eps).$
\end{itemize}
\end{proposition}

Proposition~\ref{prop:rounding} shows that so long as $A$ has bounded condition number, 
letting $A'$ be the rounding of an appropriate multiple of $A$ yields a bounded precision matrix 
that nearly preserves all of the transformed points.

\begin{proposition}\label{prop:conditionnumber}
  Let $A \in \R^{d\times d}$ be a full-rank matrix. Let $V, W$ be subspaces of $\R^d$ so that
  $$\sigma_{\min}(A^{(V^\perp)})  \ge g \, { \max \{\sigma_{\max}(A^{(W)}), \sigma_{\max}(A^{(V)}) \}}$$
  for some $g > 10$.
  Define $T = \delta I_{R} + I_{R^\perp}$ for $\delta \ge 8 g^{-1}$ and $R = \spn(W\cup V^\perp) \cap W^\perp$.
  It holds that
  $$\kappa(AT) \le 30 \delta \kappa(A) \;.$$
\end{proposition}

Proposition~\ref{prop:conditionnumber} shows that after every iteration 
the condition number of the matrix $A$ is significantly reduced by a factor of $O(\delta)$.

\begin{proposition}\label{prop:approxtransform}
Let $A \in \R^{d\times d}$ be a full-rank matrix and $X$ be a set of points. 
Let $V, W$ be subspaces of $\R^d$ so that 
$$\frac  { \min_{x \in X \setminus W} \|A x^{(W^\perp)}\|_2 / \|x^{(W^\perp)}\|_2  } { \max \{ \sigma_{\max} (A^{(W)}), \sigma_{\max} (A^{(V)}) \} } \ge g
\text{ and }
\min_{x \in X \setminus W} \|x^{(W^\perp)}\|_2 \ge \rho$$ 
for some $\rho \in (0,1)$ and $g > 1$.
Define $T = \delta I_{R} + I_{R^\perp}$ for $\delta \in (0,1)$ and $R = \spn(W\cup V^\perp) \cap W^\perp$.
For all $x \in X$, it holds that
$$ \| f_{A}(x) - f_{AT}(x) \|_2 \le \frac {16} {(g-1)\rho\delta} \;. $$
\end{proposition}

Proposition~\ref{prop:approxtransform} shows that for every iteration of the algorithm, 
the update of $A$ has a negligible effect on the transformation $f_A$. 


We defer the proofs of the propositions to Sections~\ref{sec:proofrounding}, 
\ref{sec:proofconditionnumber}, \ref{sec:proofapproxtransform}, 
and proceed with the proof of Theorem~\ref{thm:rounding}.

\subsection{Proof of Theorem~\ref{thm:rounding}}

Before we proceed with the proof, we argue that the steps of the algorithm are well-defined.
In particular, we must show that the choice of $m$ is feasible.
Indeed, the following claim shows that such an index $m \in \{0, \ldots, d-1\}$ as required in Line \ref{m-line}
always exists, as there is at least one $p_i$ with $p_i \ge (G / d) \sigma_{\max} (A^{(V)})$.

\begin{claim}
There is an index $i \in \{1, \ldots, d\}$ such that $p_i \ge (G / d) \sigma_{\max} (A^{(V)})$.
\end{claim}
\begin{proof}
Letting $u_i = w_i^{(W_{i-1}^\perp)}/\|w_i^{(W_{i-1}^\perp)}\|_2$, 
we note that the $u_i$'s form an orthonormal basis of $\R^d$ 
and that $p_i = \|A u_i\|_2$. Therefore, we have that
$\sqrt{ \sum_{i=1}^d p_i^2 } = \|A\|_F \ge \sigma_{\min} (A^{(V^\perp)}) \ge G \sigma_{\max} (A^{(V)})$.
Since $\sqrt{ \sum_{i=1}^d p_i^2 } \le \sqrt{d} \max_{i=1}^d p_i$,
this means there is at least one $p_i$ with value at least $(G / d) \sigma_{\max} (A^{(V)})$.
\end{proof}

We now proceed to bound the improvement on the condition number of matrix $A$ at every iteration.
Starting with a matrix $A$ with condition number $\kappa(A)$,
the algorithm finds a subspace $V$ with a large multiplicative singular value gap $G$,
i.e., $\frac { \sigma_{\min} (A^{(V^\perp)}) } { \sigma_{\max} (A^{(V)}) } \ge G$.
This gap $G$ is at least $\kappa(A)^{1/d}$, as the following claim shows.

\begin{lemma}
At any iteration of the algorithm, $G \ge \frac 1 2 \kappa(A)^{1/d}$, and
$g \ge \frac {1} {2 d} \kappa(A)^{1/d^2}$.
\end{lemma}
\begin{proof}
  We prove each of the bounds separately.
  
  \textbf{We first bound $G$:}
  Indeed, we have that
  $$G = \frac {\sigma_{\min} (A^{(V^\perp)})} { \sigma_{\max} (A^{(V)}) }
  \ge \frac 1 2 \max_{1\leq i\leq d-1}  \frac {\sigma_{i} (A)} { \sigma_{i+1} (A) }
  \ge \frac 1 2 \left( \frac {\sigma_{\max} (A)} { \sigma_{\min} (A) }  \right)^{1/d} = \frac 1 2 \kappa(A)^{1/d} \;.$$

  \textbf{We now bound $g$:}
  Recall that $g = \frac { \min \{p_{m+1}, \sigma_{\min}(A^{(V^\perp)}) \}}{ \max \{\sigma_{\max}(A^{(W)}), \sigma_{\max}(A^{(V)}) \}} \;.$
  To show the statement, we lower bound all four combinations
  of numerators and denominators separately.
  \begin{itemize}
  \item Term $\sigma_{\min} (A^{(V^\perp)}) / \sigma_{\max} (A^{(V)})$:
  
  By the definition of $G$, we have that
  $$ \sigma_{\min} (A^{(V^\perp)}) \ge G \, \sigma_{\max} (A^{(V)}).$$

  \item Term $p_{m+1} / \sigma_{\max} (A^{(V)})$:
  
  Recall that the subspace $W$ is defined by computing an eigendecomposition of $A$
  with respect to the set of points $X$ to obtain $(w_i,p_i)_{i=1}^d$.
  It sets $W = \spn \{ w_1,\ldots,w_m\}$ by choosing the smallest $m \ge 0$
  so that $p_{m+1} \ge \frac 1 d G^{(m+1)/d} \sigma_{\max} (A^{(V)})$. This implies that 
  $$p_{m+1} \ge \frac 1 d G^{1/d}  \sigma_{\max} (A^{(V)}) \;.$$

  \item Term $p_{m+1} / \sigma_{\max} (A^{(W)})$:
  
  We also have that $\max_{1\leq i\leq m} p_i \le  \frac 1 d G^{m/d} \le G^{1/d} p_{m+1}$

  Moreover, we have that
  $$d \max_{1\leq i\leq m} p_i \ge \sqrt{ \sum_{i=1}^m p_i^2} = \|A^{(W)}\|_F \ge \|A^{(W)}\|_2 = \sigma_{\max}(A^{(W)}) \;.$$
  This implies that $\sigma_{\max}(A^{(W)}) \le d \max_{i=1}^m p_i$ which in turn gives that $\sigma_{\max}(A^{(W)}) \le d G^{1/d} p_{m+1}$.

  \item Term $\sigma_{\min} (A^{(V^\perp)}) / \sigma_{\max} (A^{(W)})$:

  Moreover, the definition of $m$ implies that for all $i \le m$, we have
  $p_i \le \frac 1 d G^{i/d} \sigma_{\max} (A^{(V)})$, and since $m < d$, we have
  $$\max_{1\leq i\leq m} p_i \le \frac 1 d G^{1-1/d} \sigma_{\max} (A^{(V)})
  \le \frac 1 d G^{-1/d}  \sigma_{\min} (A^{(V^\perp)}) \;.$$
  This implies that
  $$ \sigma_{\min} (A^{(V^\perp)}) \ge G^{1/d} \sigma_{\max} (A^{(W)}) \;.$$

  \end{itemize}
  Overall, we have that $g \ge \frac 1 d G^{1/d} \ge \frac 1 {2d} \kappa(A)^{1/d^2}$.
\end{proof}

We can now apply Proposition~\ref{prop:conditionnumber}
to bound from above the condition number $\kappa(AT)$
of the matrix $AT$ for $T = \delta I_R + I_{R^\perp}$.
For $\delta \ge 8 g^{-1}$, we get that
\begin{equation}
  \kappa(AT) \le 30 \delta \kappa(A) \le O( 2^{-b} ) \kappa(A).
\end{equation}

Initially, the condition number is bounded by $(2d)^{rdb}$,
as the following claim shows:
\begin{claim}
For any full-rank matrix $A \in \{-2^{rb},\ldots,2^{rb}\}^{d\times d}$,
we have that $\kappa(A) \le  d^d 2^{rdb}$.
\end{claim}
\begin{proof}
Note that $\kappa(A) = \frac {\sigma_{\max}(A)} {\sigma_{\min}(A)}$.
We have that $\sqrt{ \det(A^\top A) } = \prod_{i=1}^d \sigma_i(A)$, and thus
we obtain $\sigma_{\max}(A)^{d-1} \sigma_{\min}(A) \ge \sqrt{ \det(A^\top A) } \ge 1$,
where the last inequality follows since $A$ is full-rank and has integer entries.
This implies that $\kappa(A) \le \sigma_{\max}(A)^{d}$.
The statement follows since
$\sigma_{\max}(A) \le \|A\|_F \le \sqrt{ d^2 2^{2rb}} \le d 2^{rb}$.
\end{proof}

Thus, since at any iteration we have that $g \ge \Omega(2^b)$,
as $\kappa(A) \ge N \ge 2^{\Omega(d^2 b)}$,
the condition number significantly improves
at every iteration by a factor of $O(2^{-b})$;
after $O(d r)$ iterations, it will become $\max \{ 2^{\Omega(d^2 b)}, N \}$.

We now proceed to bound the change in the transformation
$\|f_A(x) - f_{ AT}(x)\|_2$ for all $x \in X$,
using Proposition~\ref{prop:approxtransform}.
To do this, we need to lower bound $\rho$. 
In particular, we show that for all $x \in X \setminus W$, 
it holds that $\|x^{(W^\perp)}\|_2 \ge d^{-d} 2^{-db}$.

\begin{claim}
For any iteration of $\textsc{Round}$, we have that $\min_{x \in X\setminus W} \|x^{(W^\perp)}\|_2 \ge d^{-d} 2^{-db}$.
\end{claim}

\begin{proof}
  Note that at any iteration, the bit complexity of $X$ stays the same as only the matrix $A$ gets updated.
  Fix any $x\in X\setminus W$.
Since $W$ is the span of the linearly independent vectors $w_1,\ldots,w_m$,
for any $x\in X\setminus W$,
$\|x^{(W^\perp)}\|_2 = \frac {\mathrm{Vol}(w_1,\ldots,w_m,x)} { \mathrm{Vol}(w_1,\ldots,w_m)}$,
where $\mathrm{Vol}(z_1, \ldots, z_k)$ is the $k$-dimensional volume
of the parallelepiped defined by the vectors $z_1,\ldots,z_k$.

Note that $\mathrm{Vol}(z_1, \ldots, z_k)$ is given by $\sqrt{\det(Z^\top Z)}$,
where $Z$ is the matrix $[z_1| \ldots |z_k]$.
Since the vectors $w_1,\ldots,w_m,x \in \{-2^{b},\ldots,2^b\}^d$
are integer vectors and linearly independent,
the corresponding volumes of the parallelepipeds they define
have volume at least $1$ and at most  $d^d 2^{db}$.
\end{proof}

We thus get that we can apply Proposition~\ref{prop:approxtransform} with 
$\rho = d^{-d} 2^{-db}$. In particular, we conclude that for all $x\in X$ we have
$$ \| f_{A}(x) - f_{AT}(x) \|_2 \le \frac {16} {\left(g-1\right) \rho \delta} 
\le \frac {d^{O(db)}} {\kappa(A)^{1/d^2}} \;.$$
The total incurred error for the transformation across all iterations  is at most
$$\| f_{A}(x) - f_{A'}(x) \|_2 \le \frac {d^{O(db)}} {N^{1/d^2}} \sum_{i=0}^{\infty} 2^{-b i /d^2} 
\le \frac {d^{O(db)}} {N^{1/d^2}}  \;,$$
where the first inequality is because $\kappa(A)$ ends at $N$
and decreases by a factor of $\Omega(2^{b})$ every time.
This is at most $\eps/2$,
if $N \ge \left( \frac d \eps \right)^{\Omega(d^3 b)}$.

Overall, we obtain a matrix $A'$ with condition number
at most $N \ge \left( \frac d \eps \right)^{\Omega(d^3 b)}$
such that $\| f_{A}(x) - f_{A'}(x) \|_2 \le \eps/2$. Applying Proposition~\ref{prop:rounding},
we show that by rounding $A'$ to have integer entries
of magnitude at most $\left( \frac d \eps \right)^{O(d^3 b)}$,
we get that $\| f_{A}(x) - f_{A'}(x) \|_2 \le \eps$

To ensure that the numerical operations of every iteration are $\poly(d,n)$,
we need to efficiently identify the subspace $V$.
We note that we only need to compute the singular vectors approximately,
so that the singular values are approximated within a small constant.
We can achieve this in $\poly(d)$ operations
using the algorithm of Proposition~\ref{eigendecomposition proposition}.
Under this approximation, the multiplicative gap $G$ satisfies
$G \ge \frac 1 2 \kappa(A)^{1/d}$, which results in the same order
of improvement, when the condition number of $A$ goes from $\kappa(A)$
to $\Theta(2^{-b}) \kappa(A)$ at every iteration.
Thus, overall, the number of arithmetic operations is $\poly(d,n,r)$,
as the number of iterations is at most $O(d r)$.

To ensure that the bit complexity of the operations remains bounded,
we need some additional care. It is easy to see that if the bit complexity
of the points is $\poly(d,n,r,\log(1/\eps))$,
it remains $\poly(d,n,r,\log(1/\eps))$ at the end of a single iteration.
Yet, the increase may be significant over multiple iterations
and the bit complexity may blow up exponentially.

To keep the bit complexity bounded, we can apply
Proposition~\ref{prop:rounding} with accuracy parameter $\eps / 2^{O(rdb)}$,
to round the resulting matrix $A$, so that it has entries with bit complexity
$O( \log(\kappa(A) ) +  d r b \log(1/\eps) )$.
This has negligible effect on the decrease of the condition number
at every iteration (as it only increases it by at most a fixed constant),
and ensures that the bit complexity of $A$ at every iteration
is bounded by a fixed polynomial in $d, n, r, b$ and $\log(1/\eps)$.
Moreover, the introduced error in the transformation $f_A$
for every iteration is at most $\eps / 2^{O(drb)}$,
and thus over all $O( d r)$ iterations, it is at most $O(\eps)$, as desired.

We finally remark that all constants used in the algorithm 
can be computed as a function of the entries of $X$ and $A$. 
In particular, $2^{O(b)} = \max_{x \in X} \|x\|_\infty^{O(1)}$ 
and $2^{O(r b)} = \max_{i,j} \|A_{ij}\|^{O(1)}$.

This completes the proof of Theorem \ref{thm:rounding}. 
We now proceed with the proofs of Propositions \ref{prop:rounding}, 
\ref{prop:conditionnumber}, and \ref{prop:approxtransform}.

\subsection{Proof of Proposition~\ref{prop:rounding}}
\label{sec:proofrounding}
Let $A_r \triangleq \frac d {\bar \sigma \eps} A$.
To prove the proposition, we first bound
the condition number of $\hat A$
and then show that the transformation $f_{\hat A}$ is close to $f_A$.

Matrix $A_r$ has singular values $\sigma_1(A_r) = \frac d {\bar \sigma \eps} \sigma_1(A) \le 2 \frac d \eps \kappa(A)$
and $\sigma_d(A_r) = \frac d {\bar \sigma \eps} \sigma_d(A) \ge \frac d {\eps}$. Moreover, after rounding
the entries of matrix $A_r$ to the nearest integer to obtain $\hat A$,
for any unit vector $v$, we have that
$$ \|(\hat A - A_r ) v\|_2 \le \frac d 2 \;,$$
since $\hat A - A_r$ has entries of magnitude at most $1/2$.

Therefore, we get that $\hat A$ has singular values
$\sigma_1(\hat A) \le \sigma_1(A_r) + \frac d 2 \le \sigma_1(A_r) (1 + \eps / 2)$
and 
$\sigma_d(\hat A) \ge \sigma_d(A_r) - \frac d 2 \ge \sigma_d(A_r) (1 - \eps / 2)$.
This implies that the condition number is at most $\kappa(A_r) (1+4\eps) = \kappa(A) (1+4\eps)$.
Moreover, the magnitude of every entry is at most $\sigma_1(\hat A) \le 2 \frac d \eps \kappa(A) (1+\eps/2)$.

To bound the effect of the rounding on the transformation,
we note that by Fact~\ref{fact:transform-properties},
$f_{A_r} = f_A$ and thus for any $x$ we have that
\begin{align*}
\|f_A(x) - f_{\hat A}(x)\|_2
&= \|f_{A_r}(x) - f_{\hat A}(x)\|_2 = \left\| \frac {A_r x} {\|A_r x\|_2 } - \frac {\hat A x} {\|\hat A x\|_2 }\right\|_2 \\
&\le \left\| \frac{A_r x - \hat A x } { \|A_r x\|_2 }  \right\|_2 + \left\| \frac{  \hat A x} {\|A_r x\|_2}
- \frac{  \hat A x} {\|\hat A x\|_2}  \right\|_2 \\
&= \frac{\|(A_r - \hat A) x\|_2 } { \|A_r x\|_2 } + \left| 1 - \frac {\|\hat A x\|_2} {\| A_r x\|_2} \right| \\
&\le 2 \frac{\|(A_r - \hat A) x\|_2 } { \|A_r x\|_2 } \le \frac { d} {\sigma_d(A_r)}
 \le 2 \eps \;.
\end{align*}


\subsection{Proof of Proposition~\ref{prop:conditionnumber}}\label{sec:proofconditionnumber}
For any matrix $\hat T$, we can bound the condition number of $\kappa(AT)$
by $\kappa(A\hat T)\kappa( \hat T^{-1} T )$.
We define $\hat T = \delta I_{V^\perp} + I_{V}$ as the matrix that rescales
the subspace $V^\perp$ by $\delta$. We argue that $\kappa(A\hat T) \leq 10 \delta \kappa(A)$.

  \begin{claim}
    For any $\delta \ge g^{-1}$, it holds that $\kappa(A\hat T) \le 10 \delta \kappa(A)$.
  \end{claim}
  \begin{proof}
  We first show that $\sigma_{\max}(A\hat T) \le 2\delta  \sigma_{\max}(A)$.
  Indeed, for any unit vector $v \in \R^d$, we have that
  $$\|A\hat T v\|_2 = \|\delta A v^{(V^\perp)} + A v^{(V)}\|_2
  \le \delta \|A v^{(V^\perp)}\|_2 + \|A v^{(V)}\|_2 \le \delta \sigma_{\max}(A) + g^{-1} \sigma_{\max}(A) \;.$$
  and since $\delta \ge g^{-1}$, this gives the required bound on $\sigma_{\max}(A\hat T)$.

  We now argue that $\sigma_{\min}(A\hat T) \ge \frac 1 5 \sigma_{\min}(A)$.
  For any unit vector $v \in \R^d$, we have that
    \begin{align*}
  \|A\hat T v\|_2 &= \|\delta A v^{(V^\perp)} + A v^{(V)}\|_2 \ge \delta \|A v^{(V^\perp)}\|_2 - \|A v^{(V)}\|_2 \\
  &\ge \delta \sigma_{\min} (A^{(V^\perp)}) \|v^{(V^\perp)}\|_2 - \sigma_{\max} (A^{(V)}) \|v^{(V)}\|_2 \\
  &\ge \delta g \sigma_{\max} (A^{(V)}) \|v^{(V^\perp)}\|_2 - \sigma_{\max} (A^{(V)}) \|v^{(V)}\|_2 \\
  &\ge \sigma_{\max} (A^{(V)}) (\|v^{(V^\perp)}\|_2 - \|v^{(V)}\|_2 ) \;,
    \end{align*}
    where the last inequality follows since $\delta \ge g^{-1}$.
    Since, $\sigma_{\min}(A) \le \min_{v \in V: \|v\|_2=1} \|A v\|_2 \le  \sigma_{\max}(A^{(V)})$, we get that the above is at least $\sigma_{\min}(A) (\|v^{(V^\perp)}\|_2 - \|v^{(V)}\|_2) $.
  This is greater than $\frac 1 5 \sigma_{\min}(A)$ if $\|v^{(V)}\|_2 < \frac 3 5$.
  Moreover, if $\|v^{(V)}\|_2 \ge \frac 3 5$, we have that
  $$\|A\hat T v\|_2 \geq \sigma_{\min}(A) \|\hat T v\|_2 \ge \sigma_{\min}(A) \|v^{(V)}\|_2 \ge \frac 3 5 \sigma_{\min}(A) \;,$$
  which again gives the required bound on $\sigma_{\min}(A\hat T)$.
  \end{proof}

  To complete the proof of
  Proposition~\ref{prop:conditionnumber},
  we now bound $\kappa(\hat T^{-1} T)$.

  \begin{claim}
    We have that $\kappa(\hat T^{-1} T) \le 3 $.
  \end{claim}
  \begin{proof}
  Note that $\hat T^{-1} = I_V + 1/\delta I_{V^\perp}$ and that $T = \delta I_R + I_W + I_{W^\perp \cap V}$.
    Since,
    \begin{align*}
      \delta \hat T^{-1}I_R &= I_R - (1-\delta) I_V I_R, 
      \\
      \hat T^{-1} I_W &= I_W + (1-1/\delta) I_{V^\perp} I_W,\\  
      \hat T^{-1} I_{W^\perp \cap V} &= I_{W^\perp \cap V},
      \end{align*}
  and $I =  I_R + I_W + I_{W^\perp \cap V}$, this implies that
  \begin{align*}
  \hat T^{-1} T =  I + (1/\delta-1)I_{V^\perp}I_W - (1-\delta) I_V I_R \;.
  \end{align*}
  Our result will follow by bounding each of the terms $\|(1/\delta-1)I_{V^\perp}I_W\|_2, \|(1-\delta)I_{V}I_R\|_2$ individually below $1/4$.

  To bound the first term, we argue that for every unit vector $w \in W$, we have that
  $$\| w^{(V^\perp)} \|_2 \le 2/g \;.$$
  This is because $\|A w \|_2 \le \sigma_{\max} (A^{(W)}) \le g^{-1} \sigma_{\min} (A^{(V^\perp)})$,
  but
  \begin{align*}
    \|A w \|_2 &\ge \|A w^{(V^\perp)} \|_2  - \|A w^{(V)} \|_2
    \ge  \sigma_{\min} (A^{(V^\perp)}) \|w^{(V^\perp)} \|_2 - \sigma_{\max} (A^{(V)}) \|w^{(V^\perp)} \|_2 \\
    &\ge \sigma_{\min} (A^{(V^\perp)}) (\|w^{(V^\perp)} \|_2 - g^{-1} \|w^{(V)} \|_2)
       \ge \sigma_{\min} (A^{(V^\perp)}) (\|w^{(V^\perp)} \|_2 - g^{-1} ) \;.
  \end{align*}
  We thus get that $\| I_{V^\perp} I_{W} \|_2 \le 2g^{-1}$. Thus,
  $$
  \|(1/\delta-1)I_{V^\perp}I_W\|_2 \leq 2/(\delta g) le 1/4 \;.
  $$
  To bound $\|I_VI_R\|_2$, it suffices to show that for any unit vector $x\in R$, $\|I_Vx\|_2$ is small.
  As $x\in \mathrm{span}(W\cup V^\perp)$, we can write $x=x_W+x_{V^\perp}$
  with $x_W\in W$ and $x_{V^\perp}\in V^\perp$. Since $x\in W^\perp$, we have that
  $$ 0 = x_W \cdot x = \|x_W\|_2^2 + x_W\cdot x_{V^\perp} \;.$$
  Thus,
  $$\|x_W\|^2_2 = |x_W\cdot x_{V^\perp}| \leq \|x_{V^\perp}\|_2 \|I_{V^\perp}I_W x_W\|_2 \leq (2/g)\|x_{V^\perp}\|_2\|x_W\|_2 \;,$$
  which implies that $\|x_W\|_2 \leq (2/g)\|x_{V^\perp}\|_2.$ and that $\|x\|_2 \ge \|x_{V^{\perp}}\|_2 - \|x_{W}\|_2 \ge (g/2 - 1)\|x_{W}\|_2$.
  Therefore,
  $$
  \|I_Vx\|_2 = \|I_V x_W\|_2 \leq \|x_W\|_2 \leq (g/2-1)^{-1} \|x\|_2 \;.
  $$
  Thus for $g > 10$,
  $ \|(1-\delta)I_{V}I_R\|_2 \leq 1/4$.
  \end{proof}

\subsection{Proof of Proposition~\ref{prop:approxtransform}}\label{sec:proofapproxtransform}

We first show the following claim that relates the subspace $R$ to $V$ and $W$.

\begin{claim}\label{claim:subspace partition}
  For any subspaces $V,W$ of $\R^d$, and $R = \spn(W\cup V^\perp) \cap W^\perp$, it holds that $I = I_R + I_W + I_{W^\perp\cap V}$.
\end{claim}
\begin{proof}
  We first argue that vectors in $W$, $R$ and $W^\perp\cap V$ are pairwise orthogonal. Indeed, any vector in the latter two subspaces belongs in $W^\perp$ and thus is orthogonal to $W$. Moreover, note that any vector $r \in R$ belongs in $\spn(W\cup V^\perp)$ and can be written as $r = w + v$ for some $w \in W$ and $v \in V^\perp$.
  For any $u \in W^\perp \cap V$, it holds that $u\cdot r = u\cdot w + u\cdot v = 0$, because $u\cdot w = 0$ as $u \in W^\perp$ and $w \in W$ and $u\cdot v = 0$ as $u \in V$ and $v \in V^\perp$.
  We now argue that $\spn(W,R,W^\perp \cap V) = \R^d$.
  Indeed, $\spn(W,R) = \spn(W,\spn(W\cup V^\perp) \cap W^\perp) = \spn(W,\spn(W, \proj_{W^\perp} V^\perp) \cap W^\perp) = \spn(W, \proj_{W^\perp} V^\perp) = \spn(W, V^\perp)$. This is true, as for any subspaces $A,B$, $\spn(A,B) = \spn(A,\proj_{A^\perp}B)$.  Thus, $\spn(W,R,W^\perp \cap V) = \spn(W, V^\perp, W^\perp \cap V)  = \R^d$.
\end{proof}

We now proceed to show Proposition~\ref{prop:approxtransform}
assuming $x \in X \setminus W$ as in the case that $x \in W$, we have that $Tx =x$ and so $f_A(x) = f_{AT}(x)$.
We argue that for any such point $x \in X \setminus W$, both $f_A(x)$ and $f_{AT}(x)$
are close to $ \frac {A x^{(R)} } {\|A x^{(R)}\|_2}$.
To do this, we use Claim~\ref{claim:smallcontributions},
which shows that the contributions to $A x$
of the projections of $x$ to $W$ and $W^\perp \cap V$ are small.

\begin{claim}\label{claim:smallcontributions}
For any $x \in X \setminus W$,
  $$\|A x^{(R)} \|_2 \ge (g-1) \rho \max \{ \|A x^{(W)} \|_2 , \|A x^{(W^\perp \cap V)} \|_2 \} \;.$$
\end{claim}

\begin{proof}
  We analyze two cases separately.

  \paragraph{We first bound $\| A x^{(W^\perp \cap V)}\|_2$.}
  If $x^{(W^\perp \cap V)} = 0$, we have that $\| A x^{(W^\perp \cap V)}\|_2 = 0$.
  Otherwise, we have that
  $$\|A x^{(W^\perp)}\|_2 \ge g \| x^{(W^\perp)} \|_2 \sigma_{\max} (A^{(V)})
  \ge g \| x^{(W^\perp)} \|_2\frac {\|A x^{(W^\perp \cap V)} \|_2}{  \| x^{(W^\perp \cap V)} \|_2} \ge g \|A x^{(W^\perp \cap V)} \|_2 \;.$$
  By the triangle inequality and since $x^{(W^\perp)} = x^{(R)} + x^{(W^\perp \cap V)}$,
  we get that $\|A x^{(R)}\|_2 \ge \|A x^{(W^\perp)}\|_2 - \|A x^{(W^\perp \cap V)}\|_2$.
  This implies that $\|A x^{(R)}\|_2 \ge (g - 1) \|A x^{(W^\perp \cap V)}\|_2$ and
  $\|A x^{(R)}\|_2 \ge (1 - 1/g) \|A x^{(W^\perp)}\|_2$.
  \paragraph{
  We now bound $\|A x^{(W)}\|_2$.}
  If $x^{(W)} = 0$, we have that $\|A x^{(W)}\|_2=0$.
  Otherwise, we have that
  $$\frac { \|A x^{(W^\perp)}\|_2 } {\|x^{(W^\perp)}\|_2}
  \ge g \sigma_{\max} (A^{(W)}) \ge g \frac {\|A x^{(W)} \|_2} {\|x^{(W)}\|_2} \;.$$
  Now since $\|A x^{(R)}\|_2 \ge (1-1/g) {\|A x^{(W^\perp)}\|_2}$, we get that
  $$ {\|A x^{(R)}\|_2} \ge (g-1) \frac { \|x^{(W^\perp)}\|_2} {\|x^{(W)}\|_2} {\|A x^{(W)} \|_2} \ge (g-1) \rho {\|A x^{(W)} \|_2} \;.$$
\end{proof}

Using Claim~\ref{claim:smallcontributions}, we now have that for $\mu = \frac 1 {(g-1) \rho}$:
\begin{align*}
  \|f_A(x) - \frac {A x^{(R)} } {\|Ax^{(R)}\|_2}\|_2
  &\le
  \|f_A(x) - \frac {A x^{(R)} } {\|A x\|_2}\|_2 + \|\frac {A x^{(R)} } {\|A x\|_2} - \frac {A x^{(R)} } {\|Ax^{(R)}\|_2}\|_2 \\
  &\le
  \frac {\|A x^{(W^\perp \cap V)}\|_2 + \|A x^{(W)}\|_2 } {\|A x\|_2} +
   \left| \frac { \|A x^{(R)}\|_2} {\|A x\|_2} - 1\right| \\
   &\le
   2 \mu  \frac {\|A x^{(R)}\|_2 } {\|A x\|_2} +
   \left| \frac { \|A x^{(R)}\|_2} {\|A x\|_2} - 1\right| \;
\end{align*}
Since $\frac  {\|A x\|_2} { \|A x^{(R)}\|_2} \in [1-2\mu, 1+2\mu]$,
we get that:
$$\|f_A(x) - \frac {A x^{(R)} } {\|Ax^{(R)}\|_2}\|_2 \le \frac {4\mu} { 1 - 2\mu} \;.$$

Similarly we get that
$\|f_{AT}(x) - \frac {A x^{(R)} } {\|Ax^{(R)}\|_2}\|_2 \le \frac {4\mu/\delta} { 1 - 2\mu/\delta}$,
by noting that
$A T x^{(R)} = \delta A x^{(R)}$,
$A T x^{(W^\perp \cap V)} = A x^{(W^\perp \cap V)}$ and
$A T x^{(W)} = A x^{(W)}$.

Combining the above we get that,
$$\|f_A(x) - f_{AT}(x)\|_2 \le \frac {4\mu} { 1 - 2\mu} + \frac {4\mu/\delta} { 1 - 2\mu/\delta} \le 16 \mu/\delta \;,$$
for $\delta \in [0,1]$ and $\mu /\delta \in [0,1/4]$.

\section{Full Algorithm: Proof of Theorem~\ref{thm:forster-intro}} \label{sec:full-forster}

In this section, we will put together the basic algorithm from Section~\ref{sec:forster-modular}
with the approximate eigendecomposition algorithm from Section~\ref{sec:eigen}
and the bit complexity reduction routine from Section~\ref{sec:rounding}
to prove Theorem~\ref{thm:forster-intro}.
The final algorithm is given in pseudocode below.

\begin{algorithm}[hbt!]
   \caption{Full Forster Transform Algorithm}
   \label{alg:forster-full}
\begin{algorithmic}[1]
    \Function{$\textsc{ForsterTransform}$ }{set $X {\subset \R_{\ast}^d}$ of $n$ points, accuracy parameter $\eps$}
    \State Let $A \leftarrow I$  {\hspace{22pt} \(\triangleright\) Initialization  of transformation matrix $A$}
    	\While{$\| M_A(X) \|^2_F > \frac 1 d + \frac {\eps^2} {d^2}$}
                \State Set $A \leftarrow \textsc{ImproveTransform}(A,X,\eps,\delta)$,
                for $\delta$ a sufficiently small polynomial in $\eps/dn$.
                \If{ \textsc{ImproveTransform} returned a subspace $V$}
                     \State \Return $V$.
                \EndIf
                \State Set $A \leftarrow \textsc{Round}(A,X,\zeta)$,
                for $\zeta$ a sufficiently small multiple of $\eps^5/(d^{10} n^5)$.
        \EndWhile

        \State \Return $A$
    \EndFunction
  \end{algorithmic}
\end{algorithm}

\noindent The full version of our \textsc{ImproveTransform} function is given in pseudocode below.

\begin{algorithm}[hbt!]
    \caption{Find Improved Transform Matrix}
    \label{alg:improvement-full}
 \begin{algorithmic}[1]
     \Function{$\textsc{ImproveTransform}$ }
     {current matrix $A \in \R^{d\times d}$, $X \subset \R_{\ast}^d$, {accuracy parameter $\eps$}, {error parameter $\delta$}}
          \State Let $C$ be a sufficiently large constant.
          \State Set $n\leftarrow |X|$.
          \State Set $a_1, \ldots,a_d, q_1,\ldots,q_d \leftarrow\textsc{EigenDecomposition}(M_A(X),\eta,\delta)$,
          for $\eta = \eps^4/(C^3 d^8 n^4)$.
          \State Sort $a_i \|q_i\|_2^2$ in descending order of size.
          \State Find $k$ maximizing $a_k \|q_k\|_2^2 - a_{k+1} \|q_{k+1}\|_2^2$.
          \State Let $W$ be the span of $q_{k+1},\ldots,q_d$.
          \State Let $\gamma$ be $\eps^2/(Cd^4n^2)$

          \Statex  {\vspace{5pt} \hspace{22pt} \(\triangleright\) Consider the Following Two Cases \vspace{5pt}}

          \If{there exists $x \in X$ such that $\left\| \proj_{W} f_A(x) \right\|_2, \left\| \proj_{W^\perp} f_A(x) \right\|_2 \ge \gamma$}

                \State Set $V \leftarrow W$.
                \State Set $\alpha \leftarrow \eps/(64 n d^3)$.

	  \Else

                \State {Set} $X^\Bg {\leftarrow} \{x \in X: \| \text{proj}_{W^\perp} f_A(x)\|_2 \ge \gamma \}$.
                \State Let $c_1,c_2,\ldots,c_d,r_1,r_2,\ldots,r_d \leftarrow \textsc{EigenDecomposition}(M_A(X^\Bg),\eta)$
                \State Let $V$ be the span of the $d-k$ vectors $r_i$ with the smallest values of $c_i \|r_i \|_2^2$.
        		\State Set $\beta  \leftarrow \max_{x \in X^\Bg} \| f^{(V)}_A(x) \|_2 $.
                \If{$\beta = 0$}
                       \Statex {\vspace{5pt} \hspace{22pt}  \(\triangleright\)   No Forster Transform Exists  \vspace{5pt}}
		         \State \Return The subspace $V^\perp$.
		         \hspace{22pt}
		 \Else     {\vspace{5pt} \hspace{22pt}  \(\triangleright\)   Case where $\beta>0$  \vspace{5pt}}
		
			\State Set $\alpha \leftarrow \eps/(3\beta d^2 n) - 1$
		\EndIf
         \EndIf
        \State \Return $\left( I + \alpha I_V \right) A$
         \EndFunction
   \end{algorithmic}
 \end{algorithm}

The rest of this section will be devoted to proving the correctness of this algorithm.

To begin with, we note that if the algorithm returns a matrix $A$,
it must be the case that $\|M_A(X)\|_F^2 \leq 1/d + \eps^2/d^2$,
and so by Lemma~\ref{lem:phi-prop},
$A$ will be an $\eps$-Forster transform matrix.
Also note that upon applying \textsc{Round},
we replace $A$ with a matrix whose entries have bit complexity $\poly(bdn/\eps)$.
From there it is not hard to see that all arithmetic computations performed
by this algorithm are done to only polynomial precision.
Finally, we note that in each iteration of the main while loop,
our algorithm performs a polynomial number of arithmetic operations.
Therefore, in order to prove correctness, we need to establish the following:
\begin{enumerate}
\item If our algorithm returns a subspace $V$, then $|X\cap V| > |X|\dim(V)/d$.
\item In each iteration of our while loop, the potential function
$\Phi_X(A) := \|M_A(X)\|_F^2$ decreases by at least an inverse-polynomial amount.
\end{enumerate}
We note that if this is the case,
we will only need to call \textsc{EigenDecomposition} a polynomial number of times,
and thus we may assume that all such calls succeed, which we will assume hereafter.

We begin with a basic consequence of our eigendecomposition lemma:
\begin{lemma}\label{Frob-Eigendecomp-Error-Lemma}
We have that $\|M_A(X) - \sum_{i=1}^d a_i q_iq_i^\top\|_F \leq \sqrt{d}\eta$.
\end{lemma}
\begin{proof}
Letting $M=M_A(X)$ and $\hat M=\sum_{i=1}^d a_iq_iq_i^\top$, we have that
for any unit vector $v$ it holds that $|v^\top (M-\hat M)v| \leq \eta \, (v^\top M v) \leq \eta$.
This means that $\|M - \hat M\|_2 \leq \eta$. We note that since this is the maximum eigenvalue of $M-\hat M$,
and since the Frobenius norm of $M-\hat M$ is the square root of the sum of squares of the eigenvalues,
we have that $\|M-\hat M\|_F \leq \sqrt{d}\eta$, as desired.
\end{proof}

We next show that our approximate eigendecomposition exhibits an eigenvalue gap. In particular,
we establish the following lemma.
\begin{lemma}
We have that
$ a_k \, \|q_k \|^2 - a_{k+1} \, \|q_{k+1}\|^2 \geq (3/4)(\eps/d^3)$.
\end{lemma}
\begin{proof}
Note that $\tr(M_A(X))=1$.
Therefore, by Lemma \ref{Frob-Eigendecomp-Error-Lemma} and
letting $\hat M = \sum_{i=1}^d a_iq_iq_i^\top$,
we have that that $|\tr(\hat M) - 1| \leq \sqrt{d} \, \|M_A(X)-\hat M\|_F \leq d \, \eta$.
Moreover, we have that $\|\hat M\|_F \geq \|M_A(X)\|_F -\sqrt{d}\eta$.
Thus, $\|\hat M\|_F^2 \geq 1/d+\eps^2/d^2+O(\eta).$ On the other hand, we can write
$$\|\hat M\|_F^2 = \sum_{i=1}^d (a_i \, \|q_i\|_2^2)^2
= \sum_{i=1}^d (a_i \|q_i\|_2^2-1/d)^2 + 2 \, \tr(\hat M)/d - 1/d
= \sum_{i=1}^d (a_i \|q_i\|_2^2-1/d)^2+1/d +O(d\eta) \;.$$
This implies that
$$
\sum_{i=1}^d (a_i \, \|q_i\|_2^2-1/d)^2 \geq \eps^2/d^2 + O(d\eta) \;.
$$
Thus, there must be some $i$ with $\left| a_i \|q_i\|_2^2-1/d \right| \geq (99/100) \eps/d^2.$
Since the average value of $a_i \|q_i \|_2^2-1/d$ is
$(\tr(\hat M)-1)/d = O(d\eta)$, the difference between the biggest and smallest values
of $a_i \, \|q_i\|_2^2-1/d$ must differ by at least $(3/4)(\eps/d^2)$.
Therefore, the biggest single gap between consecutive values of $\|q_i\|_2^2$
must be at least $(3/4)(\eps/d^3)$. This completes our proof.
\end{proof}

It is now easy to show that \textsc{ImproveTransform}
decreases our potential in the case where there exists $x \in X$
such that $\left\| \proj_{W} f_A(x) \right\|_2, \left\| \proj_{W^\perp} f_A(x) \right\|_2 \ge \gamma$.
To prove this, we would like to apply Proposition \ref{case-1-prop}.
In particular, in this case, we let
$\rho = \max_{x\in X}(\min(\left\| \proj_{W} f_A(x) \right\|_2, \left\| \proj_{W^\perp} f_A(x) \right\|_2))$.
By assumption, we have that
$\rho > \gamma$ and $\rho = \max_{x\in X}(\min(\left\| \proj_{W} f_A(x) \right\|_2, \left\| \proj_{W^\perp} f_A(x) \right\|_2))$,
as the first property in Proposition \ref{case-1-prop} requires.

Next we let $M=M_A(X)$, and $\hat M = \sum_{i=1}^d a_i q_i q_i^\top$. We note that
$$
\lambda_{\min}(\hat M^{V^\perp,V^\perp}) - \lambda_{\max}(\hat M^{V,V}) = a_k \, \|q_k\|_2^2 - a_{k+1} \, \|q_{k+1}\|_2^2
\geq (3/4)(\eps/d^3) \;.
$$
Therefore, by Lemma \ref{Frob-Eigendecomp-Error-Lemma}, we have that
$$
\lambda_{\min}( M^{V^\perp,V^\perp}) - \lambda_{\max}( M^{V,V}) =
a_k \, \|q_k\|_2^2 - a_{k+1} \, \|q_{k+1}\|_2^2 \geq \eps/(2d^3) \;,
$$
showing that Property 2 holds.

Finally, we note that
$$
\hat M^{V,V^\perp} = \0 \;.
$$
Thus, by Lemma \ref{Frob-Eigendecomp-Error-Lemma}, we have that
$$
\|M^{V,V^\perp}\|_F \leq \sqrt{d}\eta \leq \alpha \leq \alpha \rho \;,
$$
thus showing that the third property applies.

Therefore, applying Proposition \ref{case-1-prop},
if there is an $x\in X$ such that
$\left\| \proj_{W} f_A(x) \right\|_2, \left\| \proj_{W^\perp} f_A(x) \right\|_2 \ge \gamma$,
then setting $C=\textsc{ImproveTransform}(A,X,\eps)$, we have that
$$
\Phi_X(C) \leq \Phi_X(A) - \rho^2 \eps / (8nd^2) \leq \Phi_X(A) - \gamma^2 \eps / (8nd^2) \;.
$$

For the case where all $x\in X$ have
$\min(\left\| \proj_{W} f_A(x) \right\|_2, \left\| \proj_{W^\perp} f_A(x) \right\|_2 ) \leq \gamma$,
we would like to apply Proposition \ref{case-2-prop}.
We begin by showing that all of the necessary properties apply.

For starters, letting $\hat M = \sum_{i=1}^d a_i q_i q_i^\top$ and $M=M_A(X)$, we have that
$$
\lambda_k(\hat M) - \lambda_{k+1}(\hat M) = a_k \|q_k\|_2^2 - a_{k+1} \|q_{k+1} \|_2^2 \geq (3/4)(\eps/d^3) \;.
$$
Since $\|M_A(X) -\hat M\|_F \leq \sqrt{d}\eta$, we have that
$$
\lambda_k(M_A(X)) - \lambda_{k+1}(M_A(X)) \geq (3/4)(\eps/d^3) - 2\sqrt{d}\eta \geq \eps/(2d^3) \;.
$$
The requirement that for each $x\in X$ that
$\min(\left\| \proj_{W} f_A(x) \right\|_2, \left\| \proj_{W^\perp} f_A(x) \right\|_2 ) \leq \gamma$ is a bit subtle,
since the $W$ used in Proposition \ref{case-2-prop} is the relevant eigenspace of $M_A(X)$,
while our $W$ is merely an approximation of it.
Fortunately, it is not hard to show that these spaces are relatively close to each other.

\begin{lemma}\label{eigenvectors close lemma}
If $v$ is a unit eigenvector of $M$, then
$
\|\proj_W(v)\|_2 \leq  3d^{7/2}\eta/\eps  \leq \gamma/\sqrt{d}
$
if it is one of the top $k$ eigenvectors, and
$
\|\proj_{W^\perp}(v)\|_2  \leq (3 d^{7/2}\eta/\eps ) \leq \gamma/\sqrt{d}
$
otherwise.
\end{lemma}
\begin{proof}
We have by definition that $Mv= \lambda v$ for some $\lambda$.
We have that either $\lambda \leq \lambda_k(\hat M) - (3/8)(\eps/d^3)$
or $\lambda \geq \lambda_{k+1}(\hat M) + (3/8)(\eps/d^3)$.
Without loss of generality, we assume the latter.
We note that by Lemma \ref{Frob-Eigendecomp-Error-Lemma}
that $\|Mv-\hat M v\|_2 \leq \sqrt{d}\eta$, and thus
$\|\hat M v - \lambda v\|_2 \leq \sqrt{d}\eta$.

Letting $v^{(W)}$ and $v^{(W^\perp)}$ denote the projections of $v$ onto $W$ and $W^\perp$,
and noting that $\hat M v^{(W)} \in W$ and $\hat M v^{(W^\perp)} \in W^\perp$,
we have that
$$
\| \hat M v^{(W)} - \lambda v^{(W)}\|_2 \leq \sqrt{d}\eta \;.
$$
On the other hand, we have that $\lambda I_W - \hat M^{W,W} \geq (3/8)(\eps/d^3)I_W$.
Therefore, we have that
$$
(3/8)(\eps/d^3)\|v^{(W)}\|_2 \leq \sqrt{d}\eta \;.
$$
From this, we conclude that $\|v^{(W)}\|_2 \leq 3 d^{7/2}\eta/\eps \leq \gamma/\sqrt{d}$.
This completes our proof.
\end{proof}

From the preceding, we note that for any unit vector $u$
that is a linear combination of either the top-$k$ or bottom $d-k$
eigenvectors of $M$ that $u$ is $\gamma$-close
to either $W^\perp$ or $W$ respectively (since it is a sum of relevant eigenvectors).
We have that for any $x\in X$ that either $\|f_A^{(W)}(x)\|_2 \leq \gamma$ or
$\|f_A^{(W^\perp)}(x)\|_2 \leq \gamma$. In the former case,
if $u$ is a unit vector that is a linear combination of the bottom $d-k$ eigenvectors,
then $u$ is $\gamma$-close to $W$,
so $u\cdot f_A(x) \leq 2 \gamma$. This implies that the projection of $f_A(x)$
onto the eigenspace of the bottom $d-k$ eigenvectors has norm at most $2\gamma$.
Similarly, if $\|f_A^{(W^\perp)}(x)\|_2 \leq \gamma$, then the projection of $f_A(x)$
onto the eigenspace of the top-$k$ eigenvectors is at most $2\gamma$.
This shows that the hypothesis of Proposition \ref{case-2-prop}
involving the projections of these vectors onto what it calls $W$
is satisfied with $\gamma$ replaced by $2\gamma$.

For Property 1 
we let $\tilde M:= \sum_{i=1}^d c_i r_i r_i^\top$,
and note that $\|\tilde M - M_A(X^\Bg)\|_F \leq \sqrt{d} \eta$.
We note that $V$ is the $(d-k)$-dimensional subspace minimizing $\tr(\tilde M^{V,V})$.
In particular, this implies that for $U$ the span of the bottom $d-k$ eigenvectors of $M_A(X)$,
we have that
$$
\tr(M_A^{V,V}(X^\Bg)) \leq \tr(\tilde M^{V,V}) + d \eta
\leq \tr(\tilde M^{U,U}) + d \, \eta
\leq \tr(M_A^{U,U}(X^\Bg)) + 2 d \eta
\leq \tr(M_A^{U,U}(X^\Bg))+\gamma^2/4 \;.
$$
This shows that Property 1 holds for $\delta = \gamma/2.$

Property 2 
holds similarly.
Property 3 
holds since
$$
\lambda_k(M_A^{V^\perp,V^\perp}(X^\Bg))
\geq \lambda_k(\tilde M^{V^\perp,V^\perp}) -\sqrt{d}\eta
= \lambda_k(\tilde M) -\sqrt{d}\eta
\geq \lambda_k(M_A(X^\Bg)) -2\sqrt{d}\eta
\geq \lambda_k(M_A(X^\Bg))-\gamma/2 \;.
$$
For Property 4,  
recall that $\beta = \max_{x\in X^\Bg}\|f_A^{(V)}(x)\|_2$.
This implies that $\|M_A^{V,V}(X^\Bg)\|_2 \leq \beta^2$.
By the relative error property of Proposition \ref{eigendecomposition proposition},
this implies that for $\tilde M := \sum_{i=1}^d r_i r_i^\top$
that $\|\tilde M^{V,V}\|_2 \leq 2\beta^2$.
Also note that by definition $\tilde M^{V,V^\perp}= \0$.
Next, let $v$ be a unit vector in $V$ and $w$ a unit vector in $V^\perp$.
We have that
$$
(v\pm\beta w)^\top \tilde M (v\pm\beta w) = v^\top \tilde M v +\beta^2 w^\top \tilde M w \leq 3\beta^2 \;.
$$
Thus, by the relative error property of Proposition \ref{eigendecomposition proposition},
we have that
$$
(v\pm\beta w)^\top M_A(X^\Bg) (v\pm\beta w) = (v\pm\beta w)^\top \tilde M (v\pm\beta w) + O(\eta \beta^2) \;.
$$
Taking the difference, we get that
\begin{align*}
2\beta v^\top M_A(X^\Bg) w
& = (v+\beta w)^\top M_A(X^\Bg) (v+\beta w)-(v-\beta w)^\top M_A(X^\Bg) (v-\beta w)\\
& = (v+\beta w)^\top \tilde M(X^\Bg) (v+\beta w)-(v-\beta w)^\top \tilde M (v-\beta w) + O(\eta \beta^2)\\
& = (v^\top \tilde M v + \beta^2 w^\top \tilde M w) - (v^\top \tilde M v + \beta^2 w^\top \tilde M w)+ O(\eta \beta^2)\\
& = O(\eta\beta^2) \;.
\end{align*}
Thus,
$$
v^\top M_A(X^\Bg) w = O(\eta\beta) \;.
$$
Summing over a basis of $v\in V$ and $w\in V^\perp$, we get that
$$
\|M_A^{V,V^\perp}(X^\Bg)\|_F
= O(d\eta\beta) \leq (\gamma/2)\beta \;.
$$
This shows that Property 4 
holds.

Thus, we can apply Proposition \ref{case-2-prop}
and find that if we return a subspace,
it has the desired property;
and otherwise that setting $C=\textsc{ImproveTransform}(A,X,\eps)$, we have that
$$
\Phi_X(C) \leq \Phi_X(A) - \Omega(\eps^3/(d^7 n)) \;.
$$
Thus, in either case, if $\textsc{ImproveTransform}(A,X,\eps)$ returns a matrix,
the value of $\Phi_X(A)$ decreases by $\Omega(\eps^5/(d^{10} n^5))$.
Since $\zeta$ is less than half of this, each iteration of \textsc{ForsterTransform}'s
main while loop decreases $\Phi_X(A)$ by at least $\Omega(\eps^5/(d^{10} n^5))$.
Therefore, our algorithm terminates in at most polynomially many iterations.

\section{PAC Learning Halfspaces in Strongly Polynomial Time} \label{sec:ltfs}

In this section, we give our strongly polynomial improper PAC learner for halfspaces,
thereby establishing Theorem~\ref{thm:lft-real-intro}.

\subsection{Approximate Forster Decomposition} \label{ssec:forster-decomp}


Theorem~\ref{thm:forster-intro} is often difficult to use directly
as it does not always guarantee a Forster ransform.
This is necessary because if many points are concentrated on a subspace,
it may be the case that no such transform exists.
However, in this case we can at least find a dense subspace
and hopefully can find a Forster transform on that subspace.
In general, we have the following result:

\begin{proposition}[Forster Decomposition] \label{ForsterDecompositionProp}
There is an algorithm that given a multiset $X$ of $n$ points in $\R_{\ast}^d$
and $\eps > 0$, runs in time strongly-polynomial in $d \, n /\eps$,
and with high probability returns a subspace $V \subseteq \R^d$ with $V \neq \0$
and a linear transformation $A: V \rightarrow \R^{\dim(V)}$, such that
\begin{enumerate}

\item $|X \cap V| \geq (n/d) \, \dim(V)$.

\item The eigenvalues of $\frac{1}{|X \cap V|} \sum_{x \in X \cap V} f_A(x) (f_A(x))^\top$
are in $[(1-\eps)/\dim(V), (1+\eps)/\dim(V)].$
\end{enumerate}
\end{proposition}

\begin{proof}
The algorithm here is quite simple, presented in pseudocode below.

\begin{algorithm}[hbt!]
   \caption{Extended Forster Transform Algorithm}
\begin{algorithmic}[1]
   \Function{$\textsc{ForsterSubspace}$ }{set $X {\subset \R_{\ast}^d}$ of $n$ points, accuracy parameter $\eps$}
   \State Let $V = \R^d$.
   \State \label{return-step}  Let $d' = \dim(V)$ and let $L$ be a linear isomorphism
between $V$ and $\R^{d'}$ of bit complexity comparable to the bit complexity of $V$.
   \State Let $X ' := \{ L(x): x \in X  \cap(V) \}$.
   \State Run Algorithm \ref{alg:forster-full} on $X' \subseteq \R^{d'}.$
   \State If it returns a subspace $W$, set $V \leftarrow L^{-1} W$ and return to Step \ref{return-step}.
   \State Otherwise, if it returns a matrix $A$, return $(V, A L)$.
   \EndFunction
  \end{algorithmic}
\end{algorithm}

The essential guarantee of this algorithm is that $V$
is always a subspace of bounded bit complexity, 
such that $|X \cap V| \geq |X| \dim(V)/d$. This is clearly true initially.
If it was true for $V$, and our algorithm finds a subspace $W$,
it will also be true of $V' = L^{-1} W$. To see this, we note that
\begin{align*}
|V' \cap X | &= |\{x \in X \cap V : L(x) \in W \}| =
|X' \cap W | \geq |X'| \dim(W) / d' \\
&\geq |X| \dim(V) \dim(W) / (\dim (V) d) = |X | \dim(W)/ d \;.
\end{align*}
On the other hand, we note that $W$ is generated by points in $X'$,
and thus $V'$ is generated by points in $X$, which in turn implies the bounded complexity claim.
In particular, this allows us to define an $L$ with polynomial bit-complexity,
which (along with the observation that $\dim(V)$ shrinks by at least one each iteration)
makes the algorithm clearly strongly polynomial.

The correctness follows from the above proof that $|X \cap V| \geq |X| \dim(V)/d$,
and the fact that $A$ gives an $\eps$-approximate Forster transform of $X'$ on $W$.
This completes the proof of Proposition~\ref{ForsterDecompositionProp}.
\end{proof}


\subsection{PAC Learning Halfspaces} \label{ssec:ltfs-realizable}

Since we work in the distribution-independent setting,
will assume without loss of generality that the target halfspace
is homogeneous, i.e., has zero threshold. 
We can straightforwardly reduce the general case to the homogeneous case 
by increasing the dimension by $1$. In particular, 
if we associate point $x\in \R^d$ with $x'=(x,-1)\in \R^{d+1}_\ast$, 
then we note that $w\cdot x-t =(w,t)\cdot(x,-1)$, 
and thus a general halfspace over the $x$ vectors 
is equivalent to a homogeneous halfspace over the $x'$.

The basic idea of our PAC learning algorithm is that if we are given
a set of points in approximate radial isotropic position,
we can use a variant of the perceptron algorithm to efficiently
compute a hypothesis that correctly classifies a reasonable fraction of these points.
In particular, we will be using the following lemma,
a version of which appears in~\cite{BFK+:97, DunaganV04}:

\begin{lemma}\label{marginPerceptronLemma}
Let $S$ be a set of $n$ labeled examples $(x,y) \in \R^d \times \{\pm 1\}$
such that there exists an unknown vector $w \in \R_{\ast}^d$
with $y = \sgn(w\cdot x)$ for each $(x,y)\in S$, and let $\gamma > 0$ be a parameter.
There exists an algorithm that given $S$ and $\gamma$ has running time
strongly polynomial in $n \, d/\gamma$,
and returns a vector $v \in \R_{\ast}^d$ that
for all $(x,y)\in S$ with $|v\cdot x| \geq \gamma \|v\|_2 \|x\|_2$
satisfies $y = \sgn(v\cdot x)$.
\end{lemma}

\begin{proof}
We begin with the assumption that we know a vector $v$
such that $v \cdot w \geq 3 \|w\|_2/\gamma$ and $\|v\|_2^2 = O(d/\gamma^2)$.
The algorithm is the following:
\begin{enumerate}
\item While there exists an $(x,y) \in S$ with $|v\cdot x| \geq \gamma \|v\|_2 \|x\|_2$
and $y \neq \sgn(v\cdot x)$, do:

\begin{enumerate}
\item Let $\hat x$ be a positive multiple of $x$ with $\ell_2$-norm between $1$ and $2$.

\item $v \leftarrow v + y(\hat x)$.

\end{enumerate}

\item Return $v$.
\end{enumerate}
It is clear that the returned value of $v$ has the desired property
and that each operation can be performed with limited precision.
It remains to show that, under the given assumptions on $v$,
this algorithm will terminate in a polynomial number of iterations.

For this, we note that in each iteration if we let $v'$ be the new value of $v$, we have that
$$
\|v' \|_2^2 = \|v\|_2^2 + 2 y v \cdot (\hat x) + \|\hat x \|_2^2 \leq \|v\|_2^2 + 4 + 2 y (v \cdot x) \|\hat x\|_2/\|x\|_2 \;.
$$
Noting that $y$ and $(v\cdot x)$ have opposite signs, the RHS above is at most
$$
\|v\|_2^2 + 4 - 2 |v\cdot x| / \|x\|_2 \leq \|v\|_2^2 + 4 - 2 \|v\|_2 \gamma \;.
$$
Therefore, so long as $\|v\|_2 \geq 3/\gamma$, we have that $\|v'\|_2^2 \leq \|v\|_2^2 - 2$.

On the other hand, we have that
$$
v' \cdot w = v \cdot w + y( \hat x \cdot w ).
$$
Since $y$ has the same sign as $x\cdot w$, which has the same sign as $\hat x \cdot w$,
the above quantity is at least $v\cdot w$. This means that $v\cdot w$ only increases
over the course of our algorithm, and therefore throughout the algorithm
$\|v\|_2 \geq |v\cdot w|/\|w\|_2 \geq 3/\gamma$. Give the above,
this implies that $\|v\|_2^2$ must decrease by at least $2$ each iteration.
This cannot happen more than $\|v\|_2^2$ times, and therefore
the algorithm will terminate after at most $O(d/\gamma^2)$ iterations.

It remains to show how to efficiently find a $v$ with
$v\cdot w \geq 3 \|w\|_2/\gamma$ and $\|v\|_2^2 = O(d/\gamma^2)$.
We claim that it is always possible to take $v$ to be
an appropriately large constant multiple of $\sqrt{d}/\gamma$
times plus or minus a standard basis vector.
This is because some coordinate of $w$ must have absolute value at least $\|w\|_2/\sqrt{d}$.
Thus, we can run the above algorithm in parallel for each such initial value of $v$
and run until one of them returns an answer.
\end{proof}

Combining the modified perceptron algorithm 
of Lemma~\ref{marginPerceptronLemma}
with an approximate Forster transform,
gives us a way to learn a reasonable fraction
of the points for any linearly separable dataset.

\begin{lemma}\label{PartialClassifierLemma}
Let $S$ be a multiset of labeled examples $(x,y) \in \R^d_\ast \times \{\pm 1\}$
such that there exists an unknown vector $w \in \R_{\ast}^d$
with $y = \sgn(w\cdot x)$ for each $(x,y)\in S$.
There exists a strongly polynomial time algorithm
that with high probability returns a subspace $V$ of $\R^d$,
a linear transformation $A: V \rightarrow \R^{\dim(V)}$,
and a vector $v \in V$ such that for every $(x,y)\in S$
with $x\in V$ and $|v\cdot (Ax)| \geq \|v\|_2 \, \|Ax\|_2 / (2\sqrt{d})$
we have that $y = \sgn(v\cdot x)$.
Furthermore, this holds for at least a $1/(4d)$-fraction of points $(x,y)\in S$.
\end{lemma}

\begin{proof}
First, we apply the algorithm of Proposition~\ref{ForsterDecompositionProp}
to the multiset $X = \{x \in \R^d: (x, y) \in S\}$ with $\eps = 1/2$, to obtain $V$ and $A$.
We then let $S' := \{(Ax,y): (x,y) \in S, x \in V\}$.
We note that for all $(z,y)\in S'$ we have that
$y = \sgn(w \cdot x) = \sgn( ((A^\top)^{-1} w) \cdot z ).$
This means that we can apply the algorithm of Lemma \ref{marginPerceptronLemma} to $S'$,
which we do with $\gamma = 1/(2\sqrt{d})$ to obtain $v$.

By the statement of Lemma~\ref{marginPerceptronLemma},
we have that for each $(Ax,y) \in S'$ with $|v\cdot (Ax)| > \|v\|_2 \, \|Ax\|_2/(2\sqrt{d})$
that $y = \sgn(v\cdot (Ax))$, as desired.
It remains to show that this applies to a large fraction of points $(x,y) \in S$.

To establish this, we note that
$$
\frac{1}{|S|} \sum_{(x,y)\in S, x\in V} \frac{(Ax)\, (Ax)^\top}{\|Ax\|_2^2} =
\left( \frac{|S'|}{|S|} \right)\left( \frac{1}{|S'|} \sum_{(x,y)\in S, x\in V} \frac{(Ax) \, (Ax)^\top}{\|Ax\|_2^2} \right)
\succeq \left(\frac{\dim(V)}{d}\right)\left( \frac{I}{2\dim(V)} \right) \succeq \frac{I}{2d} \;.
$$
Therefore, we have that
$$
\frac{1}{|S|} \sum_{(x,y)\in S} \mathbf{1}\{x\in V\}\left( \frac{v\cdot (Ax)}{\|Ax\|_2 \, \|v\|_2} \right)^2
\geq v^\top \left( \frac{I}{2d} \right)v/\|v\|_2^2 \geq \frac{1}{2d} \;.
$$
This means that the average value over $(x,y)\in S$ of
$g(x):=\mathbf{1}\{x\in V\} \left( \frac{v \cdot (Ax)}{\|Ax\|_2 \|v\|_2} \right)^2$ is at least $1/(2d)$.
The contribution from terms with $g(x) < 1/(2\sqrt{d})$ is at most $1/(4d)$.
Since $g(x)\leq 1$ for all $x$, this implies that at least a $1/(4d)$-fraction of points
$(x,y)\in S$ have $g(x)\geq 1/(2\sqrt{d})$.
This completes the proof of Lemma~\ref{PartialClassifierLemma}.
\end{proof}

Ideally, we would like a version of Lemma~\ref{PartialClassifierLemma}
that works over a distribution rather than a finite set.
This can be achieved by running the algorithm of Lemma \ref{PartialClassifierLemma}
on a suitably large set of samples. To establish generalization guarantees,
we leverage the fact that the collection of possible classifiers
comes from a set of bounded VC-dimension.

\begin{proposition}\label{PartialLernerProp}
Let $\D$ be a distribution over $\R^d \times \{\pm 1\}$
such that for some unknown vector $w \in \R_{\ast}^d$ we have
that for $(x,y)\sim \D$ that $y=\sgn(w\cdot x)$ almost surely.
Given $\eps,\delta>0$ with $\eps < 1/(20 d)$,
there exists an algorithm that draws $n=O(d^2 \log(1/\delta)/\eps^2)$
i.i.d.\ samples from $\D$, runs in time strongly polynomial in $n, d$,
and with probability at least $1-\delta$ returns a vector subspace $V$ in $\R^d$,
a linear transformation $A: V \to \R^{\dim(V)}$ and a vector $v \in V$, such that:
\begin{enumerate}
\item The probability over $(x,y)\sim \D$ that $x\in V$,
$|v\cdot (Ax)| \geq \|v\|_2 \, \|Ax\|_2/(2\sqrt{d})$,
and $y \neq \sgn(v\cdot x)$ is at most $\eps$.

\item The probability over $(x,y)\sim \D$ that
$x\in V$ and $|v\cdot (Ax)| \geq \|v\|_2 \, \|Ax\|_2/(2\sqrt{d})$
is at least $1/(5d)$.
\end{enumerate}
\end{proposition}

\begin{proof}
We take a set $S$ of $n$ i.i.d.\ samples from $\D$
and apply the algorithm of Lemma~\ref{PartialClassifierLemma} to them
for $n$ a sufficiently large constant multiple of $(d^2 \log(1/\delta)/\eps^2)$.
This is clearly a strongly polynomial time algorithm. It remains to prove correctness.

We note that the probability over $(x,y)$ drawn uniformly from $S$ that
$x\in V$, $|v\cdot (Ax)| \geq \|v\|_2 \, \|Ax\|_2/(2\sqrt{d})$, and
$y \neq \sgn(v\cdot x)$ is $0$.
Furthermore, the probability over $(x,y)$ drawn uniformly from $S$
that $x\in V$ and $|v\cdot (Ax)| \geq \|v\|_2 \, \|Ax\|_2/(2\sqrt{d})$ is at least $1/(4d)$.
It suffices to show that (with high probability over our samples)
these probabilities over $S$ are within $\eps$
of the corresponding probabilities if $(x,y)$ were drawn from $\D$.

In fact, we will show that with probability $1-\delta$ over our choice of samples
the following holds: for any choice of $V, A$, and $v$,
the probabilities over $S$ and $\D$ of these events differ by at most $\eps$.
This will follow from the VC-Inequality~\cite{DL:01},
if we can show that these events come from classes of VC-dimension $O(d^2)$.
We now proceed with the argument.
We note that these events depend only on the following simpler events:
\begin{itemize}
\item Whether $y=1$.
\item Whether $x\in V$.
\item Whether $(v\cdot x) > 0$.
\item Whether $|v\cdot x|^2 \geq \|v\|_2^2 \, \|Ax\|_2^2 / (4d)$.
\end{itemize}
The first of these is a specific event, so has VC-dimension $0$.
The second event checks membership in a subspace, which has VC-dimension $d$.
The third event checks membership in a halfspace, hence also has VC-dimension $d$.
The last of these events is a degree-$2$ threshold condition,
which has VC dimension $O(d^2)$.
Since the events we care about are logical combinations
of finitely many events of VC-dimension $O(d^2)$,
they come from classes with VC-dimension $O(d^2)$.
This completes our proof.
\end{proof}

We are now ready to prove the main result of this section.

\begin{theorem}\label{thm:lft-real-full}
Let $\D$ be a distribution over $\R^d \times \{\pm 1\}$
such that for some unknown vector $w \in \R_{\ast}^d$
we have that for $(x,y)\sim \D$ that $y=\sgn(w\cdot x)$ almost surely.
Given $\eps,\delta > 0$ with $\eps < 1/(20d)$ there is an algorithm
that draws $n = O(d^{9/2}\log(1/\eps) \log(d/\eps\delta)/\eps^2)$ i.i.d.\ samples from $\D$,
runs in strongly polynomial time, and returns a strongly polynomial time computable function
$f:\R^d \to \{\pm 1\}$ such that with probability $1-\delta$ over the samples it holds that
$\pr_{(x,y)\sim \D}[f(x) \neq y] \leq \eps$.
\end{theorem}

\begin{proof}
For simplicity, we allow our algorithm to output a function $f$ valued in $\{0,1,-1\}$.
The algorithm is as follows:
\begin{algorithm}[hbt!]
   \caption{Halfspace Learning Algorithm}
\begin{algorithmic}[1]
  \Function{$\textsc{LearnLTF}$ }{sample access to distribution $\mathcal{D}$ over $\R^d_\ast \times \{\pm 1\}$, accuracy parameter $\eps$}
  \State Let $f_0 \equiv 0$.
  \State Let $C>0$ be a sufficiently large universal constant.
  \State Let $r = C\sqrt{d} \log(1/\eps) \in\Z_+$.
  \For{For $i=1$ to $r$}
    \State Take $M:=Cd^{4} \log(d/\eps\delta)/\eps^2$ samples from $\D$ and call the resulting multiset $T$.
    \State Let $S$ be the set of $(x,y)\in T$ such that $f_{i-1}(x) = 0$.
    \If{$|S| < \eps M / 4$}
      \State \Return $f_{i-1}$
    \Else
      \State Run the algorithm from Proposition \ref{PartialLernerProp}
with parameters $\eps \leftarrow \eps/(10d)$ and $\delta \leftarrow \delta/(2r)$
to obtain $V,A,v$, using $S$ as the set of samples.
       \State Let
$$ f_i(x) :=
\begin{cases}
f_{i-1}(x) & \textrm{, if }f_{i-1}(x)\neq 0\\
\sgn(v\cdot (Ax)) & \textrm{, if }f_{i-1}=0, x\in V, \textrm{and }|v\cdot (Ax)|/(2\sqrt{d})\\
0 & \textrm{, otherwise}.
\end{cases}
$$
    \EndIf
  \EndFor
  \EndFunction
  \end{algorithmic}
\end{algorithm}

It is easy to see that the sample complexity and runtime are as desired. 

For correctness, we note that $S$ is a set of i.i.d.\ samples from
the distribution of $\D$ conditioned on $f_{i-1}(x)=0$.
By the conclusion of Proposition \ref{PartialLernerProp}, this means
that, with probability at least $1-\delta/(2r)$ over the samples,
the probability that $f_i(x) = 0$ is at most $(1-1/(2\sqrt{d}))$ times the probability
that $f_{i-1}(x) = 0$. Furthermore, with probability at least $(1-\delta/(2r))$,
we have that $|S| < 2 \, M \,  \pr[f_{i-1}(x) = 0]$. Combining the above,
we see that with probability $1-\delta$,
when our algorithm returns an $f$,
it is the case that $\pr_{(x,y)\sim \D}[f(x)=0] \leq \eps/2$.

Furthermore, if all the calls to the algorithm from
Proposition \ref{PartialLernerProp} succeed,
then the probability over $(x,y)\sim \D$ conditioned on $f_{i-1}(x)=0$
that $f_i(x) = y$ is at least $1/(5d)$,
while the probability that $f_i(x)=-y$ is at most $\eps/(10 d)$.

Using this, we can show by induction on $i$ that
$$
\pr_{(x,y)\sim \D}[f_i(x) = -y] < (\eps/2) \, \pr_{(x,y)\sim \D}[f_i(x) = y] \;.
$$
This combined with the result that $\pr_{(x,y)\sim \D}[f_i(x)=0] < \eps/2$
for the returned $f_i$, gives our final result.
\end{proof}

\section{Conclusions and Open Problems} \label{sec:conc}
In this work, we designed the first strongly polynomial time algorithm for computing
$\eps$-approximate Forster transforms of a given dataset\footnote{While our Forster algorithm is randomized, 
we remark that the only source of randomness is due to the method we use 
to compute an approximate eigendecomposition. 
It is plausible that deterministic algorithms exist
for this purpose, in which case our Forster algorithm becomes deterministic as well.}. 
By using this algorithm is an essential ingredient, we gave the first strongly polynomial time 
algorithm for distribution-free PAC learning of halfspaces, 
both in the realizable setting and in the presence of semi-random label noise. 
This algorithmic result is surprising (even in the realizable case), 
as obtaining a strongly polynomial {\em proper} PAC learner 
is {\em equivalent} to strongly polynomial LP --- a major unsolved problem in TCS.


A number of open problems suggest themselves:
\begin{itemize}[leftmargin=*]
\item Our $\eps$-approximate Forster transform algorithm has runtime scaling 
polynomially with $1/\eps$. That is, our algorithm runs in strongly polynomial 
time when $\eps$ is at least inverse polynomial in $n, d$. An obvious open question
is to develop a strongly polynomial algorithm with a $\polylog(1/\eps)$ runtime dependence.
To achieve such a guarantee with our approach, one needs to circumvent two obstacles:
First, one would need to reduce the number of iterations of our algorithm 
(that is controlled by the progress in our potential function). 
Second, one would require a strongly polynomial 
approximate eigendecomposition subroutine with a $\polylog(1/\eps)$ runtime dependence.

\item We believe that the following question is of independent interest:
Is there a strongly polynomial time algorithm for approximate eigendecomposition with 
a $\polylog(1/\eps)$ runtime dependence? 
Moreover, is there a deterministic algorithm?

\item The running time of our algorithm is strongly polynomial in $n, d$, 
but the polynomial dependence is quite large (of the order of $(nd)^{10}$). 
While we did not make any effort to optimize the degree of the polynomials, 
it would be interesting to understand the quantitative limitations of our approach. 
Can our approach lead to algorithms with good practical performance?

\item As mentioned in the introduction of this paper, Forster's rescaling can be viewed as 
a very special cases of operator scaling and tensor scaling~\cite{GargO18}.
These tasks have attracted significant attention in recent years from various communities,
and efficient (weakly polynomial) algorithms (in some cases with a $\poly(1/\eps)$ dependence) 
have been developed, see, e.g.,~\cite{AGLOW18, BFGOWW18} and references therein.
It would be interesting to explore whether our approach can be extended to yield strongly
polynomial algorithms (when $\eps$ is not too small) for such generalizations.

\end{itemize}

\bibliographystyle{alpha}
\bibliography{allrefs}

\newpage

\appendix

\section*{APPENDIX} \label{sec:app}

\section{Proof of Fact~\ref{fact:transform-properties}} \label{app:transform-properties}

We show each property separately.
\begin{enumerate}[(a)]
\item This follows directly from the fact that the transformation 
$f_A(x) =  \frac {A x} {\|A x\|_2}$ is scale-invariant.

\item We have that 
$f_{B} ( f_{A} (x) ) =  f_{B} \left( \frac {A x} {\|A x\|_2} \right) = 
f_{B} ( A x ) = \frac {B A x} {\|B A x\|_2} = f_{B A} (x)$, 
where the second equality follows from part~\ref{fact:scale-inv}.

\item Let $\alpha = \|B - I\|_2$ and $y = f_A(x)$. Note that $\|y\|_2 = 1$.
By property~\ref{fact:comp}, we equivalently want to show that 
$\| f_B(y) - y\|_2 \le \alpha$. We have that 
$$\| f_B(y) - y\|_2 = \left\|  \frac B {\|B y\|_2} y - y \right\|_2 
\leq  \max_{v \in \R^d: \|v\|_2 = 1} \left\|  \frac B {\|B y\|_2} v - v \right\|_2  
= \left\| \frac B {\|B y\|_2} - I \right\|_2 \;.$$
The desired statement follows from the fact that the matrix $\frac B {\|B y\|_2}$ 
has eigenvalues between $\frac 1 {1+\alpha}$ and ${1+\alpha}$, 
as $B$ has eigenvalues between $1$ and $1+\alpha$ and 
$\|B y\|_2 \in [1,1+\alpha]$.

\item We have that 
$$f^{(V)}_{B A}(x) = \proj_V f_{B}(f_A(x)) 
= \proj_V \frac {(I+ a I_V) f_{A}(x)} {\|B f_{A}(x)\|_2} = 
\frac {(1+ a) f^{(V)}_{A}(x)} {\|B f_{A}(x)\|_2} \;.$$
Since $ 1\leq \|B f_{A}(x)\|_2 \leq 1+a$, 
it follows that 
$1 \leq \lambda(x) \eqdef \frac {(1+ a)} {\|B f_{A}(x)\|_2}  \leq 1+a$, 
as desired.
Similarly,  we can write 
$$f^{(V^\perp)}_{BA}(x) = \frac {f^{(V^\perp)}_{A}(x)} {\|B f_{A}(x)\|_2} \;.$$
Since $ 1\leq \|B f_{A}(x)\|_2 \leq 1+a$, 
it follows that $\frac{1}{1+a} \leq \mu(x) \eqdef \frac{1}{\|B f_{A}(x)\|_2} \leq 1$.
\end{enumerate}
This completes the proof of Fact~\ref{fact:transform-properties}.

\end{document}